\documentclass[12pt,a4paper]{article}
\pdfoutput=1
\usepackage{jheppub}

\usepackage{amstext,amsfonts,amsthm,mathtools}

\numberwithin{equation}{section}
\usepackage{tikz}

\newcommand{\tPi}{\widetilde{\Pi}}
\newcommand\listdots{,\!\makebox[1em][c]{.\hfil.\hfil.},}

\newcommand{\one}{\ensuremath{\mathbf{1}}}

\newcommand{\bC}{\mathbb{C}}

\newcommand{\id}{\mathrm{id}}

\DeclareMathOperator{\Hom}{Hom}
\DeclareMathOperator{\Fun}{Fun}
\DeclareMathOperator{\fmod}{-mod_f}
\DeclareMathOperator{\MF}{MF}
\DeclareMathOperator{\mf}{mf}
\DeclareMathOperator{\FF}{FF}
\newcommand{\HMF}{H^{\bullet}\mathrm{MF}}
\newcommand{\Hmf}{H^{\bullet}\mathrm{mf}}
\newcommand{\HFF}{H^{\bullet}\mathrm{FF}}
\newtheorem{prop}{Proposition}[section]
\newtheorem{corr}[prop]{Corollary}
\newtheorem{lem}[prop]{Lemma}
\newtheorem{Def}[prop]{Definition}
\newtheorem{ex}[prop]{Example}

\newcounter{dummy}

\makeatletter
\def\@fpheader{\relax}
\makeatother

\title{Modelling interface factorizations between Landau--Ginzburg models as module functors}

\author{Stefan Fredenhagen}
 
\affiliation{Mathematical Physics, Faculty of Physics, University of Vienna,\\ Boltzmanngasse 5, 1090 Vienna, Austria}
\affiliation{Erwin Schr{\"o}dinger International Institute for Mathematics and Physics,\\ University of Vienna, Boltzmanngasse 9, 1090 Vienna, Austria}

\emailAdd{stefan.fredenhagen@univie.ac.at}

\abstract{We study the fusion of B-type interfaces between $N= (2,2)$ supersymmetric Landau--Ginzburg models. Such interfaces can be described by matrix factorizations of the difference of the superpotentials, and their fusion is modelled by the tensor product of the factorizations. The effect of fusing a fixed interface gives rise to a functor on the category of matrix factorizations. For at least some interfaces, this can be lifted to a functor on the category of modules over polynomial rings. These fusion functors provide an alternative way of modelling a class of interfaces between Landau--Ginzburg models, which characterizes interfaces through their fusion properties. Interface fields correspond to morphisms between fusion functors that can be defined via a Hochschild-type cohomology. This leads to a strict monoidal supercategory of fusion functors, where horizontal composition is given by composition of functors. The category of fusion functors can be related to the category of matrix factorizations by a functor that maps a fusion functor to the corresponding interface factorization. It is not faithful, but we prove that it is full for polynomial rings in one variable. The description of interfaces in terms of fusion functors has the advantage that fusion of interfaces becomes simple. It provides a computational tool to evaluate tensor products of matrix factorizations, and could be applied for example in Kazama--Suzuki models to analyze fusion categories of rational defects, or in the context of Khovanov--Rozansky link homology.}

\preprint{UWThPh 2022-21} 

\begin{document}
\maketitle

\newpage
\section{Introduction}

The study of defects and interfaces constitutes a rich and important topic in two-dimensional field theories. Topological interfaces can be deformed and fused, which gives rise to interesting algebraic structures that generalize the notion of ordinary global symmetries \cite{Petkova:2000ip,Frohlich:2004ef,Bachas:2008jd,Frohlich:2009gb,Carqueville:2012dk,Brunner:2013ota,Bhardwaj:2017xup}. Topological field theories provide a well-defined framework in which defects can be analyzed (see, e.g., \cite{Kapustin:2010ta,Davydov:2011kb}).
In B-type topologically twisted 
$N=(2,2)$ Landau--Ginzburg (LG) models~\cite{Vafa:1990mu,Witten:1991zz}, these structures admit an explicit and computable description in terms of matrix factorizations of differences of superpotentials \cite{Brunner:2007qu}, and interface fields appear as morphisms in the corresponding categories.

While this matrix-factorization framework makes fusion concrete, performing fusion directly at the level of factorizations can still be technically involved. In many examples, however, the effect of fusing a fixed interface can be described more economically by a functorial action on categories of modules: fusion becomes composition of functors. The purpose of this paper is to develop this functorial viewpoint beyond the level of isolated examples. We study fusion functors as objects with intrinsic structure -- admitting morphisms, cone-like constructions, and graded refinements -- and explain how this provides an alternative description of a distinguished class of B-type interfaces in Landau–Ginzburg models.

\subsection{Interfaces and fusion in Landau--Ginzburg models}

Two-dimensional $N=(2,2)$ Landau--Ginzburg models provide a convenient UV description of a class of $N=(2,2)$ superconformal field theories \cite{Vafa:1988uu}, and their B-twisted topological sector can be described algebraically \cite{Eguchi:1990vz}. They can be described by an action involving chiral superfields $\phi_i$ $(i=1,\dots,n)$ on superspace,
\begin{equation}
    S = \int d^2 z\,d^2 \theta\,d^2\bar{\theta}\, \sum_i\bar{\phi}_i\phi_i + \int d^2z\,d^2\theta\,W(\phi_1,...,\phi_n) + \text{c.c.}\, ,
\end{equation}
where the superpotential $W$ is a holomorphic function. In the context of studying superconformal field theories in the infrared, we are typically interested in superpotentials that are (quasi-)homogeneous polynomials. For instance, the LG model with the monomial superpotential $W=x^{k}$ (here $x$ is a placeholder for a chiral superfield) flows to the simplest minimal model of type $A_{k-1}$ related to the supersymmetric coset model based on $SU(2)/U(1)$ \cite{Vafa:1988uu}. Due to nonrenormalization theorems, the F-type term involving the superpotential is not renormalized, and its structure survives the renormalization group flow to the infrared. Topologically twisted LG models provide a framework to study structures that do not change under renormalization group flow like for instance the chiral ring.

In this setting, also boundary conditions and interfaces can be investigated in the topologically twisted theories, which capture properties of conformal boundary conditions and interfaces in the superconformal theory in the infrared. We are considering here B-type boundary conditions or interfaces that are consistent with the B-type topological twist. On a worldsheet $\Sigma$ with a boundary, a B-type SUSY variation of the F-term in the action produces a boundary term \cite{Warner:1995ay}. Schematically, it reads
\begin{equation}
    \delta_\epsilon^B \int_\Sigma d^2z\,d^2\theta\, W \sim \bar{\epsilon} \int_{\partial \Sigma} dx\,d\theta\, W\, ,
\end{equation}
and it can be compensated by introducing boundary fermions \cite{Warner:1995ay,Brunner:2003dc}. For a single superfield $\phi$, add a term of the form $\int dx\,d\theta\, J(\phi)\Pi$ to the action, where $\bar{D}\Pi = E(\phi)$. This compensates the boundary term if $JE=W$ is a factorization of $W$. In more general situations, a boundary condition can be described by a matrix factorization $JE=W\cdot \one$ of $W$ in terms of matrices $J$ and $E$, or -- equivalently --  by an odd matrix $Q$ with respect to a grading $\sigma = (-1)^F= \begin{pmatrix}
    1 & 0\\ 0 & -1\end{pmatrix}$ of the form $Q=\begin{pmatrix}
        0 & J \\ E & 0
    \end{pmatrix}$ that satisfies 
\begin{equation}
    Q^2 = W\cdot \one\,.
\end{equation}
$Q$ has the interpretation of the boundary part of the BRST operator~\cite{Kapustin:2003ga}.
In this way, the study of boundary conditions has turned into the algebraic task of studying matrix factorizations. 

In a generalization, B-type interfaces between two LG models can be described by matrix factorizations of the difference of the superpotentials~\cite{Brunner:2007qu}. Such matrix factorizations form a category whose morphisms can be interpreted as interface fields. One should think of a physical interface as being described by an equivalence class of isomorphic matrix factorizations. 

For instance, consider the situation of a defect in a minimal model, so we have the two superpotentials $W=x^k$ and $W'=x'^k$. The odd matrix
\begin{equation}
Q_\xi = \begin{pmatrix}
    0 & x'-\xi x \\ \frac{x'^k-x^k}{x'-\xi x} & 0
\end{pmatrix}    \,,
\end{equation}
where $\xi^k=1$, satisfies $Q_\xi^2 = (W'-W)\cdot \one$. It describes a family of defects that incorporate the $\mathbb{Z}_k$ symmetry of the minimal model. In particular, for $\xi=1$ this corresponds to the identity defect (whose factorization we denote by $I_W$) which has no effect at all.

When we consider an interface between LG models with superpotentials $W'$ and $W$, specified by a factorization $Q_1$, and another interface between models with superpotentials $W$ and $W''$ specified by $Q_2$, the fusion of the two interfaces results in an interface between models with potentials $W'$ and $W''$ that is described by the tensor product factorization
\begin{equation}
    Q_1\boxtimes Q_2 = Q_1\otimes\one + \sigma_1 \otimes Q_2
    \,,
\end{equation}
where $\sigma_1$ denotes the grading matrix related to the first factorization (given by $(-1)^F$ in terms of the fermion number). This fusion is illustrated in Figure~\ref{fig:simplefusion}. Denoting the variables of the models involved as $x_i',x_i,x_i''$, respectively, the factorization $Q_1\boxtimes Q_2$ still depends on the variables $x_i$ of the intermediate LG model. It can be shown \cite{Khovanov:2004,Brunner:2007qu} that such a factorization is always isomorphic to a finite size factorization that does not contain the auxiliary variables $x_i$, but to construct such a representative can be cumbersome in practice.

\begin{figure}
\begin{center}
\scalebox{.9}{\includegraphics{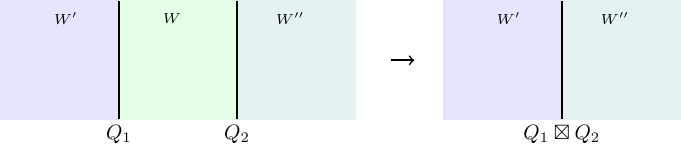}}
\end{center}
\caption{\label{fig:simplefusion}Fusion of interfaces. An interface between the Landau--Ginzburg models with
potentials $W'$ and $W$ is described by a matrix factorization $Q_1$ of $W'-W$, and an
interface between $W$ and $W''$ by a matrix factorization $Q_2$ of $W-W''$. Fusing the
two interfaces yields an interface between $W'$ and $W''$ described (before eliminating
intermediate variables) by the fusion product $Q_1 \boxtimes Q_2$.}
\end{figure}

In contrast to this, the effect of fusing an interface to another can in some cases be very simple: the extreme example is the identity defect which acts trivially under fusion (see Figure~\ref{fig:identityfusion}). Other examples are the symmetry defects $Q_\xi$ in the minimal model: their effect on another factorization just results in multiplying the variable $x$ by $\xi$. The class of interfaces which has a simple effect under fusion (in a sense made precise below) is much richer than this simple example as we will discuss in the following subsection.

\begin{figure}
\begin{center}
\scalebox{.9}{\includegraphics{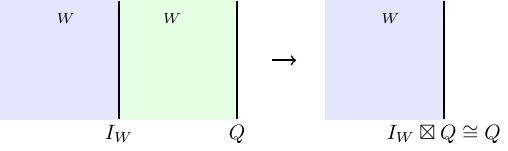}}
\end{center}
\caption{\label{fig:identityfusion}Unit for fusion. The identity defect $I_W$ (an interface from $W$ to itself)
acts trivially under fusion: fusing $I_W$ onto a boundary factorization (or, more
generally, onto an interface factorization) $Q$ yields an isomorphic factorization
$I_W \boxtimes Q \cong Q$.}
\end{figure}

\subsection{From fusion to functors}

The above examples fall in a situation where the effect of fusing an interface (between superpotentials $W'$ and $W$) to a boundary factorization $Q$ of the model with potential $W$ (this can be generalized to interface factorizations to a third LG model) can be described by a functor acting on the category of modules of the underlying rings. In particular, denote $R=\mathbb{C}[x_1,\dots,x_n]$ and $R'=\mathbb{C}[x'_1,\dots,x'_{n'}]$, then we consider a $\mathbb{C}$-linear functor from the category of (free finite-rank) $R$-modules to the category of (free finite-rank) $R'$-modules,
\begin{equation}
    U: R\fmod \longrightarrow R'\fmod\ .
\end{equation}
A matrix factorization $Q$ of $W$ can be viewed as a module endomorphism of an $R$-module $M\cong R^d$ where $d$ is the size of the matrix. The point is now to demand that $U(Q)$, which is an endomorphism of the $R'$-module $M'$, gives a factorization of $W'$. Computing $U(Q)^2$, we find, using the functorial property of $U$,
\begin{align}
    U(Q)U(Q) = U(Q^2) = U (W\cdot \one_M) \stackrel{!}{=} W'\cdot U(\one_M)= W'\cdot \one_{U(M)}\,.
\end{align}
We can read off that indeed $U(Q)$ is a matrix factorization of $W'$ provided that $U(W\cdot\one _M)=W'\cdot U(\one_M)$. Requiring this for all modules $M$, we can equivalently demand that
\begin{equation}\label{fusfunctorprop}
U (W\, \varphi) = W'\,U (\varphi)
\end{equation}
for any homomorphism $\varphi$ between (free finite-rank) $R$-modules. This condition can be read as compatibility with the F-term. As in \cite{Behr:2020gqw}, such functors will be called \textsl{fusion functors} in this work. As we will see shortly, any such functor labels an interface, which is now characterized solely by its effect via fusion (illustrated in Figure~\ref{fig:introducingfusionfunctors}).

\begin{figure}
\begin{center}
\scalebox{.9}{\includegraphics{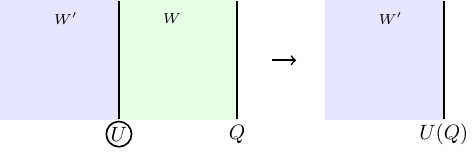}}
\end{center}
\caption{\label{fig:introducingfusionfunctors}Interfaces labelled by fusion functors. A fusion functor $U$ (circled label)
represents an interface from $W$ to $W'$ through its action on B-type factorizations: fusing the
corresponding interface onto a boundary (or interface) factorization $Q$ is implemented
by applying the functor, producing the factorization $U(Q)$. We use circled symbols to
distinguish functor-labelled interfaces from interfaces specified directly by
factorizations.}
\end{figure}

The effect of fusing a symmetry defect $Q_\xi$ has such a description. One can construct a functor $G_\xi$ which maps any free $\mathbb{C}[x]$-module to itself, and on the homomorphisms represented by polynomial matrices it acts by replacing $x\mapsto \xi  x$. Any appearance of $W=x^k$ is mapped to $\xi^k x^k=x^k$ if $\xi$ is a $k$th root of unity, hence it satisfies the requirement of a fusion functor. The functor precisely describes the effect of fusing an elementary symmetry defect $Q_\xi$ in a minimal model~\cite{Brunner:2007qu}. 

It is in general difficult to say whether a specific interface factorization has a description in terms of a fusion functor. Conversely, any fusion functor naturally produces an interface factorization by acting on the identity defect (see Figure~\ref{fig:fusionfunctorsandinterfaces}). This can be done after generalizing the application of the functor to $R\otimes R''$-modules (essentially, it only acts on the first factor). One can then apply the functor to the identity factorization $I_W$ seen as an endomorphism of an $R\otimes R$-module, and obtain an endomorphism $U(I_W)$ of an $R'\otimes R$-module which is a $(W',W)$-interface factorization. We refer to these as \textsl{operator-like interfaces}, since their action on boundary and interface factorizations admits a simple description in terms of the functor $U$. Fusing $U(I_W)$ to other factorizations has the same effect (up to isomorphisms) as applying $U$, see Proposition~\ref{prop:FFdescribeFusion}.

\begin{figure}
\begin{center}
\scalebox{.9}{\includegraphics{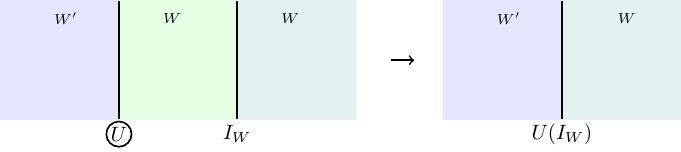}}
\end{center}
\caption{\label{fig:fusionfunctorsandinterfaces}From functors to interface factorizations. Any fusion functor $U$ naturally
produces an associated interface factorization by acting on the identity defect:
the factorization $U(I_W)$ describes the operator-like interface labelled by $U$.}
\end{figure}

There are many examples of interfaces for which a description as a functor on ring modules is explicitly known~\cite{Behr:2012xg,Behr:2020gqw} which go beyond the simple examples of symmetry defects discussed above. The advantage of this realization is a simple description of fusion: it is just given by composition of the functors (see Figure~\ref{fig:fusingfusionfunctors}). This is the main motivation to study interfaces in this approach. It has proven useful for explicit computations: it lies at the heart of the proposal for a complete identification of matrix factorizations that describe rational topological B-type defects and B-type boundary conditions in the $SU (3)/U (2)$ Kazama--Suzuki model~\cite{Behr:2010ug,Behr:2014bta}. Also, the simplicity of fusion in this approach allows one to study the algebraic structure of defect fusion and to investigate the fusion category of rational defects in specific models that correspond to rational conformal field theories.

\begin{figure}
\begin{center}
\scalebox{.9}{\includegraphics{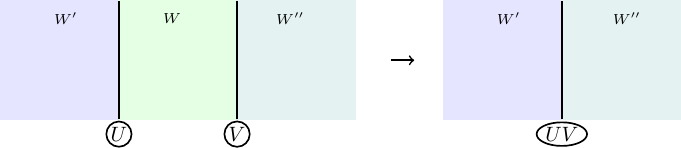}}
\end{center}
\caption{\label{fig:fusingfusionfunctors}Fusion becomes composition. Fusing two interfaces labelled by fusion functors
$U$ and $V$ corresponds to composing the functors: the fused interface is labelled by
the composite functor $UV$.}
\end{figure}

\subsection{Category of fusion functors: results and scope}

Before listing the main results, let us stress how to read the functorial structures physically. A fusion functor $U:R\fmod\to R'\fmod$ encodes a special class of B-type interfaces for which fusion acts functorially on module data. Such an interface can be represented, in the matrix-factorization framework, by the associated operator-like interface factorization $U(I_W)$ obtained by acting on the identity defect. The functorial operations introduced below are designed so that they match the standard interface operations (fusion and operator products of interface fields) after passing to the associated matrix factorizations. It is useful to distinguish two types of statements: (i) identities internal to the operator-like sector (fusion of operator-like interfaces, and the corresponding operations on their fields), and (ii) the induced action of an operator-like interface on general boundary/interface conditions, where fusion with $U(I_W)$ can be simplified to the functorial action $Q\mapsto U(Q)$.

\begin{itemize}
\item For a pair of superpotentials $W\in R$ and $W'\in R'$, we define a category $\HFF_{W',W}$ of fusion functors $U:R\fmod\to R'\fmod$ (Section~\ref{sec:subsec-fusionfunctors}). We view these as functorial descriptions of \textsl{operator-like interfaces}.

\item The morphism spaces in $\HFF_{W',W}$ are $\mathbb{Z}$-graded. Degree-$0$ morphisms are natural transformations, while higher-degree morphisms form a Hochschild-type complex (Section~\ref{sec:subsec-fusionfunctors}). We interpret these morphisms as functorial representatives for a class of interface fields and their higher operations.\footnote{Sometimes this is called Hochschild--Mitchell cohomology~\cite{Herscovich:2005} referring to a construction by Mitchell in~\cite{Mitchell:1972}.}

\item We define two compositions on morphisms: vertical composition (corresponding to operator products of interface fields) and horizontal composition (corresponding to compatibility of fields with fusion of interfaces), see Section~\ref{sec:horizontal}.

\item For defects (same superpotentials on both sides, $W'=W$), horizontal composition equips $\HFF_{W,W}$ with a monoidal structure describing fusion of operator-like defects. Together with the $\mathbb{Z}_2$-grading by parity of the degree, this makes $\HFF_{W,W}$ a monoidal supercategory (Section~\ref{sec:subsec-fusionfunctors}). There is no shift functor relating even and odd morphisms.

\item We introduce graded fusion functors, in direct analogy with graded matrix factorizations (Section~\ref{sec:gradedMF}).

\item Degree-$1$ morphisms give rise to an analogue of ``tachyon condensation'' \cite{Govindarajan:2005im}: we define a cone construction for fusion functors (Section~\ref{sec:Cones}) and relate it to the usual cone construction for matrix factorizations (see~\eqref{compatibilityConeandPi}). 

\item A key construction is the functor
\begin{equation}
    \Pi^{I_W}:\HFF_{W',W}\longrightarrow \Hmf_{W',W},
\end{equation}
obtained by evaluating a fusion functor $U$ on the identity factorization $I_W$ (Definition~\ref{def:PiIW}). Under $\Pi^{I_W}$, morphisms in $\HFF_{W',W}$ induce morphisms of interface matrix factorizations; this correspondence is not bijective in general. In particular, $\Pi^{I_W}$ is neither injective on objects nor on morphisms, and it is not surjective on objects. We show, however, that $\Pi^{I_W}$ is full in the one-variable case (polynomial rings in one variable), i.e.\ surjective on morphisms (Section~\ref{sec:onevariable}).

\item The functor $\Pi^{I_W}$ is compatible with fusion within the operator-like sector: the interface factorization associated to a composite functor agrees, up to natural isomorphism, with the fusion of the associated operator-like interfaces, see Proposition~\ref{prop:PiIWCompatibleWithFusion}.

\item The functor $\Pi^{I_W}$ is also compatible with the action of operator-like interfaces on general boundary/interface conditions: fusing $U(I_W)$ to a matrix factorization $Q$ is naturally isomorphic to applying $U$ to $Q$, i.e.\ $U(I_W)\boxtimes Q \cong U(Q)$. This is made precise by a commutative diagram of functors in Proposition~\ref{prop:FFdescribeFusion}.
\end{itemize}
The last two items provide the bridge to the standard interface description in terms of matrix factorizations. In particular, the diagrams at the end of Section~\ref{sec:ffdescribefusion} (page~\pageref{listoffigures}) summarize how an operator-like interface (with or without an insertion from the functorial morphism spaces) acts on a boundary/interface factorization (with or without a boundary/interface field) by fusion.

The non-surjectivity of $\Pi^{I_W}$ on objects reflects that not every B-type interface admits a fusion-functor description, and we do not currently have an intrinsic criterion. Empirically, in minimal models the topological defects are described by fusion functors, while for example boundary-boundary type defects are not. We treat this as an observation and leave a conceptual characterization for future work.
\smallskip

This article is structured as follows: Before introducing fusion functors, we start in Section~\ref{sec:categoryoffunctors} by defining a category of functors between categories of modules where we not only consider natural transformations, but also a tower of higher order morphisms. A subcategory thereof is identified as the category of fusion functors in Section~\ref{sec:fusionfunctors}. Here we also define the  functor $\Pi^{I_W}$ to the category of interface matrix factorizations whose properties are investigated in Section~\ref{sec:relationinterfacesandfusionfunctors}. Some more technical arguments are presented in the appendix.

\section{A category of functors}\label{sec:categoryoffunctors}
To prepare for the definition of the category of fusion functors in Section~\ref{sec:fusionfunctors}, we introduce in this section a larger category that contains the fusion functor categories, which depend on the superpotentials, as subcategories. Essentially, we consider ordinary $\mathbb{C}$-linear functors on categories of modules without the fusion functor property, but we extend the space of morphisms beyond the natural transformations. More concretely, we introduce the cohomology category $H^{\bullet}\Fun(R',R)$ of a differential graded category of functors mapping modules over a polynomial ring $R$ to modules over a polynomial ring $R'$. We also define a horizontal composition in $H^{\bullet}\Fun(R',R)$ which equips the category of endofunctors $H^{\bullet}\Fun(R,R)$ with the structure of a monoidal supercategory. Although the structures introduced in this section do not seem to have a direct physical interpretation, they provide the stage for the fusion functor description of B-type interfaces.

\subsection{A differential graded category of linear functors}
Let $R$ be a polynomial ring over $\mathbb{C}$, and
$R\fmod$ the category of free finite-rank $R$-modules with $R$-module homomorphisms $f:M\to N$ as morphisms. Let
$R'$ be another polynomial ring. We consider functors $U$ from
$R\fmod$ to $R'\fmod$ whose action on
homomorphisms is $\mathbb{C}$-linear: every $R$-module homomorphism $f:M\to N$ is mapped to an $R'$-module homomorphism $U(f):U(M)\to U(N)$ such that for every $\lambda\in\mathbb{C}$
\begin{equation}
    U(\lambda\cdot f)=\lambda\cdot U(f)\ .
\end{equation}
\begin{ex}\label{ex:invertiblefunctors}
Simple examples of such functors are given by the extension of scalars. Let $\omega :R\to R'$ be a ring homomorphism. Then $\omega^{*}$ maps an $R$-module $M$ to an $R'$-module by $\omega^{*} (M)= R'\otimes_{R}M$. It acts on homomorphisms in an obvious way -- in particular, a homomorphism from $R^{m}$ to $R^{n}$ described by a matrix with entries in $R$ is mapped to a homomorphism from $R'^{m}$ to $R'^{n}$ described by the same matrix where $\omega$ has acted on the entries.

For a concrete example, take $R'=R=\mathbb{C}[x]$ and $\omega_{\xi} (p (x))=p (\xi x)$ which is an isomorphism for $\xi\in \mathbb{C}\setminus \{0 \}$. We denote the corresponding functor by $G_{\xi}:=\omega^{*}_{\xi}$, it is invertible with inverse $G_{\xi^{-1}}$. For $\xi^k=1$, such functors implement the effect of symmetry defects in minimal models with superpotential $W=x^k$ (see also Example~\ref{ex:topdefectsinminimalmodels}).
\end{ex}
When we later impose the fusion functor property~\eqref{fusfunctorprop}, such functors can be related to B-type interfaces. B-type interfaces together with interface fields as morphisms form a category. Similarly, we want to regard the functors $U$ introduced above themselves as objects of a category. Natural transformations would be the most obvious choice for morphisms between the functors, but these will not be enough, as they would for instance not include fermionic interface fields. Instead, we define morphisms in the following way:
\begin{Def}
Let $U,V: R\fmod\to R'\fmod$ be $\mathbb{C}$-linear functors and $n\in\mathbb N_0$.
\begin{itemize}
    \item For $n=0$, a degree-$0$ morphism $\phi^{(0)}_{UV}$ assigns to every (free finite-rank) $R$-module $M$ an element
\begin{equation}
\phi^{(0)}_{UV}(M)\in \Hom_{R'}\big(U(M),V(M)\big)\ .
\end{equation}

\item For $n\ge 1$, a degree-$n$ morphism $\phi^{(n)}_{UV}$ assigns to every chain of $n$
composable $R$-module homomorphisms
\[
M_{n+1}\xleftarrow{f_n} M_n \xleftarrow{f_{n-1}} \cdots \xleftarrow{f_1} M_1
\]
an element
\begin{equation}
\phi^{(n)}_{UV}(f_n,\dots,f_1)\in \Hom_{R'}\big(U(M_1),V(M_{n+1})\big).
\end{equation}
The assignment $(f_n,\dots,f_1)\mapsto \phi^{(n)}_{UV}(f_n,\dots,f_1)$ is
$\mathbb C$-multilinear. 
\end{itemize}
We denote by $\Hom_n(U,V)$ the $R'$-module of such
degree-$n$ morphisms.
\end{Def}
For degree $n\geq 1$, we can view a morphism $\phi^{(n)}_{UV}$ as defining for every tuple $(M_{n+1},\dotsc ,M_{1})$ of modules a $\mathbb{C}$-linear map from a component $C(M_{n+1},\dotsc ,M_{1})$ of a nerve of degree $n$,
\begin{equation}
C (M_{n+1}\listdots M_{1}) =  \Hom_R
(M_{n},M_{n+1})\otimes_{\mathbb{C}}\Hom_R
(M_{n-1},M_{n})\otimes_{\mathbb{C}}\dotsc \otimes_{\mathbb{C}}\Hom_R
(M_{1},M_{2})\ ,
\end{equation}
to $\Hom_{R'} (U (M_{1}),V (M_{n+1}))$. An example for a degree-$1$ morphism is discussed in Example~\ref{ex:degree1morphism}.

\begin{figure}
\begin{center}
\scalebox{1}{\includegraphics{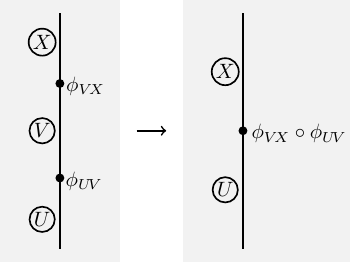}}
\end{center}
\caption{\label{fig:verticalcomp}Composition of morphisms. Having the interpretation of the functors as interfaces in mind, the composition of morphisms can be illustrated similarly to the composition of interface fields.}
\end{figure}

Composition of morphisms is defined in the obvious way: given two
morphisms $\phi_{UV}^{(n)}$ and $\phi_{VX}^{(n')}$, the composite
$\phi_{VX}^{(n')}\circ\phi_{UV}^{(n)}$ is a morphism of degree $n+n'$ that
acts on tuples of composable homomorphisms $f_i:M_i\to M_{i+1}$ as
\begin{align}
\phi_{VX}^{(n')}\circ\phi_{UV}^{(n)} \big(f_{n+n'},\dotsb ,f_{1}
\big)= \phi_{VX}^{(n')} \big(f_{n+n'},\dotsb,
f_{n+1}\big) \, 
\phi_{UV}^{(n)}\big(f_{n} ,\dotsb,f_{1} \big) \ .
\label{compofmorphisms}
\end{align}
For $n=0$ or $n'=0$, the corresponding list of homomorphisms is empty, and the degree-$0$ morphism is evaluated on the relevant module. Thus, for $n=0$ the right hand side contains $\phi^{(0)}_{UV}(M_1)$, while for $n'=0$ it contains $\phi^{(0)}_{VX}(M_{n+1})$.
As one can directly observe, this composition is associative. Graphically, we can illustrate the composition as in Figure~\ref{fig:verticalcomp}, similarly to the composition of interface fields. The identity morphism
$\id_{U}\in \Hom_{0}(U,U)$ for a $\mathbb{C}$-linear functor $U$ is of
degree $0$, and it maps a module $M$ to the identity homomorphism, 
\begin{equation}
\id_{U} (M) =  U(\one_{M})=\one_{U(M)} \ .
\end{equation}
Objects and morphisms together form a category:
\begin{Def}
$\Fun (R',R)$ is the category with $\mathbb{C}$-linear
functors from $R\fmod$ to $R'\fmod$ as objects, and for two
functors $U,V$ the space of morphisms is the $\mathbb{N}_{0}$-graded vector space
\begin{equation}
\Hom (U,V) = \bigoplus_{n=0}^{\infty} \Hom_{n} (U,V) \ .
\end{equation}
\end{Def}
The morphisms defined so far are too general. For instance, for the degree-$0$ morphisms, $\phi^{(0)}_{UV}(M)\in \Hom_{R'}(U,V)$ could be chosen arbitrarily and does not need to satisfy an additional condition. At the end, we would like the degree-$0$ morphisms to be natural transformations, and to achieve this, we introduce a differential which will allow us to reduce the morphism spaces by considering cohomology.

A differential $d$ that equips the category
$\Fun (R',R)$ with the structure of a differential graded category over $\mathbb{C}$ can be defined as follows:
\begin{Def}
For two $\mathbb{C}$-linear functors $U,V$ from $R\fmod$ to
$R'\fmod$ the differential $d$ maps a morphism $\phi_{UV}^{(n)}$ to a morphism 
$d\phi_{UV}^{(n)}$ of degree $n+1$,
\begin{align}
\big(d\phi_{UV}^{(n)}\big) \big(f_{n+1},\dotsc ,f_{1}\big) 
&= V (f_{n+1})\, \phi_{UV}^{(n)} \big(f_{n},\dotsc ,f_{1}\big) \nonumber\\
&\qquad - \phi_{UV}^{(n)}\big(f_{n+1}\, f_{n},f_{n-1},\dotsc
,f_{1}\big)\nonumber\\
&\qquad +\dotsb \nonumber\\
&\qquad + (-1)^{n} \phi_{UV}^{(n)}\big(f_{n+1},\dotsc
,f_{3},f_{2}\, f_{1}\big)\nonumber\\
&\qquad + (-1)^{n+1}  \phi_{UV}^{(n)}\big(f_{n+1},\dotsc ,f_{2}\big)\,  U (f_{1})  \ .
\label{defdiff}
\end{align}
\end{Def}
It is now straightforward to check the following lemma, which
guarantees that $\Fun (R',R)$ together with $d$ indeed forms a differential graded category.\footnote{This is the
differential $d$ that appears in Hochschild-Mitchell cohomology \cite{Herscovich:2005} with
values in the $R\fmod$-bimodule
$\prod_{(M,N)}\Hom_{R'}(U(M),V(N))$. Its construction can be understood as an alternating sum of coface maps.}
\begin{lem}
The differential $d$ squares to zero, $d \circ d = 0$. Furthermore, it satisfies a (graded) Leibniz property,
\begin{equation}
d\big(\phi^{(n)}_{VW}\circ \phi^{(n')}_{UV} \big) =
\big(d\phi^{(n)}_{VW} \big)\circ \phi^{(n')}_{UV} +
(-1)^{n} \phi^{(n)}_{VW} \circ \big(d \phi^{(n')}_{UV} \big) \ .
\end{equation}
\end{lem}
\begin{ex}\label{ex:degree1morphism}
Let $R'=R=\mathbb{C}[x]$, and we consider the invertible functors $G_{\xi_{1}}$ and $G_{\xi_{2}}$ from Example~\ref{ex:invertiblefunctors} with $\xi_{1}\not = \xi_{2}$. We then construct a degree-$1$ morphism $\phi^{(1)}_{G_{\xi_{1}}G_{\xi_{2}}}\in \Hom_{1}(G_{\xi_{1}},G_{\xi_{2}})$: it takes as argument a homomorphism $f\in \Hom_R(M,N)$ and is defined by\footnote{Notice that $G_{\xi_{1}} (f)-G_{\xi_{2}} (f)$ is always of the form $x\,g$ with some homomorphism $g$, so this is well-defined.}
\begin{equation}
\phi^{(1)}_{G_{\xi_{1}}G_{\xi_{2}}} (f):= \frac{1}{x}\big(G_{\xi_{1}} (f)-G_{\xi_{2}} (f) \big) \in\Hom_R(G_{\xi_1}(M),G_{\xi_2}(N))\ .
\end{equation}
For example, for $M=N=R$ every homomorphism $f:R\to R$ is of the form $f=p\cdot \one_R$ with a polynomial $p\in R$. Then $\phi^{(1)}_{G_{\xi_{1}}G_{\xi_{2}}} (f)$ is again a homomorphism from $R$ to $R$ which acts by multiplication with the polynomial (we write $p(x)$ for the polynomial to simplify the notation)
\begin{equation}
    \frac{1}{x}\big( p(\xi_1 x) -p(\xi_2 x)\big)\in R\,.
\end{equation}
$\phi^{(1)}_{G_{\xi_{1}}G_{\xi_{2}}}$ is $d$-closed:
\begin{align}
d\phi^{(1)}_{G_{\xi_{1}}G_{\xi_{2}}} (f_{2},f_{1}) &= G_{\xi_{2}} (f_{2}) \frac{1}{x}\big(G_{\xi_{1}} (f_{1})-G_{\xi_{2}} (f_{1}) \big) - \frac{1}{x}\big(G_{\xi_{1}} (f_{2}f_{1})-G_{\xi_{2}} (f_{2}f_{1}) \big)\nonumber\\
&\quad  + \frac{1}{x}\big(G_{\xi_{1}} (f_{2})-G_{\xi_{2}} (f_{2}) \big) G_{\xi_{1}} (f_{2})= 0 \ ,
\end{align}
where we have used the functorial property $G_{\xi} (f_{2}f_{1})=G_{\xi} (f_{2})G_{\xi} (f_{1})$. These degree-$1$ morphisms will play a role later in minimal models when constructing fusion functors for topological defects via a generalized cone construction (see Example~\ref{ex:topdefectsinminimalmodels}).
\end{ex}

\subsection{The cohomology category}
From $\Fun (R',R)$ we can construct the graded cohomology category by taking the cohomology groups
$H^{n}(\Hom(U,V))$ as morphisms: 
\begin{Def}
$H^{\bullet}\Fun_{R',R}$ is the category which has 
the same objects as $\Fun (R',R)$, and where the set of morphisms
between two functors $U,V$ is given by 
\begin{equation}
H^{\bullet}(\Hom(U,V)) = \bigoplus_{n=0}^{\infty} H^{n}(\Hom(U,V)) \ .
\end{equation}
\end{Def}

Notice that -- as advertised before -- the degree-$0$ morphisms now are just the natural
transformations between $U$ and $V$. Evaluating the condition that a degree-$0$ morphism is closed, $d\phi^{(0)}_{UV}=0$, on a homomorphism $f:M\to N$, one obtains
\begin{equation}
(d\phi_{UV}^{(0)}) (f)
= V (f) \, \phi_{UV}^{(0)}
(M) - \phi_{UV}^{(0)} (N)\,  U (f)   = 0 \ ,
\end{equation}
which is the naturality condition.

For the application of this formalism to the concept of fusion
functors, it is useful to reduce the spaces of morphisms in
$\Fun (R',R)$. For two linear functors $U,V$ we consider the
subspace of \textsl{normalized morphisms} of degree $n$ that
vanish when one entry is the identity endomorphism of a module,
\begin{equation}\label{defreducedspace}
\Hom_{n}^{\mathrm{norm}}(U,V) = \{\phi^{(n)}_{UV} \in \Hom_{n} (U,V),\
\phi^{(n)}_{UV} (f_{n},\dotsc ,f_{1}) = 0 \ \text{if any}\
f_{i}=\one \} \ .  
\end{equation}
The total space of normalized morphisms is then
\begin{equation}
\Hom^{\mathrm{norm}} (U,V) = \bigoplus_{n=0}^{\infty}
\Hom^{\mathrm{norm}}_{n}(U,V)\ .
\end{equation}
We denote by $\Fun^{\mathrm{norm}} (R',R)$ the subcategory with these normalized morphism spaces.
The differential $d$ also acts on it, so that
again we can consider its cohomology. One finds that
\begin{lem}
\begin{equation}
\forall n\in\mathbb{N}_{0}\qquad H^{n} (\Hom^{\mathrm{norm}}(U,V)) \cong  H^{n} (\Hom(U,V))\ .
\end{equation}
\end{lem}
Therefore one does not lose any information on the cohomology by 
considering only the subspace of normalized morphisms, which vanish when one
argument is the identity. This is a general fact on cochain complexes of a cosimplicial object in an abelian category (see, e.g., \cite[chapter 8]{Weibel} for the corresponding statement for the homology). To be self-contained, we present a proof of this lemma in appendix~\ref{sec:projection} by defining a projection operator $P$ on $\Hom (U,V)$ (see~\eqref{defP}) whose image is $\Hom^{\mathrm{norm}} (U,V)$ and which can be written as
\begin{equation}\label{projectionP}
P = \mathrm{id} - d R - R d \ .
\end{equation}
Here, $R$ is an operator on the morphism spaces of degree $-1$, i.e.\ it maps a degree-$n$ morphism to a degree-$(n-1)$ morphism. If $[\psi]\in H^{n} (\Hom(U,V))$ is a cohomology class, we therefore have $[\psi]=[P\psi]$, and we can always consider a representative in the space $\Hom^{\mathrm{norm}}(U,V)$. 

\subsection{Horizontal composition}\label{sec:horizontal}
If we have a functor $U_{1}:R_{1}\fmod\to R_{2}\fmod$ and another functor $U_{2}:R_{2}\fmod\to R_{3}\fmod$, we can compose them to a functor $U_{2}\otimes U_{1}:=U_{2}U_{1}$. We can then ask whether we can also compose morphisms $\phi_{U_{1}V_{1}}$ and $\phi_{U_{2}V_{2}}$ to a morphism from $U_{2}U_{1}$ to $V_{2}V_{1}$ (see Figure~\ref{fig:horizontalcomp0}). We define the following horizontal composition:\footnote{The investigation of horizontal composition in this context is based on a suggestion by Nils Carqueville.}

\begin{figure}
\begin{center}
\scalebox{.9}{\includegraphics{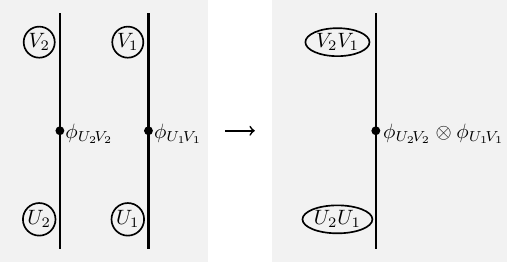}}
\end{center}
\caption{\label{fig:horizontalcomp0}Horizontal composition of morphisms $\phi_{U_{1}V_{1}}$ and $\phi_{U_{2}V_{2}}$ (see Definition~\ref{def:horizontalcomp}).}
\end{figure}

\begin{Def}\label{def:horizontalcomp}
Let $U_{1},V_{1}$ be functors from $R_{1}\fmod$ to $R_{2}\fmod$ and $U_{2},V_{2}$ functors from $R_{2}\fmod$ to $R_{3}\fmod$. For morphisms $\phi^{(n_{i})}_{U_{i}V_{i}}\in \Hom_{n_{i}}(U_{i},V_{i})  $ ($i=1,2$), we define
\begin{equation}
\phi^{(n_{2})}_{U_{2}V_{2}} \otimes \phi^{(n_{1})}_{U_{1}V_{1}} \in \Hom_{n_{1}+n_{2}}(U_{2}U_{1},V_{2}V_{1}) 
\end{equation}
by
\begin{multline}\label{horizontalcompdef}
\phi^{(n_{2})}_{U_{2}V_{2}} \otimes \phi^{(n_{1})}_{U_{1}V_{1}} (f_{n_{1}+n_{2}},\dotsc ,f_{1})\\
 = \phi^{(n_{2})}_{U_{2}V_{2}} \big(V_{1} (f_{n_{1}+n_{2}}),\dotsc ,V_{1} (f_{n_{1}+1})\big) \, U_{2} \big(\phi^{(n_{1})}_{U_{1}V_{1}} (f_{n_{1}},\dotsc ,f_{1})\big)\ .
\end{multline}
\end{Def}
When we define the composition in such a way, it immediately follows that
\begin{equation}\label{horizontalcomp1}
\phi^{(n_{2})}_{U_{2}V_{2}} \otimes \phi^{(n_{1})}_{U_{1}V_{1}} = (\phi^{(n_{2})}_{U_{2}V_{2}} \otimes \one  )\circ (\one  \otimes \phi^{(n_{1})}_{U_{1}V_{1}}) \ .
\end{equation}
In other words, the result of horizontal composition of two morphisms can be obtained by choosing a vertical ordering, followed by a horizontal composition with identity morphisms and an ordinary (vertical) composition of morphisms (as illustrated in Figure~\ref{fig:horizontalcomp1}). One could have defined horizontal composition by using the opposite vertical ordering as in Figure~\ref{fig:horizontalcomp2}. Up to a possible sign (that arises in exchanging the ordering when both morphisms are odd), this leads to the same morphism, but only in the cohomology category. Indeed, in cohomology, horizontal and vertical composition commute, which is the content of the following proposition (illustrated in Figure~\ref{fig:horizontalcomp3}).

\begin{figure}
\begin{center}
\scalebox{.9}{\includegraphics{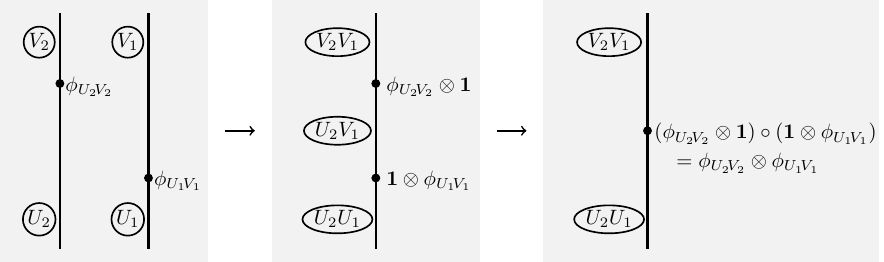}}
\end{center}
\caption{\label{fig:horizontalcomp1}Horizontal composition as defined in Definition~\ref{def:horizontalcomp} can be viewed as first separating the morphisms vertically, then composing with an identity morphism, and finally composing them vertically.}
\end{figure}

\begin{figure}[h]
\begin{center}
\scalebox{.9}{\includegraphics{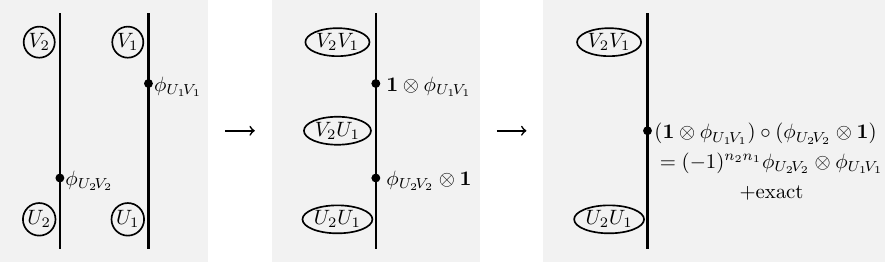}}
\end{center}
\caption{\label{fig:horizontalcomp2}Doing the horizontal composition for the reversed vertical ordering as compared to Figure~\ref{fig:horizontalcomp1}, one obtains for closed morphisms the same result (with an additional sign in case the degrees, $n_1$ and $n_2$, of both morphisms are odd) up to an exact morphism, so we have equality only in the cohomology category, see~\eqref{horizontalcomp2}.}
\end{figure}

\begin{figure}[h]
\begin{center}
\scalebox{.9}{\includegraphics{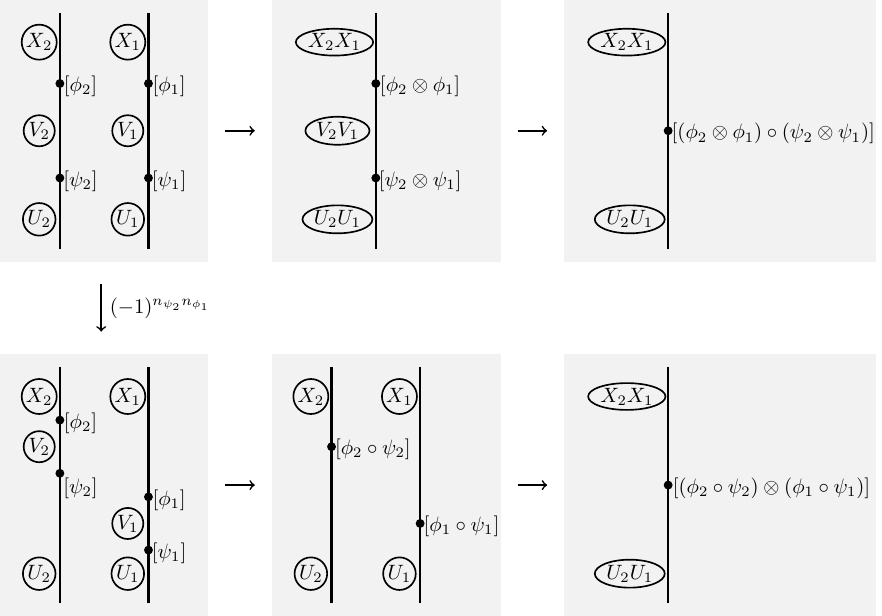}}
\end{center}
\caption{\label{fig:horizontalcomp3}Vertical and horizontal composition (anti-commute) in cohomology (see Proposition~\ref{prop:horizontalcomposition}). The upper figures illustrate first horizontal and then vertical composition. By first considering vertical composition as in the lower figures, one has to switch the vertical ordering of $\phi_1$ and $\psi_2$, which results in an overall sign if both morphisms are odd.}
\end{figure}

\begin{prop}\label{prop:horizontalcomposition}
Let $U_{1},V_{1},X_{1}$ and $U_{2},V_{2},X_{2}$ be functors as in the previous definition, and $[\psi^{( m_{i})}]\in H^{\bullet} (\Hom(U_{i},V_{i})) $ and $[\phi^{( n_{i})}]\in H^{\bullet} (\Hom(V_{i},X_{i})) $($i=1,2$) morphisms in the cohomology categories. Then
\begin{equation}
[(\phi^{(n_{2})} \otimes \phi^{(n_{1})})  \circ (\psi^{(m_{2})} \otimes \psi^{(m_{1})}) ] = (-1)^{n_{1}m_{2}} [(\phi^{(n_{2})}\circ \psi^{(m_{2})}) \otimes (\phi^{(n_{1})}\circ \psi^{(m_{1})} )]\ .
\end{equation}
\end{prop}

\begin{proof}
We first notice that because of~\eqref{horizontalcomp1} it is enough to show that
\begin{equation}\label{horizontalcomp2}
[(\one \otimes \phi^{(n)} ) \circ (\psi^{(m)} \otimes \one )] = (-1)^{nm} [\psi^{(m)} \otimes \phi^{(n)}]\ .
\end{equation}
This means that different choices for the vertical ordering of the morphisms lead to the same result up to a possible sign (see Figure~\ref{fig:horizontalcomp1} and Figure~\ref{fig:horizontalcomp2}).

We evaluate the left-hand side of~\eqref{horizontalcomp2} using the definition~\eqref{horizontalcompdef}, and we obtain
\begin{align}
C_{L}&= (\one_{V_{2}} \otimes \phi^{(n)} ) \circ (\psi^{(m)} \otimes \one_{U_{1}} ) (f_{m+n},\dotsc ,f_{1})\nonumber \\
&= V_{2}  \big(\phi^{(n)} (f_{m+n},\dotsc ,f_{m+1})\big) \, \psi^{(m)} \big(U_{1} (f_{m}),\dotsc ,U_{1} (f_{1})\big)\ .
\end{align}
On the other hand, the right-hand side is
\begin{align}
C_{R} &= (-1)^{nm} \psi^{(m)} \otimes \phi^{(n)}  (f_{m+n},\dotsc ,f_{1})\nonumber \\
&= (-1)^{nm} \psi^{(m)} \big(V_{1} (f_{m+n}),\dotsc ,V_{1} (f_{n+1})\big) \, U_{2} \big(\phi^{(n)} (f_{n},\dotsc ,f_{1})\big) \ .
\end{align}
We want to show that $C_{L}$ and $C_{R}$ agree up to exact terms if $\phi^{(n)}$ and $\psi^{(m)}$ are closed morphisms. We observe that because of $d\psi^{(m)}=0$ we have
\begin{align}
C_{L} &= C_{L}- d\psi^{(m)}\big(\phi^{(n)} (f_{m+n},\dotsc ,f_{m+1}),U_{1} (f_{m}),\dotsc ,U_{1} (f_{1}) \big) \nonumber\\
&= \psi^{(m)}\big(\phi^{(n)} (f_{m+n},\dotsc ,f_{m+1})\, U_{1} (f_{m}),U_{1}(f_{m-1}),\dotsc ,U_{1} (f_{1}) \big)\nonumber\\
&\ + \dotsb + (-1)^{m}\psi^{(m)} \big(\phi^{(n)} (f_{m+n},\dotsc ,f_{m+1}),U_{1} (f_{m}),\dotsc ,U_{1} (f_{2}) \big)\, U_{2}U_{1} (f_{1}) \nonumber\\
 &=: B^{(m)}\ .
\label{CL}
\end{align}
Similarly, we have
\begin{align}
C_{R} &= C_{R} + (-1)^{m (n+1)} d\psi^{(m)} \big(V_{1} (f_{m+n}),\dotsc ,V_{1} (f_{n+1}),\phi^{(n)} (f_{n},\dotsc ,f_{1})) \big)\nonumber\\
&= (-1)^{m (n+1)} V_{2}V_{1} (f_{m+n})\, \psi^{(m)}\big(V_{1} (f_{m+n-1}),\dotsc ,V_{1} (f_{n+1}),\phi^{(n)} (f_{n},\dotsc ,f_{1}) \big)\nonumber\\
&\  + \dotsb +(-1)^{mn} \psi^{(m)}\big(V_{1} (f_{m+n}),\dotsc ,V_{1} ( f_{n+1})\, \phi^{(n)} (f_{n},\dotsc ,f_{1})\big)\nonumber\\
&=: A^{(1)}\ .
\label{CR}
\end{align}
We construct a chain $A^{(1)},B^{(1)},\dotsc ,A^{(m)},B^{(m)}$ of terms and show that they all agree in cohomology, which then proves the proposition. To shorten the notation we introduce the symbol
\begin{equation}
\phi^{(n)}_{\ell}:=\phi^{(n)} (f_{n+\ell -1},\dotsc ,f_{\ell}) \qquad (\ell =1,\dotsc ,m+1)\ .
\end{equation}
We then define for $\ell =1,\dotsc ,m$
\begin{align}
& A^{(\ell)} := (-1)^{n(m-\ell +1)} \Big\{ \nonumber \\
&(-1)^{m-\ell +1} V_{2}V_{1} (f_{m+n})\, \psi^{(m)} \big(V_{1} (f_{m+n-1})\listdots V_{1} (f_{n+\ell}),\phi^{(n)}_{\ell} ,U_{1} (f_{\ell -1})\listdots U_{1} (f_{1}) \big)\nonumber\\
&+ (-1)^{m-\ell} \psi^{(m)} \big(V_{1} (f_{m+n}\, f_{m+n-1})\listdots V_{1} (f_{n+\ell }), \phi^{(n)}_{\ell} ,U_{1} (f_{\ell -1})\listdots U_{1} (f_{1})\big)\nonumber\\
&+ \dotsb  -\psi^{(m)} \big(V_{1} (f_{m+n})\listdots V_{1} (f_{n+\ell +1}\, f_{n+\ell}),\phi^{(n)}_{\ell},U_{1} (f_{\ell -1})\listdots U_{1} (f_{1}) \big)\nonumber\\
&+ \psi^{(m)}\big(V_{1} (f_{m+n})\listdots V_{1} (f_{n+\ell +1}),V_{1} (f_{n+\ell})\, \phi^{(n)}_{\ell} ,U_{1} (f_{\ell -1})\listdots U_{1} (f_{1}) \big)\Big\}\ ,
\end{align}
and
\begin{align}
&B^{(\ell)} := (-1)^{n (m-\ell)}\Big\{\nonumber \\
& \psi^{(m)}\big(V_{1} (f_{m+n})\listdots V_{1} (f_{n+\ell +1}),\phi^{(n)}_{\ell +1}\, U_{1} (f_{\ell}),U_{1} (f_{\ell -1})\listdots U_{1} (f_{1}) \big)\nonumber\\
&- \psi^{(m)}\big(V_{1} (f_{m+n})\listdots V_{1} (f_{n+\ell +1}),\phi^{(n)}_{\ell +1},U_{1} (f_{\ell}\, f_{\ell -1})\listdots U_{1} (f_{1}) \big)\nonumber\\
&+\dotsb +(-1)^{\ell -1} \psi^{(m)}\big(V_{1} (f_{m+n})\listdots V_{1} (f_{n+\ell +1}),\phi^{(n)}_{\ell +1},U_{1} (f_{\ell})\listdots U_{1} (f_{2}\, f_{1}) \big)\nonumber\\
&+(-1)^{\ell} \psi^{(m)}\big(V_{1} (f_{m+n})\listdots V_{1} (f_{n+\ell +1}),\phi^{(n)}_{\ell +1},U_{1} (f_{\ell})\listdots U_{1} (f_{2}) \big)\, U_{2}U_{1} (f_{1})\Big\}\ .
\end{align}
We first observe that for $\ell =2,\dotsc ,m$
\begin{align}
&A^{(\ell)} - B^{(\ell -1)}\nonumber \\
&\ = (-1)^{(n+1) (m-\ell +1)} d\psi^{(m)}\big(V_{1} (f_{m+n})\listdots V_{1} (f_{\ell+n}),\phi^{(n)}_{\ell},U_{1} (f_{\ell -1})\listdots U_{1} (f_{1}) \big)\nonumber \\
&\  = 0 \ .
\end{align}
To finish the proof we will now show that $A^{(\ell)}-B^{(\ell)}$ is exact. Introduce for $\ell =1,\dotsc ,m$ the following morphisms of degree $(m+n-1)$ between $U_{2}U_{1}$ and $V_{2}V_{1}$,
\begin{multline}
\xi_{\ell }^{(m+n-1)} (f_{m+n-1}\listdots f_{1}) \\
:= \psi^{(m)} \big(V_{1} (f_{m+n-1})\listdots V_{1} (f_{n+\ell}),\phi^{(n)}_{\ell},U_{1} (f_{\ell})\listdots U_{1} (f_{1}) \big)\ .
\end{multline}
When we evaluate $d\xi_{\ell }^{(m+n-1)}$ on homomorphisms, we recognize terms that appear in $A^{\ell}$ and $B^{\ell}$, plus additional terms:
\begin{align}
&d\xi_{\ell }^{(m+n-1)} (f_{m+n}\listdots f_{1}) - (-1)^{(n+1)(m-\ell +1)} (A^{(\ell)}-B^{(\ell)}) \\
&\quad = (-1)^{m-\ell} \psi^{(m)}\big(V_{1} (f_{m+n})\listdots V_{1} (f_{\ell +n+1}),V_{1} (f_{\ell +n})\, \phi^{(n)}_{\ell},U_{1} (f_{\ell -1})\listdots U_{1} (f_{1}) \big)\nonumber\\
&\qquad +   (-1)^{m-\ell+1} \psi^{(m)} \big(V_{1} (f_{m+n})\listdots \phi^{(n)} (f_{\ell +n}\, f_{\ell +n-1}\listdots f_{\ell})\listdots U_{1} (f_{1})\big)\nonumber\\
&\qquad +\dotsb + (-1)^{m+n-\ell}\psi^{(m)}\big(V_{1} (f_{m+n})\listdots \phi^{(n)} (f_{\ell +n}\listdots f_{\ell +1}\, f_{\ell})\listdots U_{1} (f_{1}) \big)\nonumber\\
&\qquad + (-1)^{m+n-\ell+1}\psi^{(m)}\big(V_{1} (f_{m+n})\listdots \phi^{(n)} (f_{\ell +n}\listdots f_{\ell +1})\, U_{1} (f_{\ell})\listdots U_{1} (f_{1}) \big)\nonumber\\
&\quad = (-1)^{m-\ell} \psi^{(m)}\big(V_{1} (f_{m+n})\listdots V_{1} (f_{\ell +n+1}), d\phi^{(n)} (f_{\ell +n}\listdots f_{\ell })\listdots U_{1} (f_{1})\big) = 0\ .
\end{align}
In the last step we used that $\phi^{(n)}$ is closed. This finally shows that all terms $A^{(1)},B^{(1)},\dotsc ,A^{(m)},B^{(m)}$ agree in cohomology, and the same applies to $C_{L}$ and $C_{R}$ due to~\eqref{CL} and~\eqref{CR}. 
\end{proof}

The horizontal composition provides the category $H^{\bullet}\Fun (R,R)$ of endofunctors on $R\fmod$ with the structure of a monoidal supercategory (in the sense of~\cite{Brundan:2017}):
\begin{prop}\label{prop:monoidalsupercategory}
The graded cohomology category $H^{\bullet}\Fun (R,R)$ (where we consider the morphism spaces as $\mathbb{Z}_{2}$-graded vector spaces whose even/odd parts are formed by the morphisms of even/odd degree) is a strict monoidal supercategory when supplemented by the tensor product bifunctor
\begin{equation}
\begin{array}{rl} \otimes : H^{\bullet}\Fun (R,R) \times H^{\bullet}\Fun (R,R) &\longrightarrow H^{\bullet}\Fun (R,R)\\
(U_{2},U_{1}) &\longmapsto U_{2}U_{1}\\
(\psi_{U_{2}V_{2}},\phi_{U_{1}V_{1}}) & \longmapsto \psi_{U_{2}V_{2}}\otimes \phi_{U_{1}V_{1}}\ ,
\end{array}
\end{equation}
the unit object being the identity functor $\mathrm{id}$.
\end{prop}

\subsection{Relation to Hochschild cohomology}
When we later consider fusion functors, the morphisms correspond to interface fields. Therefore it is useful to investigate the structure of the morphism spaces further. 
In fact, there exists a simpler characterization of the morphism spaces in the cohomology category in terms of Hochschild cohomology:
\begin{prop}\label{prop:reltoHochschild}
Let $U,V$ be objects in $H^{\bullet}\Fun (R',R)$. Then the space of morphisms of degree $n$ between $U$ and $V$ is isomorphic (as $R'$-modules) to the $n^{\text{th}}$ Hochschild cohomology space with coefficients in the $R$-bimodule $\Hom_{R'} (U (R),V (R))$,
\begin{equation}
H^{n} (\Hom (U,V)) \cong H\!H^{n} \big(R,\Hom_{R'} (U (R),V (R))\big) \ .
\end{equation}
\end{prop}

The reason behind this statement is that the cohomology class $[\phi^{(n)}_{UV}]$ of a morphism is completely determined by its action on endomorphisms of the rank-one module $R$. We have the following lemma:
\begin{lem}\label{lem:phideterminedbyRendomorphisms}
Let $\phi^{(n)}_{UV}$ be a closed morphism of degree $n$ between the functors $U$ and $V$ that vanishes when applied to sequences of endomorphisms of $R$. Then $\phi^{(n)}_{UV}$ is exact.
\end{lem}
The proof can be found in appendix~\ref{sec:app-proof-of-lemma} on page~\pageref{proof:phideterminedbyRendomorphisms}.
We are now ready to prove proposition~\ref{prop:reltoHochschild}.
\begin{proof}
To prove the statement we first define a map 
\[
\Pi : \Hom_{n} (U,V) \longrightarrow \mathrm{Mult}^{n}_\bC\big(R,\Hom_{R'} (U (R),V (R))\big)
\]
from the space of functor morphisms of degree $n$ to the space of $n$-($\mathbb{C}$-)multilinear homomorphisms from $R^{n}$ to $\Hom_{R'} (U (R),V (R))$ by
\begin{equation}
\Pi (\phi^{(n)}_{UV}) (p_{n},\dotsc ,p_{1}) := \phi^{(n)}_{UV} (R\xleftarrow{p_{n}} R \xleftarrow{p_{n-1}}\dotsb \xleftarrow{p_{1}}R) \in \Hom_{R'} (U (R),V (R))\ .
\end{equation}
Here we use the notation $R\xleftarrow{p}R$ for the $R$-module endomorphism $p\cdot \one_R$ (multiplication by $p\in R$).

On elements $h^{(n)}\in \mathrm{Mult}_{\mathbb{C}}^n\big(R,\Hom_{R'} (U (R),V (R))\big)$ the Hochschild differential acts as
\begin{align}
d_{H\!H} h^{(n)} (p_{n+1},\dotsc ,p_{1}) &= V (p_{n+1})h^{(n)} (p_{n},\dotsc ,p_{1}) - h^{(n)} (p_{n+1}p_{n},\dotsc ,p_{1})\nonumber \\
&\quad  + \dotsb + (-1)^{n+1}h^{(n)} (p_{n+1},\dotsc ,p_{2})\, U (p_{1}) \ ,
\end{align}
and we directly observe that
\begin{equation}
d_{H\!H}\Pi = \Pi d \ .
\end{equation}
Hence, $\Pi$ descends to a well-defined map on cohomology that we also denote by $\Pi$. We now construct an inverse map
\begin{equation}
I : H\!H^{n} \big(R,\Hom_{R'} (U (R),V (R))\big) \longrightarrow H^{n} (\Hom (U,V)) \ .
\end{equation}
For a representative $h^{(n)}$ of a cohomology class in $H\!H^{n} \big(R,\Hom_{R'} (U (R),V (R))\big)$, we define a representative $Ih^{(n)}$ of a class in $\Hom_{n}(U,V)$ by the following steps: First we define $Ih^{(n)}$ on a chain of morphisms involving only the module $R$ by
\begin{equation}
Ih^{(n)} (R\xleftarrow{p_{n}} R \xleftarrow{p_{n-1}}\dotsb \xleftarrow{p_{1}}R) := h^{(n)} (p_{n},\dotsc ,p_{1}) \ .
\end{equation}
This definition guarantees that $\Pi I=\mathrm{id}$. We now have to declare how $Ih^{(n)}$ acts on homomorphisms between general modules. Let us start with the modules $R^{m}=R\times \dotsb \times R$. We denote the projection to the $k^{\text{th}}$ factor by $\pi_{k}$, and the embedding of $R$
as the $k^{\text{th}}$ factor in $R^{m}$ as $\iota_{k}$, such that
\begin{equation}
\sum_{k=1}^{m} \iota_{k}\, \pi_{k} = \one_{R^{m}} \ .
\end{equation} 
Then we define
\begin{multline}\label{defIfreemodules}
Ih^{(n)} (R^{m_{n+1}}\xleftarrow{f_{n}}R^{m_{n}}\xleftarrow{f_{n-1}}\dotsb
\xleftarrow{f_{1}}R^{m_{1}}) \\
:= \sum_{k_{1}=1}^{m_{1}}\dotsb
\sum_{k_{n+1}=1}^{m_{n+1}}
V (\iota_{k_{n+1}}) 
h^{(n)} (\pi_{k_{n+1}}\, f_{n}\,\iota_{k_{n}},\dotsc
,\pi_{k_{2}}\, f_{1}\, \iota_{k_{1}})\, U
(\pi_{k_{1}}) \ .
\end{multline}
Because of linearity this is compatible with the differentials (in the following a summation over the labels $k_{i}$ that appear in an expression is understood):
\begin{align}
&I\,d_{H\!H}h^{(n)} (R^{m_{n+2}}\xleftarrow{f_{n+1}}R^{m_{n+1}}\xleftarrow{f_{n}}\dotsb
\xleftarrow{f_{1}}R^{m_{1}})\\
 &\quad = 
V (\iota_{k_{n+2}}) \, 
d_{H\!H}h^{(n)} (\pi_{k_{n+2}}\,  f_{n+1}\,  \iota_{k_{n+1}},\dotsc
,\pi_{k_{2}}\,  f_{1}\,  \iota_{k_{1}})\,  U
(\pi_{k_{1}}) \\
&\quad = 
V (f_{n+1}\,  \iota_{k_{n+1}}) \,  h^{(n)} (\pi_{k_{n+1}}\,  f_{n}\,  \iota_{k_{n}},\dotsc
,\pi_{k_{2}}\,  f_{1}\,  \iota_{k_{1}}) \nonumber\\
&\qquad -  V (\iota_{k_{n+2}}) \,  h^{(n)} (\pi_{k_{n+2}}\,  f_{n+1}\,  f_{n}\,  \iota_{k_{n}},\dotsc
,\pi_{k_{2}}\,  f_{1}\,  \iota_{k_{1}})\,  U
(\pi_{k_{1}})  \nonumber\\
&\qquad + \dotsb + (-1)^{n+1}  V (\iota_{k_{n+2}}) \,  h^{(n)} (\pi_{k_{n+2}}\,  f_{n+1}\,  \iota_{k_{n+1}},\dotsc
,\pi_{k_{3}}\,  f_{2}\,  \iota_{k_{2}})\,  U
(\pi_{k_{2}}\,  f_{1})\\
&\quad = V (f_{n+1}) \,  I\,h^{(n)} ( R^{m_{n+1}}\xleftarrow{f_{n}}R^{m_{n}}\xleftarrow{f_{n-1}}\dotsb
\xleftarrow{f_{1}}R^{m_{1}}) \nonumber\\
&\qquad +\dotsb + (-1)^{n} I\,h^{(n)} ( R^{m_{n+2}}\xleftarrow{f_{n+1}}R^{m_{n+1}}\xleftarrow{f_{n}}\dotsb
\xleftarrow{f_{2}}R^{m_{2}}) \, U (f_{1}) \\
& \quad = d\,Ih^{(n)} (R^{m_{n+2}}\xleftarrow{f_{n+1}}R^{m_{n+1}}\xleftarrow{f_{n}}\dotsb
\xleftarrow{f_{1}}R^{m_{1}}) \ .
\end{align}
To define $Ih^{(n)}$ in general, we fix for every free finite-rank module $M$ an $R$-module isomorphism
\begin{equation}
g_{M}:M \xrightarrow{\ \sim\ } R^{m_{M}} \ .
\end{equation}
Then we set
\begin{multline}
Ih^{(n)} (M_{n+1}\xleftarrow{f_{n}}M_{n}\xleftarrow{f_{n-1}}\dotsb
\xleftarrow{f_{1}}M_{1}) \\
:= V (g_{n+1}^{-1}) 
Ih^{(n)} (g_{n+1}\, f_{n}\,g_{n}^{-1},\dotsc
,g_{2}\, f_{1}\, g_{1}^{-1}) U
(g_{1}) \ ,
\end{multline}
where $g_{i}=g_{M_{i}}$ and $m_{i}=m_{M_{i}}$. The $Ih^{(n)}$ appearing on the right-hand side only acts on homomorphisms between modules of the form $R^{m_{i}}$ for which it has been defined before in~\eqref{defIfreemodules}. As can be straightforwardly checked, this definition is compatible with the differential,
\begin{equation}
d\,I= I \,d_{H\! H}\ .
\end{equation}
The definition seems to depend on the precise choice of the isomorphisms $g_{i}$, but as we will see shortly, the cohomology class of $Ih^{(n)}$ does not depend on it.

We conclude that $I\,\Pi$ is the identity in cohomology: Let $\phi^{(n)}_{UV}$ be a closed morphism. Then $I\Pi \phi^{(n)}_{UV}$ is also closed and by construction coincides with $\phi^{(n)}_{UV}$ on morphisms between trivial modules $R$. According to lemma~\ref{lem:phideterminedbyRendomorphisms} the difference of $I\Pi \phi^{(n)}_{UV}$ and $\phi^{(n)}_{UV}$ is exact, hence $I\,\Pi$ is the identity in cohomology.

This proves the proposition. On cohomology spaces, $I$ is fixed by being the inverse of $\Pi$, therefore the choices of isomorphisms $g_{i}$ that enter the construction of $I$ do not matter in cohomology.  
\end{proof}

\subsection{Structure of the morphism spaces}\label{sec:structuremorphismspaces}
We now analyze the structure of the cohomology
$H^{\bullet}(\Hom(U,V))$ more closely. For degree-$0$ morphisms $\phi^{(0)}$,
the only condition is the closure with respect to $d$, so
$H^{0}(\Hom(U,V))$ consists of all natural transformations
of $U$ and $V$. As discussed in the previous section, we can restrict to the case of the
rank-one module $R$. The closure condition reads
\begin{equation}\label{closuredegreezero}
\phi^{(0)}(R)\, U (p) = V (p) \, \phi^{(0)} (R) \ ,
\end{equation}
where $p$ is any polynomial (any element in $R$) and is viewed as an endomorphism $p\cdot \one_R$ of $R$. Because $U$ and $V$
are $\mathbb{C}$-linear functors, it is enough to require this
condition for $p$ being a linear monomial, 
$p\in\{x_{1},\dotsc ,x_{d}\}$. Therefore, the degree-$0$ cohomology is given by
\begin{equation}
H^{0} (\Hom (U,V)) \cong \{f \in \Hom_{R'}(U(R),V(R)):
\forall i\in \{1,\dotsc ,d \}\  f\, U(x_{i}) = V(x_{i})\, f  \} \ .
\end{equation}
\begin{ex}\label{ex:degree0forinvertiblefunctors}
    For $R=\mathbb{C}[x]$, consider the functor $G_{\xi}$ from Example~\ref{ex:invertiblefunctors} with $G_\xi(R)=R$, and which on $x$ (seen as an endomorphism of $R$) acts as $G_\xi(x)=\xi \,x$ ($\xi\not= 0$). Because $f\,G_\xi(x)=G_\xi(x)\,f$ is satisfied for all $f\in \Hom_R(R,R)$, we have
    \begin{equation}
        H^{0} (\Hom (G_\xi,G_\xi)) \cong \Hom_R(R,R) \cong R \,.
    \end{equation}
    If we consider functors $G_{\xi_1}$, $G_{\xi_2}$ for $\xi_1\not=\xi_2$, we cannot satisfy $f\,G_{\xi_1}(x)=G_{\xi_2}(x)\,f$ except for $f=0$, hence
    \begin{equation}
        H^{0} (\Hom (G_{\xi_1},G_{\xi_2})) \cong \{0\}\,.
    \end{equation}
\end{ex}
Alternatively we can identify the degree-$0$ morphisms with bimodule
homomorphisms from $R$ (viewed as the bimodule $_{R}R_{R}$) to $\Hom_{R'}(U(R),V(R))$:
\begin{lem}
\begin{equation}
H^{0}(\Hom(U,V)) \cong \Hom_{R,R}\!\big({}_{R}R_{R},\Hom_{R'}(U(R),V(R))
\big) \ .
\end{equation}
\end{lem}
\begin{proof}
A bimodule homomorphism $\eta\in
\Hom_{R,R}\!\big({}_{R}R_{R},\Hom_{R'}(U(R),V(R))\big)$ satisfies for all $p,r\in R$
\begin{equation}
\eta(r\cdot p) = V(r) \, \eta(p) \ ,\quad \eta (p \cdot r) =
\eta(p) \, U (r) \ .
\end{equation}
Hence
\begin{equation}
\eta (1)\, U (p) = \eta (1\cdot p)= \eta (p) = \eta (p\cdot 1) = V
(p)\, \eta (1) \ ,
\end{equation}
and therefore $\eta(1)\in \Hom_{R'}(U(R),V(R))$ obeys the closure
condition~\eqref{closuredegreezero}, and it can be identified with
$\phi^{(0)}(R)$. Conversely, given a $\phi^{(0)}(R)$
satisfying~\eqref{closuredegreezero}, we can define a bimodule
homomorphism $\eta$ by setting $\eta(p):=V(p)\, \phi^{(0)}(R)$.
\end{proof}

Let us now look at degree-$1$ morphisms. The closure condition $d\phi^{(1)}=0$ for a degree-$1$ morphism can be written as 
\begin{equation}\label{productpropertydegree1}
\phi^{(1)} (f_{2}\, f_{1}) = V (f_{2})\, \phi^{(1)} (f_{1}) + \phi^{(1)} (f_{2}) \, U (f_{1})
\end{equation}
for all homomorphisms $f_{1}\in \Hom_R (M_{1},M_{2})$ and $f_{2}\in \Hom_R (M_{2},M_{3})$. As a direct consequence, a closed degree-$1$ morphism satisfies the following factorization property:
\begin{lem}
Let $\phi^{(1)}$ be a morphism between $U$ and $V$. It is closed if and only if for any family of homomorphisms $f_{i}\in \Hom_R (M_{i},M_{i+1})$, $i=1,\dotsc ,p\geq 1$, we have
\begin{equation}
\phi^{(1)} (f_{p} \dotsb f_{1}) = \sum_{i=1}^{p} V (f_{p}\dotsb  f_{i+1})\, \phi^{(1)} (f_{i})\, U (f_{i-1} \dotsb  f_{1}) \ .
\end{equation}
\end{lem}
We can again restrict ourselves to the case where all modules involved are just given by
$R$, and the homomorphisms are given by polynomials $p,q\in R$. From the above lemma, we conclude that a closed morphism $\phi^{(1)}$ is completely determined by specifying $\phi^{(1)}$ on elements that generate $R$. For $R=\mathbb{C}[x_{1},\dotsc ,x_{d}]$ this means that $\phi^{(1)}$ is determined by specifying $d$ elements $e_{i}\in \Hom_{R'}(U (R),V (R))$ such that
\begin{equation}
\forall i=1,\dotsc ,d \qquad \phi^{(1)} (x_{i}) = e_{i} \ . 
\end{equation}
On the other hand, given any such $d$ elements $e_{i}$ we can try to define a morphism $\phi^{(1)}$ on monomials as
\begin{equation}
\phi^{(1)} (x_{i_{k}}\cdot \dotsb \cdot x_{i_{1}}) = \sum_{j=1}^{k} V (x_{i_{k}}\cdot \dotsb \cdot x_{i_{j+1}})\,  e_{i_j} \, U (x_{i_{j-1}}\cdot \dotsb \cdot x_{i_{1}}) \ .
\end{equation}
This is only well-defined if the expression on the right-hand side is invariant under reordering the $x_{i}$ which is equivalent to the condition that
\begin{equation}
\forall i,j =1,\dotsc ,d \qquad V (x_{i}) \, e_{j} + e_{i} \, U (x_{j}) = 
V (x_{j}) \, e_{i} + e_{j} \, U (x_{i}) \ . 
\end{equation}
$B=\Hom_{R'}(U (R),V (R))$ is an $(R,R)$-bimodule, where the multiplication by $r\in R$ from the left and right on $e\in B$ is given by\footnote{It also naturally carries the structure of an $R'$-module.}
\begin{equation}
r\cdot e = V (r)\, e \ ,\quad e \cdot r = e \, U (r) \ .
\end{equation}
The closed degree-$1$ morphisms can therefore be characterized as
\begin{equation}
\{ \phi \in \Hom_{1} (U,V)\vert\, d\phi =0 \} \cong \{(e_{i})\in B^{d} \vert\,\forall i,j \quad x_{i}\cdot e_{j} - e_{j}\cdot x_{i} = x_{j}\cdot e_{i}  - e_{i}\cdot x_{j}  \}\ .
\end{equation}
An exact degree-$1$ morphism $\phi^{(1)}=d\phi^{(0)}$ applied to polynomials (seen as homomorphisms from $R$ to $R$) acts as
\begin{equation}
\phi^{(1)} (x_{i}) = d\phi^{(0)} (x_{i}) = V (x_{i})\, \phi^{(0)} (R) - \phi^{(0)} (R)\, U (x_{i})
\end{equation}
with some $\phi^{(0)} (R)\in B$.
The degree-$1$ cohomology is therefore given as
\begin{equation}
H^{1} (\Hom (U,V)) \cong \frac{\{(e_{i})\in B^{d} \vert \,\forall i,j \quad x_{i} \cdot e_{j} - e_{j} \cdot x_{i} = 
x_{j} \cdot e_{i}  - e_{i} \cdot x_{j}  \}}{\{(e_{i})\in B^{d}\vert\, \exists e\in B:\ e_{i} = x_{i}\cdot e-e\cdot x_{i}\}} \ .
\end{equation}
This becomes particularly simple in the case of one variable, $R=\mathbb{C}[x]$. Then
\begin{equation}
H^{1} (\Hom (U,V)) \cong \frac{\Hom_{R'} (U (R),V (R))}{\{ x\cdot e -e\cdot x\, \vert \,e\in \Hom_{R'} ( U (R),V (R)) \}} \ .
\end{equation}
\begin{ex}
    Consider again the functors $G_\xi$ from Example~\ref{ex:invertiblefunctors} for which we discussed the degree-0 morphisms in Example~\ref{ex:degree0forinvertiblefunctors}. For $U=V=G_\xi$, we then have 
    \begin{equation}
        x\cdot e - e\cdot x = \xi x\,e - e\,\xi x=0 \quad \text{for} \ e\in \Hom_R(R,R)\,,
    \end{equation}
    and hence
    \begin{equation}
        H^{1} (\Hom (G_\xi,G_\xi)) \cong \Hom_{R} (R,R)\cong R\,.
    \end{equation}
    For $U=G_{\xi_1}$, $V=G_{\xi_2}$ and $\xi_1\not= \xi_2$, we instead have
    \begin{equation}
        x\cdot e-e\cdot x = (\xi_1-\xi_2)x\,e \quad \text{for} \ e\in \Hom_R(R,R)\cong R\,,
    \end{equation}
    and we find
    \begin{equation}
        H^{1} (\Hom (G_{\xi_1},G_{\xi_2})) \cong \frac{R}{xR} \cong \mathbb{C}\,.
    \end{equation}
    The space of degree-1 morphisms in this case is one-dimensional as a complex vector space. A basis element is given by the degree-1 morphism $\phi^{(1)}_{G_{\xi_{1}}G_{\xi_{2}}}$ of Example~\ref{ex:degree1morphism}.
\end{ex}
The higher cohomology groups for the one-variable case are trivial. In the general case, we have the following result:
\begin{prop}\label{prop:highercohomologytrivial}
Let $R=\mathbb{C}[x_{1},\dotsc ,x_{d}]$ be a polynomial ring in $d$ variables. Then
\begin{equation}
\text{for}\ n>d:\ H^{n} (\Hom (U,V))\cong 0\ . 
\end{equation}
\end{prop}
This follows from a result on the corresponding Hochschild cohomology; for completeness a proof is given in appendix~\ref{sec:app-proof-of-prop} on page~\pageref{proof:highercohomologytrivial}.

The morphisms become particularly simple for the identity functors $U=V=\id$:
\begin{prop}\label{prop:spectrumid}
Let $R=\mathbb{C}[x_{1},\dotsc ,x_{d}]$ be a polynomial ring in $d$ variables. Then the space of morphisms between the identity functors is given by the exterior algebra of $R^{d}$,
\begin{equation}
 H^{n} (\Hom (\id ,\id ))\cong \mathsf{\Lambda}^{n}R^{d}\ .
\end{equation}
\end{prop}
\begin{proof}
The statement follows from the well-known structure of the Hochschild cohomology for polynomial algebras (see~\eqref{HHcohomologyrings}).
\end{proof}

\subsection{Cones}\label{sec:Cones}
Morphisms $\phi^{(1)}_{UV}$ of degree $1$ are special because they can be used to define a new object $\mathcal{C} (\phi^{(1)}_{UV})$ in analogy to the usual cone construction.\footnote{Usually one considers degree-$0$ morphisms from objects $A$ to $B $ for the cone construction and then shifts it to a degree-$1$ morphism between $A[1]$ and $B$. As we do not have a shift operation in our case, we directly work with degree-$1$ morphisms.} Physically, this corresponds to turning on an interface field in a configuration of two interfaces to produce a new interface (configuration).
\begin{Def}
Let $\phi^{(1)}_{UV}$ be a closed morphism of degree~1 between the functors $U$ and $V$. We define a $\mathbb{C}$-linear functor $\mathcal{C} (\phi^{(1)}_{UV})$ from $R\fmod$ to $R'\fmod$ by
\begin{equation}
\mathcal{C} (\phi^{(1)}_{UV}) (M) = U (M) \oplus V (M)
\end{equation}
and  
\begin{equation}\label{cone}
\mathcal{C}(\phi^{(1)}_{UV}) (f) = \begin{pmatrix}
U (f) & 0 \\
\phi^{(1)}_{UV} (f) & V (f)
\end{pmatrix} \ ,
\end{equation}
and call it the cone of $\phi^{(1)}_{UV}$. 
\end{Def}
One can straightforwardly verify that $\mathcal{C}(\phi^{(1)}_{UV})$ is indeed a functor. The functorial property follows from the closure of $\phi^{(1)}_{UV}$:
\begin{align}
\mathcal{C}(\phi^{(1)}_{UV}) (f) \,\mathcal{C}(\phi^{(1)}_{UV}) (g) &=
\begin{pmatrix}
U (f) & 0\\
\phi^{(1)}_{UV} (f) & V (f)
\end{pmatrix}
\begin{pmatrix}
U (g) & 0\\
\phi^{(1)}_{UV} (g) & V (g)
\end{pmatrix}\\
&= \begin{pmatrix}
U (f)U (g) & 0\\
\phi^{(1)}_{UV} (f)U (g) +V (f)\phi^{(1)}_{UV} (g) & V (f)V (g)
\end{pmatrix}\\
&= \begin{pmatrix}
U (fg) & 0\\
\phi^{(1)}_{UV} (fg) & V (fg)
\end{pmatrix}\\
&=\mathcal{C} (\phi^{(1)}_{UV}) (fg)\ .
\end{align}
The cone construction will be useful to construct new fusion functors from known ones.
\begin{ex}\label{ex:cone}
Let $R'=R=\mathbb{C}[x]$. Furthermore, let $\phi^{(1)}_{G_{\xi_{1}}G_{\xi_{2}}}$ be the morphism between the invertible functors  $G_{\xi_{1}}$ and $G_{\xi_{2}}$ from Example~\ref{ex:degree1morphism}. Then the cone 
\begin{equation}\label{coneexample}
\mathcal{C} (\phi^{(1)}_{G_{\xi_{1}}G_{\xi_{2}}}) = \begin{pmatrix}
G_{\xi_{1}} & 0 \\
\frac{1}{x} \big(G_{\xi_{1}}-G_{\xi_{2}} \big) & G_{\xi_{2}}
\end{pmatrix}
\end{equation}
provides our first example of a functor that is not the extension of scalars of a ring homomorphism $\omega$.
\end{ex}

\subsection{Functors on graded modules}

Before we close this section, we need to consider functors acting on $\mathbb{Z}_{2}$-graded modules in preparing for the introduction of fusion functors in Section~\ref{sec:fusionfunctors}. Fusion functors act on matrix factorizations, which can be viewed as homomorphisms on a graded vector space where the grading $\sigma=(-1)^F$ can be understood in terms of the fermion number.

A $\mathbb{Z}_{2}$-graded $R$-module is a pair $(M,\sigma)$ of an $R$-module $M$ together with an $R$-linear automorphism $\sigma :M\to M$ that is an involution, $\sigma^{2}=\one$. We denote the category of $\mathbb{Z}_{2}$-graded free finite-rank $R$-modules by $R\fmod^{\mathbb{Z}_{2}}$ where the morphism spaces also have $\mathbb{Z}_{2}$-grading and can be split into spaces of even and odd morphisms ($R\fmod^{\mathbb{Z}_{2}}$ therefore has the structure of a supercategory, see, e.g., \cite{Brundan:2017}). A $\mathbb{C}$-linear functor $U$ on $R\fmod$ can be extended to a (super-)functor on $R\fmod^{\mathbb{Z}_{2}}$ where the grading on $U (M)$ is induced by\footnote{One could also induce the grading on $U (M)$ by $-U (\sigma)$. This option could potentially be interesting in the context of fusion functors, which we consider in the next section, but we restrict the discussion here to functors that do not twist the $\mathbb{Z}_{2}$-grading.} $U (\sigma)$. As morphisms between two functors we consider all normalized morphisms that are compatible with the graded structure, and we define the following category:
\begin{Def}
The category $\Fun^{\mathrm{norm},\mathbb{Z}_{2}} (R',R)$ has as objects $\mathbb{C}$-linear functors from $R\fmod^{\mathbb{Z}_{2}}$ to $R'\fmod^{\mathbb{Z}_{2}}$ satisfying $U (M,\sigma )= (U (M),U (\sigma))$. For two functors $U,V$ the space of morphisms is
\begin{equation}
\Hom^{\mathrm{norm},\mathbb{Z}_{2}} (U,V) = \bigoplus_{n=0}^{\infty} \Hom_{n}^{\mathrm{norm},\mathbb{Z}_{2}} (U,V) \ .
\end{equation}
For $n=0$, a morphism $\phi^{(0)}_{UV}\in \Hom_{0}^{\mathrm{norm},\mathbb{Z}_2}(U,V)$ assigns to every graded module $(M,\sigma)$ an $R'$-module homomorphism
\begin{equation}
\phi^{(0)}_{UV}(M,\sigma):(U(M),U(\sigma))\longrightarrow (V(M),V(\sigma))\ ,
\end{equation}
which is \textsl{parity preserving}, i.e.\ it commutes with the grading operators:
\begin{equation}\label{paritypreserving}
    V(\sigma)\,\phi^{(0)}_{UV}(M,\sigma)=\phi^{(0)}_{UV}(M,\sigma)\,U(\sigma)\ .
\end{equation}
For $n\geq 1$, a degree-$n$ morphism $\phi_{UV}^{(n)}\in \Hom_{n}^{\mathrm{norm},\mathbb{Z}_{2}} (U,V)$ assigns to every chain of composable homomorphisms 
\begin{equation}
(M_{n+1},\sigma_{n+1}) \xleftarrow{f_n} (M_n,\sigma_n) \xleftarrow{f_{n-1}}
\dotsb \xleftarrow{f_1} (M_1,\sigma_1)
\end{equation}
between graded modules an element
\begin{align}
\phi^{(n)}_{UV}(f_{n},\dotsc ,f_{1}) \in 
\Hom_{R'} \big(( U (M_{1}),U (\sigma_{1})),( V (M_{n+1}),V (\sigma_{n+1}))\big) \ .
\end{align}
The map
\begin{equation}
    (f_n,\dots,f_1)\mapsto \phi^{(n)}_{UV}(f_{n},\dotsc ,f_{1})
\end{equation}
is $\mathbb{C}$-linear in each entry $f_{i}$. It satisfies compatibility conditions with the grading,
\begin{subequations}\label{gradedcompatibility}
\begin{align}
\phi^{(n)}_{UV} (\sigma_{n+1}f_{n},f_{n-1},\dotsc ,f_{1}) &= V (\sigma_{n+1}) \phi^{(n)}_{UV} (f_{n},f_{n-1},\dotsc ,f_{1})\\
\phi^{(n)}_{UV} (f_{n},\dotsc ,f_{2},f_{1}\sigma_{1}) &= \phi^{(n)}_{UV} (f_{n},\dotsc ,f_{2},f_{1})U (\sigma_{1})\\
\phi^{(n)}_{UV} (f_{n},\dotsc ,f_{i}\sigma_{i},f_{i-1},\dotsc ,f_{1}) &= \phi^{(n)}_{UV} (f_{n},\dotsc ,f_{i},\sigma_{i}f_{i-1},\dotsc ,f_{1})
\ ,
\end{align}
\end{subequations}
where the last condition is only imposed for $n\ge 2$ and $2\leq i\leq n$. In addition we demand that the morphisms are normalized,
\begin{equation}
\phi^{(n)}_{UV} (f_{n},\dotsc ,f_{1}) = 0 \ \text{if any of the}\
f_{i}=\one \ \text{(identity homomorphism)} \ .
\end{equation}
Composition of morphisms is defined as before (see~\eqref{compofmorphisms}).
\end{Def}
The compatibility conditions~\eqref{gradedcompatibility} guarantee that if the homomorphisms $f_{i}$ have $\mathbb{Z}_{2}$-degree $n_{i}$, then $\phi^{(n)}_{UV} (f_{n},\dotsc ,f_{1})$ has degree $n_{1}+\dotsb +n_{n}$,
\begin{equation}
V (\sigma_{n+1})\phi^{(n)}_{UV} (f_{n},\dotsc ,f_{1}) = (-1)^{n_{1}+\dotsb +n_{n}}\phi^{(n)}_{UV} (f_{n},\dotsc ,f_{1}) U (\sigma_{1})\ .
\end{equation}
The definition~\eqref{defdiff} of the differential $d$ is compatible with the graded structure. We denote the corresponding cohomology category by $H^{\bullet}\Fun^{\mathbb{Z}_{2}} (R',R)$.

Note that this definition could be generalized when one also allows for signs in the compatibility conditions~\eqref{paritypreserving} and~\eqref{gradedcompatibility}. With the current definition, the closed morphisms $\phi_{UV}^{(0)}$ of degree $0$ are the \textsl{even} supernatural transformations (i.e.\ parity-preserving natural transformations) between the superfunctors $U$ and $V$, and a modified definition would allow for odd ones. It could be interesting to investigate the consequences of such a generalization for the category of fusion functors defined in the next section.

\section{Fusion functors}\label{sec:fusionfunctors}
When an interface between Landau--Ginzburg models of superpotentials $W'$ and $W$ is fused to a boundary in the model with potential $W$, we obtain a boundary condition for the $W'$-model. For B-type boundary conditions, which can be encoded by matrix factorizations, this leads to a map of one matrix factorization of $W$ to a matrix factorization of $W'$. We can view the matrices appearing in the factorizations as homomorphisms between free finite-rank modules of $R$ or $R'$, hence it is natural to ask whether this map can be realized as a functor on the categories of modules. In this section, we define categories of such \textsl{fusion functors} as subcategories of the category of functors between categories of modules that we investigated in the previous section. We start with a brief summary of some concepts for matrix factorization and then turn towards fusion functors.

\subsection{Background on matrix factorizations}\label{sec:BackgroundMF}

Let $R,R'$ be polynomial rings and $W\in R$, $W'\in R'$ two polynomial superpotentials.
\begin{Def}
A $(W',W)$-matrix factorization $(M,\sigma ,Q)$ consists of a free $\mathbb{Z}_{2}$-graded $(R'\otimes_{\mathbb{C}}R)$-module $(M,\sigma)$ and a homomorphism $Q:M\to M$ that is odd ($Q\sigma +\sigma Q=0$) and satisfies
\begin{equation}
Q^{2} = (W'-W)\,\one_{M} \ .
\end{equation}
 If the module $M$ is of finite rank, we call $(M,\sigma ,Q)$ a finite-rank matrix factorization.
\end{Def}
We often simply say that $Q$ is a matrix factorization of $W'-W$. It corresponds to an interface between B-type topologically twisted Landau--Ginzburg models with superpotentials $W'$ and $W$~\cite{Brunner:2007qu}: $Q$ has the interpretation of the boundary part of the BRST charge, and $\sigma=(-1)^F$ is the grading associated to the fermion number. Similarly, one can describe boundary conditions in a model with superpotential $W'$ by matrix factorizations of $W'$ \cite{Kapustin:2002bi,Orlov:2003yp,Brunner:2003dc,Kapustin:2003ga,Herbst:2004ax} (which corresponds to the case that $R=\mathbb{C}$ is trivial and $W=0$). To describe interface fields, we need to define morphisms between matrix factorizations:
\begin{Def}
A morphism between the $(W',W)$-matrix factorizations $(M_{1},\sigma_{1},Q_{1})$ and $(M_{2},\sigma_{2},Q_{2})$ is given by a module homomorphism $\varphi:M_{1}\to M_{2}$. We distinguish even ($n=0$) and odd ($n=1$) morphisms $\varphi^{(n)}$ satisfying $\sigma_{2}\varphi^{(n)}= (-1)^{n}\varphi^{(n)} \sigma_{1}$. This leads to a category $\MF_{W',W}$ of $(W',W)$-matrix factorizations. The subcategory of finite-rank factorizations is denoted by $\mf_{W',W}$.
\end{Def}
Note that due to the $\mathbb{Z}_{2}$-grading, $\MF_{W',W}$ and $\mf_{W',W}$ can be considered as supercategories in the sense of~\cite{Brundan:2017}.
On this category one can introduce a differential and then consider the corresponding cohomology category.
\begin{Def}
Given two factorizations as above, a differential on the space of morphisms from $M_{1}$ to $M_{2}$ is defined by
\begin{equation}\label{MFdifferential}
\delta_{Q_{1},Q_{2}} \varphi^{(n)} = Q_{2}\, \varphi^{(n)} +
(-1)^{n+1}\varphi^{(n)}\, Q_{1}  \ . 
\end{equation}
Its cohomology defines the space of morphisms in the cohomology category of matrix factorizations that is denoted by $\HMF_{W',W}$ (similarly we denote the subcategory of finite-rank factorizations as $\Hmf_{W',W}$). 
\end{Def}
The cohomology classes of morphisms are interpreted as interface fields in the topologically twisted Landau--Ginzburg model~\cite{Brunner:2007qu}.

Given an odd morphism $\varphi^{(1)}$ between matrix factorizations $Q_{1}$ and $Q_{2}$, one can obtain a new matrix factorization 
\begin{equation}\label{mfcone}
\mathcal{C} (\varphi^{(1)}) = \begin{pmatrix}
Q_{1} & 0 \\
\varphi^{(1)} & Q_{2}
\end{pmatrix}
\end{equation}
which is known as the cone construction.\footnote{Usually this would not be denoted as the mapping cone of the odd morphism $\varphi^{(1)}$, but of its even counterpart between the shifted factorization $Q_{1}[-1]$ and $Q_{2}$; we choose the notation given here for comparison with an analogous construction in the category of fusion functors (see Section~\ref{sec:Cones}).} This new matrix factorization corresponds to the configuration of two interfaces deformed by an interface field~\cite{Herbst:2004zm}.

From a $(W',W)$-factorization $Q_{1}$ and a $(W,W'')$-factorization $Q_{2}$ one can build the tensor product factorization \cite{Yoshino:1998,Khovanov:2004,Khovanov:2004bc,Ashok:2004zb}
\begin{equation}
Q_{1}\boxtimes Q_{2}=Q_{1}\otimes \one + \sigma_{1}\otimes Q_{2}  \ , 
\end{equation}
which is a $(W',W'')$-factorization realized on the $(R'\otimes_{\mathbb{C}} R'')$-module $M_{1}\otimes_{R}M_{2}$. Physically, the tensor product realizes the fusion of the corresponding interfaces~\cite{Brunner:2007qu}.

\subsection{Fusion functors and their morphisms}\label{sec:subsec-fusionfunctors}
We define fusion functors as linear functors on graded ring modules with a special
property that guarantees that a matrix factorization is mapped again to a matrix factorization (see Section~\ref{sec:actingonMF}). Let $R$ and $R'$ be polynomial rings over $\mathbb{C}$,
and $W\in R$ and $W'\in R'$ specific polynomials. We then define
\begin{Def}
A $(W',W)$-fusion functor $U$ is a $\mathbb{C}$-linear functor from $R\fmod^{\mathbb{Z}_{2}}$ to
$R'\fmod^{\mathbb{Z}_{2}}$ with the property
\begin{equation}\label{definingpropFF}
U (W\, f) = W'\, U (f)
\end{equation}
for any module homomorphism $f$.
\end{Def}
The condition~\eqref{definingpropFF} says that $U$ intertwines multiplication by the superpotential. In particular, $U(W\cdot \one)= W' \cdot U(\one)$, which can be interpreted physically as the condition that the F-term of one model is mapped to the F-term of the other model. As we discuss shortly, this property guarantees that matrix factorizations of $W$ are mapped to matrix factorizations of $W'$ (see Proposition~\ref{prop:UQismF}).

Morphisms are defined similarly as for functors between categories of modules, but with an intertwining property with respect to the superpotentials:
\begin{Def}
We define a morphism of degree $n$ between $(W',W)$-fusion functors $U,V$ to be
a morphism $\phi^{(n)}_{UV}\in\Hom^{\mathrm{norm},\mathbb{Z}_{2}}_{n}(U,V)$ that for $n\geq 1$ satisfies
\begin{equation}\label{morphismfusionfunctors}
\phi^{(n)} (f_{n},\dotsc ,W\, f_{i},\dotsc ,f_{1}) = W'\, 
\phi^{(n)} (f_{n},\dotsc ,f_{1})
\end{equation}
for every $1\leq i \leq n$.
\end{Def}
For $n=0$, the condition~\eqref{morphismfusionfunctors} is empty, so degree-$0$ morphisms are precisely natural transformations. For $n=1$, the motivation for imposing~\eqref{morphismfusionfunctors} is that we want the cone $\mathcal{C}(\phi^{(1)}_{UV})$ (see~\eqref{cone}) of fusion functors to be a fusion functor again; this requires compatibility with the superpotentials in the sense of~\eqref{morphismfusionfunctors}. For higher $n$, we impose the analogous condition as a natural extension that is compatible with the vertical and horizontal compositions (see Proposition~\ref{prop:horizontalFF}); in particular, $n=2$ morphisms obtained as composites of two degree-$1$ morphisms automatically satisfy it.

With the definitions above, we define for each pair $W\in R$, $W'\in R'$ a category $\FF_{W',W}$ of fusion
functors which is a subcategory of $\Fun^{\mathbb{Z}_{2}}_{R',R}$. We denote the space of morphisms between the $(W',W)$-fusion functors $U$ and $V$ by $\Hom^{W',W}  (U,V)$. The action of the differential $d$ carries over and equips $\FF_{W',W}$ with the structure of a differential graded category. We denote its cohomology category by $\HFF_{W',W}$.
\begin{ex}
Let $R=\mathbb{C}[x]$ and $R'=\mathbb{C}[x']$, and we consider the superpotentials $W=x^k$ and $W'=x'^{kd}$ with integers $k,d$ greater than $1$. To construct a functor that maps $W$ to $W'$, we introduce the ring homomorphism $\omega_{1} :R\to R'$ defined by $\omega_{1} (x)=x'^{d}$. The extension of scalars functor $U_1=\omega_1^*$ induced by $\omega_1$ acts on a homomorphism between free modules given by a matrix with entries in $R$ by applying $\omega_1$ entrywise to the matrix entries. It is a $(W',W)$-fusion functor,
\begin{equation}\label{U1fusionfunctor}
    U_1(x^k f) =(x'^d)^{k}\, U_1(f)=  x'^{dk}\, U_1(f)\ .
\end{equation}
Consider now another ring homomorphism determined by $\omega_{\xi} (x)=\xi x'^{d}$ with $\xi^{k}=1$. Then the extension of scalars functor $U_{\xi}=\omega_{\xi}^{*}$ also is a $(W',W)$-fusion functor, 
\begin{equation}\label{Uxifusionfuntor}
    U_{\xi} (x^{k}f)= (\xi x'^{d})^{k}\,U_\xi(f)=x'^{dk}\,U_\xi(f)\,.
\end{equation}
Similarly to Example~\ref{ex:degree1morphism}, we define a closed degree-$1$ morphism by
\begin{equation}
\phi^{(1)}_{U_{1}U_{\xi}} (f) = \frac{1}{x'} \big(U_{1} (f)-U_{\xi} (f) \big)\ .
\end{equation}
Because of~\eqref{U1fusionfunctor} and~\eqref{Uxifusionfuntor}, this also satisfies the property~\eqref{morphismfusionfunctors} of a morphism between fusion functors,
\begin{align}
\phi^{(1)}_{U_{1}U_{\xi}} (x^{k}f) &= \frac{1}{x'} \big(U_{1} (x^k f)-U_{\xi} (x^k f) \big) \\
&= \frac{1}{x'} \big(x'^{dk} U_{1} (f)- x'^{dk} U_{\xi} ( f) \big)\\
&= x'^{kd} \,\phi^{(1)}_{U_{1}U_{\xi}} (f) \ .
\end{align}
In physical terms, $W=x^{k}$ and $W'=x'^{kd}$ are the Landau--Ginzburg superpotentials of $A$-type minimal models of different levels \cite{Vafa:1988uu}. Using the relation between fusion functors and interfaces discussed in Section~\ref{sec:ffdescribefusion}, the functors $U_1=\omega_1^*$ and $U_\xi=\omega_\xi^*$ can be viewed as describing operator-like interfaces between these models, with $U_\xi$ differing from $U_1$ by composition with the $\mathbb{Z}_k$ phase symmetry $x\mapsto \xi x$.
\end{ex}

The fusion functor conditions~\eqref{definingpropFF} and~\eqref{morphismfusionfunctors} are compatible with the horizontal composition introduced in Section~\ref{sec:horizontal}:
\begin{prop}\label{prop:horizontalFF}
    Let $U_{1},V_{1}$ be $(W_2,W_1)$-fusion functors and let $U_{2},V_{2}$ be $(W_3,W_2)$-fusion functors for $W_i\in R_i$ $(i=1,2,3)$. Then:
    \begin{itemize}
        \item The composed functors $U_2U_1$ and $V_2V_1$ are $(W_3,W_1)$-fusion functors.
        \item For two morphisms $\phi^{(n_{i})}_{U_{i}V_{i}}\in \Hom_{n_{i}}^{W_{i+1},W_i}(U_{i},V_{i})  $ ($i=1,2$), the tensor product is again a morphism of fusion functors,
        \begin{equation}
\phi^{(n_{2})}_{U_{2}V_{2}} \otimes \phi^{(n_{1})}_{U_{1}V_{1}} \in \Hom_{n_{1}+n_{2}}^{W_3,W_1}(U_{2}U_{1},V_{2}V_{1}) \ .
\end{equation}
    \end{itemize}
    \end{prop}
\begin{proof} We discuss the two statements separately:
    \begin{itemize}
        \item We have to verify the fusion functor property~\eqref{definingpropFF} for the composite $U_2 U_1$: for any $R_1$-module homomorphism $f$ we have
        \begin{equation}
            U_2 U_1 (W_1 \,f) = U_2 \big(W_2 \,U_1(f)\big) = W_3 \,U_2 U_1 (f) \ .
        \end{equation}
        The argument is analogous for $V_2V_1$.
        \item We have to show that the tensor product satisfies the property~\eqref{morphismfusionfunctors} for morphisms of fusion functors. Using the definition~\eqref{horizontalcompdef} of the tensor product, we have to distinguish two cases: either we multiply one of the first $n_2$ entries by $W_1$ (for example $f_{n_1+n_2}$), or we multiply one of the last $n_1$ entries by $W_1$ (for example $f_1$). In the first case, we have
        \begin{align}
            &\phi^{(n_{2})}_{U_{2}V_{2}} \otimes \phi^{(n_{1})}_{U_{1}V_{1}} (W_1\,f_{n_{1}+n_{2}},\dotsc ,f_{1})\\
            &= \phi^{(n_{2})}_{U_{2}V_{2}} \big(V_{1} (W_1\,f_{n_{1}+n_{2}}),\dotsc ,V_{1} (f_{n_{1}+1})\big) \, U_{2} \big(\phi^{(n_{1})}_{U_{1}V_{1}} (f_{n_{1}},\dotsc ,f_{1})\big)\\
            &= \phi^{(n_{2})}_{U_{2}V_{2}} \big(W_2\, V_{1} (f_{n_{1}+n_{2}}),\dotsc ,V_{1} (f_{n_{1}+1})\big) \, U_{2} \big(\phi^{(n_{1})}_{U_{1}V_{1}} (f_{n_{1}},\dotsc ,f_{1})\big)\\
            &= W_3\,\phi^{(n_{2})}_{U_{2}V_{2}} \big( V_{1} (f_{n_{1}+n_{2}}),\dotsc ,V_{1} (f_{n_{1}+1})\big) \, U_{2} \big(\phi^{(n_{1})}_{U_{1}V_{1}} (f_{n_{1}},\dotsc ,f_{1})\big)\ .
        \end{align}
        Similarly, in the second case, we find
        \begin{align}
            &\phi^{(n_{2})}_{U_{2}V_{2}} \otimes \phi^{(n_{1})}_{U_{1}V_{1}} (f_{n_{1}+n_{2}},\dotsc ,W_1\,f_{1})\\
            &= \phi^{(n_{2})}_{U_{2}V_{2}} \big(V_{1} (f_{n_{1}+n_{2}}),\dotsc ,V_{1} (f_{n_{1}+1})\big) \, U_{2} \big(\phi^{(n_{1})}_{U_{1}V_{1}} (f_{n_{1}},\dotsc ,W_1\,f_{1})\big)\\
            &= \phi^{(n_{2})}_{U_{2}V_{2}} \big( V_{1} (f_{n_{1}+n_{2}}),\dotsc ,V_{1} (f_{n_{1}+1})\big) \, U_{2} \big(W_2\,\phi^{(n_{1})}_{U_{1}V_{1}} (f_{n_{1}},\dotsc ,f_{1})\big)\\
            &= W_3\,\phi^{(n_{2})}_{U_{2}V_{2}} \big( V_{1} (f_{n_{1}+n_{2}}),\dotsc ,V_{1} (f_{n_{1}+1})\big) \, U_{2} \big(\phi^{(n_{1})}_{U_{1}V_{1}} (f_{n_{1}},\dotsc ,f_{1})\big)\ ,
        \end{align}
        where in the last line we used the fusion functor property of $U_2$ to replace $U_2(W_2\,\cdot)$ by $W_3\,U_2(\cdot)$, and the $R_3$-linearity of $\phi^{(n_2)}_{U_2V_2}(\dots)$ to pull out the factor $W_3$.
    \end{itemize}
\end{proof}

In particular, for $W'=W$ the class of fusion
functors is closed under composition (and contains $\id$), and the space of morphisms between fusion functors is closed under horizontal composition. Hence $\HFF_{W,W}$ inherits a monoidal supercategory structure: the tensor product on objects is given by composition of functors, the unit object is $\id$, and the $\mathbb{Z}_2$-grading is given by the parity of the (cohomological) degree of morphisms. On morphisms, the tensor product is induced by horizontal composition~\eqref{horizontalcompdef} and satisfies the super interchange law~\eqref{prop:horizontalcomposition}.

\subsection{Action on matrix factorizations}\label{sec:actingonMF}
The defining property~\eqref{definingpropFF} of fusion functors guarantees that they map matrix factorizations of $W$ to matrix factorizations of $W'$, which is the basis of the interpretation of fusion functors as interfaces. Let $(M,\sigma ,Q)$ be a matrix factorization of $W$ with a $\mathbb{Z}_{2}$-graded free finite-rank $R$-module $(M,\sigma)$, and an odd module 
homomorphism $Q$ with $Q^{2}=W\cdot \one_{M}$. Acting on it with $U$ results in a matrix factorization of $W'$:
\begin{prop}\label{prop:UQismF}
Let $U$ be a $(W',W)$-fusion functor, and $(M,\sigma ,Q)$ a matrix factorization
of $W$. Then $(U (M),U (\sigma), U (Q))$ is a matrix factorization of $W'$.
\end{prop}
\begin{proof}
$U (Q)$ is an odd homomorphism,
\begin{equation}
U (Q)U (\sigma) = U (Q\sigma) = -U (\sigma Q)= - U (\sigma)U (Q)\ .
\end{equation}
To see that it defines a matrix factorization, we evaluate
\begin{equation}
U(Q)\, U(Q) = U(Q^{2}) = U(W\cdot \one_{M}) 
= W'\cdot U(\one_{M}) = W'\cdot \one_{U(M)}\ .
\end{equation}
\end{proof}
A fusion functor not only acts on matrix factorizations, but also on morphisms between matrix factorizations.
Indeed, a fusion functor $U$ induces a differential graded functor $\tilde{\Pi}(U)$ from the category $\mf_{W}$ to $\mf_{W'}$:
\begin{prop}
Let $U$ be a $(W',W)$-fusion functor. Then the functor
\begin{align}
\tilde{\Pi} (U): \mf_{W} &\to \mf_{W'}\nonumber\\
(M,\sigma ,Q)&\mapsto (U (M),U (\sigma),U (Q))\nonumber\\
(\varphi:M_{1}\to M_{2}) &\mapsto \big(U (\varphi):U (M_{1})\to U (M_{2})\big)
\label{Utilde}
\end{align}
is differential graded. In particular, it induces a functor $\Pi(U)$ between the cohomology categories,
\begin{equation}\label{defPiU}
    \Pi (U):\Hmf_{W} \longrightarrow \Hmf_{W'}\,.
\end{equation}
\end{prop}
\begin{proof}
Let $\varphi$ be a morphism from\footnote{As before we often denote a matrix factorization $(M,\sigma ,Q)$ just by $Q$.} $Q_{1}$ to $Q_{2}$ of degree $n$. Then $U (\varphi)$ has the same degree:
\begin{equation}
U (\sigma_{2})U (\varphi)+ (-1)^{n+1}U (\varphi)U (\sigma_{1}) = U (\sigma_{2}\varphi + (-1)^{n+1}\varphi \sigma_{1}) = 0\ .
\end{equation}
Furthermore, the functor is compatible with the differential defined in~\eqref{MFdifferential},
\begin{align}
U (\delta_{Q_{1},Q_{2}}\varphi) &= U (Q_{2}\varphi + (-1)^{n+1}\varphi Q_{1})\nonumber\\
&= U (Q_{2})U (\varphi) + (-1)^{n+1}U (\varphi) U (Q_{1})\nonumber\\
&= \delta_{U (Q_{1}),U (Q_{2})} U (\varphi)\,.
\end{align}
\end{proof}
The question arises whether the map $U\mapsto \tilde{\Pi} (U)$ itself gives rise to a functor from $\FF_{W,W'}$ to the category of functors from $\mf_{W}$ to $\mf_{W'}$. For this we would need to map a morphism $\phi_{UV}$ to a natural transformation between $\tilde{\Pi} (U)$ and $\tilde{\Pi} (V)$, which however we can only do for the cohomology categories as we will see later (see Proposition~\ref{prop:superfunctor}). To prepare this construction, we first keep the matrix factorization $Q$ fixed, and consider the map $U\mapsto U(Q)$ that maps fusion functors to matrix factorizations. It can be extended to a functor: the morphisms of fusion functors $U,V$ can be mapped to morphisms between the
matrix factorization $U(Q)$ and $V(Q)$ by applying them to a matrix factorization $Q$:
\begin{Def}
Let $(M,\sigma ,Q)$ be a finite-rank $W$-matrix factorization. Then we define a functor $\tPi^{Q}$ from the category $\FF_{W',W}$ of $(W',W)$-fusion functors to the category $\mf_{W'}$ of finite-rank $W'$-matrix factorizations by
\begin{align}
\tPi^{Q} (U) &= \big(U (M),U (\sigma), U (Q) \big)\\
\tPi^{Q} (\phi^{(n)}_{UV})&= \phi^{(n)}_{UV} (Q,\dotsc ,Q)\ .
\end{align}
\end{Def}
\begin{prop}
$\tPi^{Q}$ is a differential graded functor.
\end{prop}
\begin{proof}
Because of the compatibility conditions~\eqref{gradedcompatibility} of $\phi^{(n)}_{UV}$ with the grading we have
\begin{align}
V (\sigma ) \phi^{(n)}_{UV} (Q,\dotsc ,Q) &= \phi^{(n)}_{UV} (\sigma Q,\dotsc ,Q)\\
&= - \phi^{(n)}_{UV} (Q\sigma ,\dotsc ,Q)\\
&= -\phi^{(n)}_{UV} (Q,\sigma Q,\dotsc ,Q)\\
&=\dotsb \\
&= (-1)^{n}\phi^{(n)}_{UV} (Q,\dotsc ,Q) U (\sigma) \ ,
\end{align}
hence $\tPi^{Q} (\phi^{(n)}_{UV})=\phi^{(n)}_{UV} (Q,\dotsc ,Q)$ has degree $n$.

We now show that the functor is compatible with the differential,
\begin{equation}\label{Pirelation}
\tPi^{Q} \big(d \phi^{(n)}_{UV} \big) =
\delta_{U(Q),V(Q)} \tPi^{Q} (\phi^{(n)}_{UV}) \ .
\end{equation}
Start with the left-hand side
of~\eqref{Pirelation}, 
\begin{align}
\tPi^{Q} \big(d \phi^{(n)}_{UV} \big) &= d\phi^{(n)}_{UV} (Q,\dotsc ,Q)\nonumber\\
& = V(Q) \phi^{(n)}_{UV} (Q,\dotsc ,Q)\nonumber\\
&\quad  - \phi^{(n)}_{UV}(Q^2,Q,\dotsc ,Q)+\dotsb 
+ (-1)^{n} \phi^{(n)}_{UV} (Q,\dotsc ,Q,Q^2)\nonumber\\
&\quad + (-1)^{n+1}\phi^{(n)}_{UV} (Q,\dotsc,Q) U(Q)\ .
\label{lhs}
\end{align}
Consider the second term on the right-hand side,
\begin{equation}
\phi^{(n)}_{UV} (Q^2,Q,\dotsc ,Q) = \phi^{(n)}_{UV} (W\cdot
\one,Q,\dotsc ,Q) = W'\, \phi^{(n)}_{UV} (\one,Q,\dotsc ,Q)
= 0 \ ,
\end{equation}
where in the last step we made use of the fact that we took
$\phi^{(n)}_{UV}$ from the space of normalized morphisms: it vanishes when any of
its arguments is the identity. This reasoning also applies to the other
terms involving $Q^{2}$, and~\eqref{lhs} reduces to
\begin{align}
\tPi^{Q} \big(d \phi^{(n)}_{UV} \big)&=
V(Q) \phi^{(n)}_{UV}(Q,\dotsc ,Q) +(-1)^{n+1}
\phi^{(n)}_{UV} (Q,\dotsc,Q) U(Q)
\nonumber\\
&= \delta_{U(Q),V(Q)}  \tPi^{Q} (\phi^{(n)}_{UV})\ ,
\end{align}
which proves~\eqref{Pirelation}.
\end{proof}
Being a differential graded functor, $\tPi^{Q}$ induces a functor between the cohomology categories,
\begin{equation}\label{PiQ}
\Pi^{Q} : \HFF_{W',W} \to \Hmf_{W'}  \ .
\end{equation}
The functor $\Pi^{Q}$ depends on the fixed factorization $Q$. When we consider two factorizations $Q_1$ and $Q_2$, then a morphism between the factorizations will induce a natural transformation between the functors $\Pi^{Q_{1}}$ and $\Pi^{Q_{2}}$ -- this is the content of the following proposition. When the matrix factorizations $Q_{1}$ and $Q_{2}$ are isomorphic, the proposition implies that there is a natural isomorphism between the functors $\Pi^{Q_{1}}$ and $\Pi^{Q_{2}}$. In particular, the factorizations $U (Q_{1})$ and $U (Q_{2})$ are isomorphic for isomorphic factorizations $Q_1$, $Q_2$. 
\begin{prop}\label{prop:naturaltransformationPiQ}
Let $\varphi$ be an even $\delta_{Q_{1},Q_{2}}$-closed morphism between two $W$-matrix factorizations $Q_{1}$ and $Q_{2}$. For each object $U$ in $\HFF_{W',W}$ we set
\begin{equation}
\pi^{\varphi} (U):= U (\varphi)\ .
\end{equation}
This defines a natural transformation $\pi^{\varphi}$ from $\Pi^{Q_{1}}$ to $\Pi^{Q_{2}}$: for all closed morphisms $\phi^{(n)}_{UV}$ between $(W',W)$-fusion functors $U$ and $V$ we have
\begin{equation}\label{naturalityxi}
V (\varphi)\, \phi^{(n)}_{UV} (Q_{1},\dotsc ,Q_{1}) = \phi^{(n)}_{UV} (Q_{2},\dotsc ,Q_{2})\,U (\varphi) + (\delta_{U (Q_{1}),V (Q_{2})}\text{-exact terms}) \ .
\end{equation}
\end{prop}
The proof can be found in Appendix~\ref{sec:app-proof-naturality}.

\subsection{A monoidal functor}

This subsection continues the discussion of the relation between fusion functors and functors on categories of matrix factorizations. It does not have direct implications on the connection between fusion functors and interface factorizations -- readers wishing to skip this part can therefore directly proceed to Section~\ref{sec:ffdescribefusion}.

In the last subsection, we have discussed the action of fusion functors on matrix factorizations, and we have seen that every fusion functor gives rise to a functor between categories of matrix factorizations. We show here that this structure gives rise to a superfunctor from the cohomology category of fusion functors to the supercategory  of functors between categories of matrix factorizations, and we study its properties, in particular with respect to horizontal composition. 

We first define the supercategory $\mathrm{Fun} (\Hmf_{W},\Hmf_{W'})$ of functors between categories of matrix factorizations.
\begin{Def}
We denote by $\mathrm{Fun} (\Hmf_{W},\Hmf_{W'})$ the supercategory that has as objects $\mathbb{C}$-linear functors from $\Hmf_{W}$ to $\Hmf_{W'}$, and where the morphisms between two functors $U,V$ are supernatural transformations $\phi^{(n)}$ of degree $n=0,1$, which assign to any matrix factorization $Q$ of $W$ a morphism $\phi^{(n)} (Q)$ of degree $n$ from $U(Q)$ to $V(Q)$. For any morphism $\varphi:Q_{1}\to Q_{2}$ of degree $m$, they satisfy
\[
V (\varphi)\phi^{(n)} (Q_1) = (-1)^{nm}\phi^{(n)} (Q_2)U (\varphi)\ .
\]
\end{Def}

We can now introduce a superfunctor from $\HFF_{W',W}$ to $\mathrm{Fun}(\Hmf_{W},\Hmf_{W'})$:
\begin{prop}\label{prop:superfunctor}
For every $(W',W)$-fusion functor $U$, we consider the functor $\Pi (U)$ from $\Hmf_{W}$ to $\Hmf_{W'}$ as in~\eqref{defPiU}. 
For every morphism $[\phi_{UV}] \in H^{n}(\Hom(U,V))$ between fusion functors,
we then define for every $W$-matrix factorization $Q$ the morphism $\Pi^{Q} ([\phi_{UV}])=[\phi_{UV} (Q,\dotsc ,Q)]$ in $\Hmf_{W'}$ between $U (Q)$ and $V (Q)$. This defines a supernatural transformation between $\Pi (U)$ and $\Pi (V)$ (which is even for $n$ even, and odd for $n$ odd), so the following defines a superfunctor:
\begin{equation}
\begin{array}{rl}
\Pi_{W',W} : \HFF_{W',W} & \longrightarrow \mathrm{Fun}(\Hmf_{W},\Hmf_{W'})\\
U & \longmapsto (Q \mapsto U (Q))\\
\mbox{$[ \phi_{UV} ]$} & \longmapsto \big( Q \mapsto [ \phi_{UV} (Q,\dotsc ,Q)] \big)\ .
\end{array}
\end{equation}
\end{prop}
\begin{proof}
From~\eqref{naturalityxi} (naturality of $\pi^{\varphi}$) we conclude that the following diagram is commutative:
\begin{center}
\begin{tikzpicture}[thick, scale=1.2, rotate=0]

\coordinate (A) at (-2,1);
\coordinate (B) at ( 2,1);
\coordinate (C) at (-2,-1);
\coordinate (D) at ( 2,-1);

\draw[->] (-1.5,1)  -- (1.5,1) ;
\draw[->] (-1.5,-1) -- (1.5,-1) ;
\draw[->] (-2,0.7)  -- (-2,-0.7)  ;
\draw[->] (2,0.7)   -- (2,-0.7) ;

\draw (A) node[color=black] {$U (Q_{1})$};
\draw (B)+(0.1,0) node[color=black] {$V (Q_{1})$};
\draw (C) node[color=black] {$U (Q_{2}) $};
\draw (D)+(0.1,0) node[color=black] {$V (Q_{2})$};

\draw (0,1.4) node[color=black] {$[\phi_{UV} (Q_{1},\dotsc ,Q_{1})] $};
\draw (0,-1.3) node[color=black] {$[\phi_{UV} (Q_{2},\dotsc ,Q_{2})]   $};
\draw (-2.5,0) node[color=black] {$U (\varphi)  $};
\draw ( 2.5,0) node[color=black] {$ V (\varphi)$};
\end{tikzpicture}
\end{center}
for an even morphism $\varphi$. Similarly, the diagram is anti-commutative for an odd morphism $\varphi$ if also $n$ is odd. 
\end{proof}
For the category of functors on matrix factorizations, there naturally is a notion of a tensor product (or horizontal composition):
\begin{Def}
We define for functors $F\in \mathrm{Fun}(\Hmf_{W'},\Hmf_{W''})$ and $G\in \mathrm{Fun}(\Hmf_{W},\Hmf_{W'})$
the tensor product
as the composite $F\otimes G = F\,G$. For morphisms $\phi_{F_{i}G_{i}}$ between functors $F_{i}$ and $G_{i}$ ($i=1,2$) we define the tensor product as the supernatural transformation that assigns to a matrix factorization $Q$ the morphism
\begin{equation}
(\phi_{F_{2}G_{2}}\otimes \phi_{F_{1}G_{1}}) (Q) = \phi_{F_{2}G_{2}} \big(G_{1} (Q) \big) \,F_{2}\big(\phi_{F_{1}G_{1}} (Q) \big)\ .
\end{equation}
\end{Def}
For equal superpotentials, $\mathrm{Fun}(\Hmf_{W},\Hmf_{W})$ is a strict monoidal supercategory.
The functor $\Pi_{W'W}$ is compatible with the horizontal composition:
\begin{prop}
The following diagram of functors is commutative:
\begin{center}
\includegraphics[scale=.95]{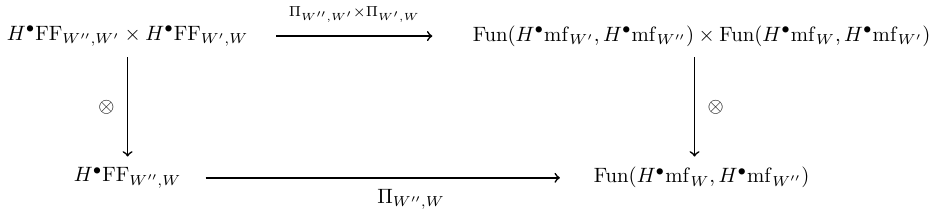}
\end{center}
\end{prop}
\begin{proof}
Applying the functors to an object $(U_{2}, U_{1})$ and following the diagram at the left and at the bottom, we obtain
\begin{align}
\Pi_{W'',W}\big(U_{2}\otimes U_{1} \big) = \Pi_{W'',W} ( U_{2}U_{1}) = \big(Q\mapsto U_{2}U_{1} (Q) \big)\ .
\end{align}
This coincides with the result of following the diagram at the top and at the right,
\begin{align}
\Pi_{W'',W'} (U_{2})\otimes \Pi_{W',W}(U_1)&= \big(Q'\mapsto U_{2} (Q') \big)\otimes \big(Q\mapsto U_{1} (Q) \big)\\
 &= \big(Q\mapsto U_{2}U_{1} (Q) \big)\ .
\end{align}
Applying the functors to a morphism $(\phi_{U_{2}V_{2}},\phi_{U_{1}V_{1}})$, we find by following the diagram along the left and the bottom arrow
\begin{align}\label{leftbottomresult}
\Pi_{W'',W} \big( \phi_{U_{2}V_{2}}\otimes \phi_{U_{1}V_{1}}\big) &=\Big(Q\mapsto 
\phi_{U_{2}V_{2}}\big(V_{1} (Q),\dotsc ,V_{1} (Q) \big)U_{2}\big(\phi_{U_{1}V_{1}} (Q,\dotsc ,Q) \big) \Big)
\end{align}
where we have used the definition~\eqref{horizontalcompdef} of $\phi_{U_{2}V_{2}}\otimes \phi_{U_{1}V_{1}}$. By following the diagram at the top and at the right side, we obtain
\begin{align}
&\Pi_{W'',W'} (\phi_{U_{2}V_{2}})\otimes  \Pi_{W',W} (\phi_{U_{1}V_{1}})\nonumber\\
&\qquad =
\big(Q'\mapsto \phi_{U_{2}V_{2}} (Q',\dotsc ,Q') \big)\otimes \big(Q\mapsto \phi_{U_{1}V_{1}} (Q,\dotsc ,Q) \big)\\
&\qquad =\Big(Q\mapsto \phi_{U_{2}V_{2}}\big(V_{1} (Q),\dotsc ,V_{1} (Q) \big)U_{2}\big(\phi_{U_{1}V_{1}} (Q,\dotsc ,Q) \big) \Big) \ ,
\end{align}
which coincides with~\eqref{leftbottomresult}.
\end{proof}
For equal superpotentials, we can then conclude:
\begin{corr}
$\Pi_{W,W} : \HFF_{W,W} \longrightarrow \mathrm{Fun}(\Hmf_{W},\Hmf_{W})$ is a strict monoidal superfunctor.
\end{corr}

\subsection{Fusion functors describe fusion}\label{sec:ffdescribefusion}
We have seen in the previous subsection that a $(W',W)$-fusion functor $U$ gives rise to a functor $\Pi (U)$ that maps  matrix factorizations of $W$ to those of $W'$. As we discuss in the following, this can be interpreted as fusing a $(W',W)$-interface to a boundary condition in $W$, mathematically realized as tensoring a $(W',W)$-matrix factorization to a $W$-matrix factorization.

There is a special defect called the identity defect. Physically, it corresponds to not introducing a defect at all, and fusing such a defect to any interface will not change the interface. A $(W,W)$-identity defect can be realized as a $(W,W)$-matrix factorization $(M_{I_{W}},\sigma_{I_{W}},I_{W})$ \cite{Khovanov:2004,Brunner:2007qu}. This means that for any $W$-matrix factorization $Q$ there exist
closed (even) homomorphisms
\begin{equation}\label{lambdaandinverse}
\lambda_{Q}: M_{I_{W}}\otimes_R M_{Q} \to M_{Q} \ ,\quad 
\lambda^{-1}_{Q}: M_{Q}\to M_{I_{W}}\otimes_R M_{Q} \ ,
\end{equation}
such that
\begin{subequations}\label{inverses}
\begin{align}
\lambda_{Q}\lambda_{Q}^{-1} &= \one_{M_{Q}} +
(\delta_{Q,Q}\text{-exact terms}) \label{inversea}\\
\lambda_{Q}^{-1}\lambda_{Q} & = \one_{M_{I_{W}}\otimes_R M_{Q}} +
(\delta_{I_{W}\boxtimes Q,I_{W}\boxtimes Q}\text{-exact terms})\ .\label{inverseb}
\end{align}
\end{subequations}
This realizes an isomorphism between the matrix factorizations $I_{W}\boxtimes Q$ and $Q$.\footnote{Note that $M_{I_{W}}\otimes_R M_{Q}$ are not of finite rank as $R$-modules, so this isomorphism holds in the category $\HMF$ of matrix factorizations of not necessarily finite rank.} The construction can be straightforwardly extended
to $(W,W'')$-matrix factorizations $Q$ for some $W''\in R''$ by treating the second factor trivially.

We are now ready to show that the action of a fusion functor $U$ on matrix factorizations has the same effect as tensoring with the interface factorization $U(I_{W})$. For this we need to extend $U$ to a functor on the category of free $(R\otimes_{\mathbb{C}}\tilde{R})$-modules by acting trivially on the second factor. More concretely, we can for any finite-rank $(R\otimes_{\mathbb{C}}\tilde{R})$-module $N$ fix an isomorphism $\xi_{N}:N\to M\otimes_{\mathbb{C}} \tilde{M}$ to a tensor product of a finite-rank $R$-module $M$ and a finite-rank $\tilde{R}$-module $\tilde{M}$.  Then we set  $U (N)=U (M\otimes_{\mathbb{C}} \tilde{M})$, and also for homomorphisms $h:N_{1}\to N_{2}$ we set $U (h)=U (\xi_{N_{2}}\,h\,\xi_{N_{1}}^{-1})$ where $U$ on these tensor products is defined as
\begin{align}
U ( M\otimes_{\mathbb{C}} \tilde{M}) &= U (M)\otimes_{\mathbb{C}} \tilde{M}\\
U (f\otimes g) &= U (f) \otimes g\\
\phi_{UV} (f_{n}\otimes g_{n},\dotsc ,f_{1}\otimes g_{1}) &= \phi_{UV} (f_{n},\dotsc ,f_{1})\otimes (g_{n}\dotsb g_{1})\ .
\end{align}
This extension depends on the precise identifications $N\cong M\otimes_{\mathbb{C}} \tilde{M}$, but different extensions result in isomorphic functors (see Appendix~\ref{sec:app-bimodules}). In a similar fashion, we can extend $U$ to a functor on the category of all free $R$-modules (not necessarily of finite rank). For this we identify every infinite-rank free $R$-module $M\cong M_{f}\otimes_{\mathbb{C}}V_{M}$ as a tensor product of a finite-rank free $R$-module $M_{f}$ and an infinite-dimensional vector space $V_M$, and then set $U (M)=U (M_{f})\otimes_{\mathbb{C}}V_M$.

In analogy to the functor $\Pi^Q$ (see~\eqref{PiQ}), we can then introduce the functor $\Pi^{I_W}$ that evaluates fusion functors and morphisms on the identity factorization:
\begin{Def}\label{def:PiIW}
The functor $\Pi^{I_W}:\HFF_{W',W}\longrightarrow \Hmf_{W',W}$ maps a $(W',W)$-fusion functor $U:R\fmod\to R'\fmod$ to the associated interface factorization
\begin{equation}
\Pi^{I_W}(U)=U(I_W)\ .
\end{equation}
A degree-$n$ morphism $\phi^{(n)}_{UV}$ is mapped to
\begin{equation}
\Pi^{I_W}(\phi^{(n)}_{UV})=\phi^{(n)}_{UV}(I_W,\dots, I_W) \ .
\end{equation}
\end{Def}

As we show in the following, fusing $U(I_W)$ to another factorization has the same effect as applying the functor $U$ to the factorization. Indeed, when we apply $U$ to the tensor product factorization $I_{W}\boxtimes Q=I_{W}\otimes \one + \sigma_{I_{W}}\otimes Q$, we obtain
\begin{equation}
U (I_{W}\otimes \one + \sigma_{I_{W}}\otimes Q) \cong U (I_{W})\otimes \one + U (\sigma_{I_{W}})\otimes Q = U (I_{W})\boxtimes Q \ .
\end{equation}
This is isomorphic to the factorization that is obtained by applying $U$ to $Q$:
\begin{prop}
Let $(M_{I_{W}},\sigma_{I_{W}},I_{W})$ be an identity matrix factorization, and $Q$ a $W$-matrix factorization. Then for all $(W',W)$-fusion functors $U$ 
\begin{equation}\label{isomorphism}
U (I_{W})\boxtimes Q \cong U (Q) \quad \text{in}\ \HMF_{W'}\ .
\end{equation}
\end{prop}
\begin{proof}
By applying $U$ to the relations~\eqref{inverses} we verify that $U(\lambda_{Q})$ and $U
(\lambda_{Q}^{-1})$ implement the isomorphism~\eqref{isomorphism}.
\end{proof}
A morphism $\phi_{UV}$ induces a morphism $\varphi =\phi_{UV} ( I_{W},\dotsc ,I_{W})$ between the factorizations $U (I_{W})$ and $V (I_{W})$. Taking the tensor product with the factorization $Q$ one induces the morphism $\varphi \otimes \one$ between $U (I_{W})\boxtimes Q$ and $V (I_{W})\boxtimes Q$. This corresponds to the morphism $V (\lambda_{Q}) (\varphi \otimes \one)U (\lambda_{Q}^{-1})$ between $U (Q)$ and $V (Q)$. It is equivalent to $\phi_{UV} (Q,\dotsc ,Q)$ as the following proposition shows.
\begin{prop}
Let $\lambda_{Q},\lambda_{Q}^{-1}$ be the homomorphisms from~\eqref{lambdaandinverse}. Then for all morphisms $\phi_{UV}$ between $(W',W)$-fusion functors $U,V$ we have
\begin{subequations}\label{naturalityUIQ-UQ}
\begin{align}
V (\lambda_{Q})\big(\phi_{UV} (I_{W},\dotsc ,I_{W})\otimes \one \big) &= \phi_{UV} (Q,\dotsc ,Q)\,U (\lambda_{Q})+\text{exact terms}\\
\big(\phi_{UV} (I_{W},\dotsc ,I_{W})\otimes \one \big)U (\lambda_{Q}^{-1}) &= V (\lambda_{Q}^{-1})\, \phi_{UV} (Q,\dotsc ,Q)+\text{exact terms} \ .
\end{align}
\end{subequations}
\end{prop}
\begin{proof}
The relations~\eqref{naturalityUIQ-UQ} are proven analogously to the proof of proposition~\ref{prop:naturaltransformationPiQ}. In fact, they are essentially the statement that $U\mapsto U (\lambda_{Q}^{-1})$ and $U\mapsto U (\lambda_{Q})$ define natural isomorphisms between $\Pi^{Q}$ and $\Pi^{I_{W}\boxtimes Q}$. Here, we can use that a degree-$n$ morphism $\phi_{UV}$ acts trivially on the factorization $Q$ appearing in the second factor:
\begin{align}
\phi_{UV} (\sigma_{I_{W}}\otimes Q,\dotsc ,\sigma_{I_{W}}\otimes Q) &= \phi_{UV} (\sigma_{I_{W}},\dotsc ,\sigma_{I_{W}})\otimes Q^{n} \\
&= V (\sigma_{I_{W}}) \phi_{UV} (\one ,\sigma_{I_{W}},\dotsc ,\sigma_{I_{W}})\otimes Q^{n} = 0\ 
\end{align}
because $\phi_{UV}$ vanishes when one of its arguments is the identity. Therefore,
\begin{equation}
\phi_{UV} (I_{W}\boxtimes Q,\dotsc ,I_{W}\boxtimes Q) = \phi_{UV} (I_{W},\dotsc ,I_{W})\otimes \one \ ,
\end{equation}
which explains the appearance of these terms in~\eqref{naturalityUIQ-UQ}.
\end{proof}

The results above can be summarised in the following proposition:
\begin{prop}\label{prop:FFdescribeFusion}
Let $(M_{I_{W}},\sigma_{I_{W}},I_{W})$ be an identity matrix factorization for the superpotential $W$. Recall from Proposition~\ref{prop:superfunctor} the functor 
\begin{equation}
    \Pi_{W',W} : \HFF_{W',W} \to \mathrm{Fun} (\Hmf_{W},\Hmf_{W'})
\end{equation}
whose action on objects is $\Pi_{W',W}(U)=\big(P\mapsto U(P)\big)$. Via the canonical embedding $\Hmf_{W'}\subset \HMF_{W'}$ we may regard $\Pi_{W',W}(U)$ as a functor $\Hmf_W\to \HMF_{W'}$. Then the following diagram of functors commutes, up to natural isomorphisms:
\begin{center}
\begin{tikzpicture}[thick, scale=1.2, rotate=0]

\coordinate (A) at (-3,1);
\coordinate (B) at ( 3,1);
\coordinate (C) at (0,-1);

\draw[->] (-2.2,1)  -- (1.5,1) ;
\draw[->] (.3,-.7)  -- (2.7,.7) ;
\draw[->] (-2.7,.7)  -- (-.3,-.7) ;

\draw (A)+(-.05,0) node[color=black] {$\HFF_{W',W}$};
\draw (B)+(0.3,0) node[color=black] {$\mathrm{Fun} (\Hmf_{W},\HMF_{W'})$};
\draw (C) node[color=black] {$\Hmf_{W',W}$};

\draw (0,1.3) node[color=black] {$\Pi_{W',W}$};
\draw (-1.8,-.2) node[color=black] {$\Pi^{I_{W}}$};
\draw (2.7,-.2) node[color=black] {$\big(Q\mapsto Q \boxtimes \underline{\ \ } \big)$};
\end{tikzpicture}
\end{center}
Here, the functor on the right-hand side sends $Q\in\Hmf_{W',W}$ to the functor whose action on objects in $\Hmf_W$ is $P\mapsto Q\boxtimes P$.
\end{prop}
Note that the functor $\Pi^{I_{W}}$ depends on the concrete realization of the identity defect as a factorization. Different realizations $I_{W}\cong \tilde{I}_{W}$ lead to isomorphic functors, $\Pi^{I_{W}}\cong \Pi^{\tilde{I}_{W}}$ (see Proposition~\ref{prop:naturaltransformationPiQ}).

The $(W',W)$-interfaces $\Pi^{I_{W}} (U)= U (I_{W})$ that arise from applying fusion functors to the identity defect factorization will be called \textsl{operator-like interfaces}, because the effect of fusing such an interface to another matrix factorization is implemented by simply applying the fusion functor to the factorization.

The fusion of operator-like interfaces can be realized by the composition of the corresponding fusion functors:
\begin{prop}\label{prop:PiIWCompatibleWithFusion}
The following diagram commutes up to natural isomorphisms:
\begin{center}
\begin{tikzpicture}[thick, scale=1.2, rotate=0]

\coordinate (A) at (-3,1);
\coordinate (B) at ( 3,1);
\coordinate (C) at (-3,-1);
\coordinate (D) at ( 3,-1);

\draw[->] (-1,1)  -- (1.1,1) ;
\draw[->] (-2.1,-1) -- (2.1,-1) ;
\draw[->] (-3,0.7)  -- (-3,-0.7)  ;
\draw[->] (3,0.7)   -- (3,-0.7) ;

\draw (A) node[color=black] {$\HFF_{W'',W'}\times \HFF_{W',W}$};
\draw (B)+(0.1,0) node[color=black] {$\Hmf_{W'',W'}\times \Hmf_{W',W}$};
\draw (C) node[color=black] {$\HFF_{W'',W}$};
\draw (D)+(0.1,0) node[color=black] {$\HMF_{W'',W}$};

\draw (0.2,1.3) node[color=black] {$\Pi^{I_{W'}}\times \Pi^{I_{W}}$};
\draw (0,-1.3) node[color=black] {$\Pi^{I_{W}}$};
\draw (-3.5,0) node[color=black] {$\otimes $};
\draw ( 3.5,0) node[color=black] {$\boxtimes$};
\end{tikzpicture}
\end{center}
\end{prop}
\begin{proof}
This follows straightforwardly from the previous results.
\end{proof}

\begin{ex}\label{ex:topdefectsinminimalmodels}
Consider $R=R'=\mathbb{C}[x]$ and $W=W'=x^{k}$. This Landau--Ginzburg model corresponds to a superconformal minimal model in the infrared~\cite{Vafa:1988uu}. The invertible functor $G_{\xi}$ from Example~\ref{ex:invertiblefunctors} for $\xi^{k}=1$ is a fusion functor that implements the effect of fusing a symmetry defect $G_{\xi} (I_{W})$~\cite{Brunner:2007qu}. The cone~\eqref{coneexample} we can construct from a degree-1 morphism between $G_{\xi_{1}}$ and $G_{\xi_{2}}$ (see Example~\ref{ex:degree1morphism}) is a fusion functor that describes a cone of a morphism between symmetry defect factorizations. It is isomorphic to a  defect factorization of $x^{k}-x'^{k}$ based on the factors $(x-\xi_{1}x') (x-\xi_{2}x')$ and $\frac{x^{k}-x'^{k}}{(x-\xi_{1}x') (x-\xi_{2}x')}$. For specific pairs of $\xi_{1}$ and $\xi_{2}$ these correspond to rational topological defects~\cite{Brunner:2005fv}. In such a way, one can construct fusion functors for any rational topological defects in minimal models.
\end{ex}

\begin{ex}\label{ex:KR}
Let $R=\mathbb{C}[y_{1},y_{2}]$, $R'=\mathbb{C}[x_{1},x_{2}]$, and
$W'=x_{1}^{m}+x_{2}^{m}$. Consider the homomorphism $\Omega :R\to R'$ given by
$\Omega (y_{1})=x_{1}+x_{2}$ and $\Omega (y_{2})=x_{1}x_{2}$.
Besides extension of scalars $\Omega^*:R\fmod\to R'\fmod$ we use the
restriction-of-scalars functor $\Omega_*:R'\fmod\to R\fmod$, which views an
$R'$-module as an $R$-module via $\Omega$. Consider the endofunctor
\begin{equation}
F:=\Omega^*\circ \Omega_*:R'\fmod\longrightarrow R'\fmod .
\end{equation}
Its concrete action on $R'$-modules is
\begin{equation}
F(M) = R'\otimes_R M,
\end{equation}
where $R$ acts on $R'$ and on $M$ via $\Omega$.

We introduce two degree-$0$ morphisms (natural transformations). The first is a
morphism $\varepsilon$ from $F$ to $\id$ given by multiplication,
\begin{equation}
\varepsilon(M):R'\otimes_R M\to M\ ,\quad a\otimes m\longmapsto a\cdot m \ .
\end{equation}
The second is a morphism $\eta$ from $\id$ to $F$ defined by
\begin{equation}
\eta(M): M \to R'\otimes_R M\ ,\quad
m \longmapsto \Delta\otimes m + 1\otimes (\Delta m)\ ,
\quad \Delta:=x_1-x_2 .
\end{equation}
This is well defined and $R'$-linear since $(a\otimes 1)e=e(1\otimes a)$ for all
$a\in R'$, where $e:=\Delta\otimes 1 + 1\otimes \Delta \in R'\otimes_R R'$. Note that $R'$ is free of rank $2$ as an $R$-module with basis $\{1,\Delta\}$.

Applying $\Pi^{I_{W'}}$ to $F$ yields an operator-like interface matrix
factorization, and $\varepsilon$ and $\eta$ induce morphisms between this
factorization and the identity factorization. Functors of the form
$R'\otimes_R(-)$ together with such maps are closely related to the algebraic
building blocks that appear in the matrix factorization approach to
Khovanov--Rozansky link homology~\cite{Khovanov:2004}; we do not pursue the
explicit identification here.
\end{ex}

Let us summarize the action of operator-like interfaces $U(I_W)$ on boundary (or interface) matrix factorizations. Conceptually, operator-like interfaces form a class closed under fusion: fusing two such interfaces corresponds to composition of the underlying fusion functors. Moreover, the functorial framework provides functorial representatives for a class of interface fields and their vertical/horizontal compositions, which under $\Pi^{I_W}$ induce the standard operations on the matrix-factorization side. The following diagrams focus on the fusion of an operator-like interface (with or without a field insertion) to a boundary/interface factorization (again with or without a field insertion). Since the interfaces considered here are topological, they can be freely deformed on the worldsheet unless operator insertions are hit; in particular we may slide $U(I_W)$ to the boundary (or onto another interface) without changing correlators. The arrow in the following diagrams denotes this fusion operation (composition of interfaces), i.e.\ it maps a factorization $Q$ to the fused factorization
\begin{equation}
Q \longmapsto U(I_W)\boxtimes Q \cong U(Q)\ .
\end{equation}
We label the interfaces $U(I_W)$ by the fusion functor $U$; as before, we use circled symbols to distinguish whether an interface is labelled by a functor or a factorization.
Factorizations in the diagrams are always understood up to isomorphism in the cohomology category of matrix factorizations.
\begin{itemize}\label{listoffigures}
\item \begin{minipage}[t]{6.8cm}
An operator-like interface $U (I_{W})$ (represented by the fusion functor $U$) fused to an interface factorization $Q$ results in the interface factorization $U (Q)$.
\end{minipage}\hspace*{8mm} \begin{minipage}[t]{6cm}\vspace*{-.3cm}
\scalebox{.9}{\includegraphics{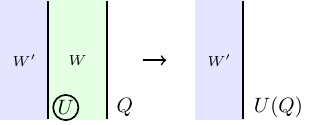}}
\end{minipage}
\item \begin{minipage}[t]{6.8cm}A configuration of operator-like interfaces $V (I_{W})$ and $U (I_{W})$ supporting a morphism $\phi_{UV} (I_{W}\listdots I_{W})$ at their contact point can be fused to $Q$. It results in the factorizations $V (Q)$ and $U (Q)$ with the morphism $\phi_{UV} (Q\listdots Q)$ in between.\end{minipage}\hspace*{8mm} \begin{minipage}[t]{6cm}\vspace*{-.3cm}
\scalebox{.9}{\includegraphics{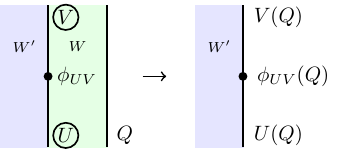}}
\end{minipage}
\item \begin{minipage}[t]{6.8cm}An operator-like interface $U ( I_{W})$ fused to a configuration of interface (boundary) factorizations $Q_{2}$ and $Q_{1}$ joined by a morphism $\varphi$ results in the factorizations $U (Q_{2})$ and $U (Q_{1})$ joined by the morphism $U (\varphi)$.\end{minipage}\hspace*{8mm}\begin{minipage}[t]{6cm}\vspace*{-.3cm}
\scalebox{.9}{\includegraphics{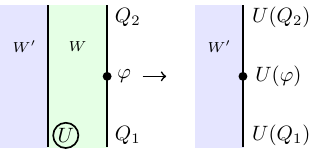}}
\end{minipage}
\item We can also combine the processes above in two different orders as depicted below. We arrive at configurations with the interfaces $V (Q_{2})$ and $U (Q_{1})$ connected by an intermediate interface. In one case it is the interface $U (Q_{2})$ connected by the morphisms $\phi_{UV} (Q_{2}\listdots Q_{2})$ and $U (\varphi)$, in the other case, the linking interface is $V (Q_{1})$ connected by $V (\varphi)$ and $\phi_{UV} (Q_{1}\listdots Q_{1})$.
\begin{center}\scalebox{.9}{\includegraphics{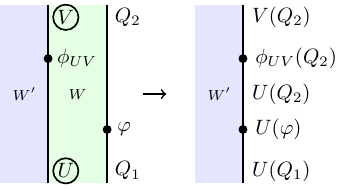}}\ \ \ \ \ \ \ \ \ \ \ \ \scalebox{.9}{\includegraphics{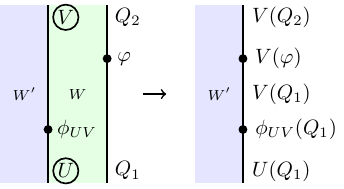}}\end{center}
Composing the connecting morphisms (which means to shrink the intermediate piece) results in a morphism between $V (Q_{2})$ and $U ( Q_{1})$ that agrees in both cases up to a sign if both $\phi_{UV}$ and $\varphi$ have odd degree,
\begin{equation}
\phi_{UV} (Q_{2}\listdots Q_{2})U (\varphi) = (-1)^{nm} V (\varphi ) \phi_{UV} (Q_{1}\listdots Q_{1}) +\text{exact terms}\ ,
\end{equation}
where $n$ is the degree of $\phi_{UV}$ and $m$ the degree of $\varphi$. For even morphisms $\varphi$ this follows from~\eqref{naturalityxi}, the proof for odd $\varphi$ is analogous.
\end{itemize}

\section{Relation between interfaces and fusion functors}\label{sec:relationinterfacesandfusionfunctors}
In the last section, we have seen that a fusion functor $U$ can be related to the interface factorization $U (I_{W})$. This gives rise to a functor $\Pi^{I_{W}}$ from $\HFF_{W',W}$ to $\Hmf_{W',W}$ (see Definition~\ref{def:PiIW}) that maps fusion functors and their morphisms to interface factorizations and their morphisms (interface fields). To better understand the relation of these categories, we need to discuss properties of this functor:
\begin{itemize}
\item In general, $\Pi^{I_{W}}$ is not essentially surjective, which means that there are interface factorizations which cannot be realized in terms of a fusion functor. An extreme example that illustrates this is the situation of a boundary (which can be thought of as an interface between a theory with nontrivial superpotential $W'$ and a trivial theory with $R=\mathbb{C}$ and $W=0$). For a $(W',0)$-fusion functor $U$, we would have 
\begin{equation}
0=U (0\cdot \one )=U (W\cdot \one ) = W' \cdot U (\one )=W' \cdot \one \not=0 \ ,
\end{equation}
a contradiction from which we conclude that no such fusion functors exist. 
\item $\Pi^{I_{W}}$ being a functor implies that isomorphic fusion functors $U\cong \tilde{U}$ are mapped to isomorphic matrix factorizations $U (I_{W})\cong \tilde{U} (I_{W})$. The converse is not true: there are non-isomorphic fusion functors $U,\tilde{U}$ that are mapped to isomorphic matrix factorizations (for an example, see Appendix~\ref{sec:nonisoFF}).
\item $\Pi^{I_{W}}$ is not injective on the morphism spaces (i.e.\ not a faithful functor). This can be most easily seen when we apply it to the identity functors $U=V= \id$. Then $\Pi^{I_{W}} (\id)=I_{W}$ is just the identity factorization, and the space of even (bosonic) morphisms between $I_{W}$ and $I_{W}$ is isomorphic~\cite{Khovanov:2004,Kapustin:2004df,Brunner:2007qu} to the Jacobi ring 
\begin{equation}
J_{W} = \frac{\mathbb{C}[x_{1},\dotsc ,x_{d}]}{\langle\partial_{1}W,\dotsc ,\partial_{d}W\rangle}\quad \text{for}\ R=\mathbb{C}[x_{1},\dotsc ,x_{d}]\ ,
\end{equation}
which is finite-dimensional for suitable $W$. On the other hand, the space of degree-$0$ morphisms between $U=V=\id$ is isomorphic to $R$: recall from Proposition~\ref{prop:spectrumid} that the space of degree-$0$ morphisms for the identity functor is isomorphic to $R$; this also holds for fusion functors because the extra condition~\eqref{morphismfusionfunctors} is empty for degree-$0$ morphisms (see also the discussion in Section~\ref{sec:Idefect} below).
\item As we show below in Section~\ref{sec:onevariable}, $\Pi^{I_{W}}$ is surjective on the space of morphisms (i.e.\ a full functor) in the case of one variable, $R=\mathbb{C}[x]$. It can also be shown to be surjective on the space of odd morphisms in the case of two variables (see Section~\ref{sec:morevariables}). It is not clear to us whether it is surjective in general.
\item The functor $\Pi^{I_{W}}$ is compatible with taking cones,
\begin{equation}\label{compatibilityConeandPi}
\Pi^{I_{W}} \big(\mathcal{C} (\phi_{UV}^{(1)})  \big) =\mathcal{C} \big( \Pi^{I_{W}} (\phi_{UV}^{(1)}) \big) =
\begin{pmatrix}
U (I_{W}) & 0 \\
\phi_{UV}^{(1)} (I_{W}) & V (I_{W})
\end{pmatrix}
 \ ,
\end{equation}
as the resulting matrix factorization on the right has the structure of a cone (see~\eqref{mfcone}).
For polynomial rings in one or two variables, the functor $\Pi^{I_{W}}$ is surjective on odd morphisms (as explained in Sections~\ref{sec:onevariableoddmorphisms} and~\ref{sec:morevariables}). Because the higher odd-degree morphism spaces are trivial in this case, every odd morphism between $U (I_{W})$ and $V (I_{W})$ has a degree-$1$ morphism between $U$ and $V$ as a pre-image. Hence, we can realize any cone on $U (I_{W})$ and $V (I_{W})$ as an operator-like interface for $R=\mathbb{C}[x]$ and $R=\mathbb{C}[x_{1},x_{2}]$.
\end{itemize}

In this section, we explore the relation between interfaces and fusion functors in more detail. We start with a comparison of the spectrum of morphisms for the identity functor and the identity defect, and realize the bulk-interface map in the language of fusion functors. Subsequently, we show in Section~\ref{sec:gradedMF} that also the grading of matrix factorizations has an analogue for fusion functors. Finally, Sections~\ref{sec:onevariable} and~\ref{sec:morevariables} discuss concretely the surjectivity of $\Pi^{I_{W}}$ on the space of morphisms in the case of one variable and on the space of odd morphisms in the case of two variables.

\subsection{Identity defect and bulk-interface map}\label{sec:Idefect}
The identity defect is the simplest operator-like interface, and it is described by the identity functor, $U=\id$. Bulk fields can be represented as endomorphisms of the
identity defect; in particular, the space of degree-$0$ endomorphisms is isomorphic
\cite{Khovanov:2004,Kapustin:2004df,Brunner:2007qu} to the Jacobi ring 
\begin{equation}
H^{0}\big(\Hom_{\MF} (I_{W},I_{W})\big)\cong J_{W}=\frac{R}{\langle \partial_{i}W\rangle}\ .
\end{equation}
On the level of fusion functors, a degree-$0$ morphism $\phi^{(0)}$ from
$\id$ to itself is given by multiplication with a polynomial $p\in R$, i.e.\
$\phi^{(0)}(M)=p\cdot \one_M \in \Hom_R (M,M)$. The condition~\eqref{closuredegreezero} is empty in this case, and hence we obtain
\begin{equation}
H^{0}\big(\Hom^{W,W}  (\id ,\id)\big) \cong R \ .
\end{equation}
This $R$ should be regarded as the algebra of polynomial representatives of
degree-$0$ endomorphisms in the functorial framework. Applying $\Pi^{I_W}$
induces on cohomology the canonical quotient map
$R\rightarrow R/\langle\partial_i W\rangle$,
and in particular a surjection
\begin{equation}
H^{0}\big(\Hom^{W,W}(\id,\id)\big)\xrightarrow{\ \Pi^{I_W}\ }
H^{0}\big(\Hom_{\MF} (I_{W},I_{W})\big)\cong J_W \ .
\end{equation}

One can also consider morphisms of higher degrees for the identity functor. As we have discussed in Section~\ref{sec:structuremorphismspaces}, the degree-$1$ morphisms $\phi^{(1)}$ are characterized by specifying $d$ elements $e_{i}=\phi^{(1)} (x_{i})\in R$, where $d$ is the number of variables of the polynomial ring $R=\mathbb{C}[x_1,\dots,x_d]$. For the identity functors there are no further conditions on the $e_{i}$ from closure or exactness with respect to the differential $d$ on the morphism complex, but there is a non-trivial condition from the requirement that $\phi^{(1)}$ is a morphism of $(W,W)$-fusion functors, $\phi^{(1)} (W\, f)=W\, \phi^{(1)} (f)$ (see~\eqref{morphismfusionfunctors}). For a degree-$1$ morphism this condition is equivalent to $\phi^{(1)} (W)=0$ because $\phi^{(1)}$ satisfies the product property~\eqref{productpropertydegree1}. For the identity functors the product property is just the Leibniz rule, and hence $\phi^{(1)}$ acts as a derivation,
\begin{equation}
\phi^{(1)} (f) = \sum_{i=1}^{d}e_{i}\frac{\partial}{\partial x_{i}}f
\end{equation}
with $e_i=\phi^{(1)}(x_i)$.
The space of degree-$1$ morphisms of the identity fusion functor is then
\begin{equation}
H^{1}\big(\Hom^{W,W}  (\id ,\id)\big) \cong \bigg\{(e_{i})\in R^{d}\,\bigm\vert\, \sum_{i=1}^{d}e_{i}\frac{\partial}{\partial x_{i}}W=0 \bigg\}\ .
\end{equation}
We emphasize that this is the degree-$1$ part of the morphism complex in the fusion-functor framework; we do not attempt here to identify it with the standard bulk cohomology of the B-twisted LG model.
In a similar fashion one can analyze morphisms of higher degree.
\smallskip

For general interfaces, we can also analyze those morphisms that correspond to bulk fields which have been moved to the interface. These morphisms can be identified in the language of fusion functors by interpreting bulk fields as morphisms of identity defects and then using horizontal composition. We have to distinguish two situations depending on the direction from which the bulk field approaches the interface.
In the first case, we move a bulk operator represented (in the functorial formalism) by a polynomial $p\in R$ from the right to the interface specified by the fusion functor $U$. We can view $p$ as representing a degree-$0$ morphism $\phi^{(0)}$ from the identity functor to itself by associating $\phi^{(0)} (M)=p\cdot \one_{M}$ (as in the discussion above). Horizontal composition (see~\eqref{horizontalcompdef}) then leads to a degree-$0$ morphism $\phi_{p}$ of $U$, acting by 
\begin{equation}\label{bulkfieldfromright}
\phi_{p} (M) = U (p\cdot \one_{M})\in \Hom_{R'} (U (M),U (M)) \ .
\end{equation}
This is illustrated in Figure~\ref{fig:bulkfieldfromright}.
Similarly, we can move a bulk field given by $p'\in R'$ from the other side to the interface. Again, we view the morphism as sitting on an identity defect $I_{W'}$ that we move to $U$ from the left. Horizontal composition results in the degree-$0$ morphism
\begin{equation}\label{bulkfieldfromleft}
_{p'}\phi (M) = p'\cdot \one_{U (M)} \in \Hom_{R'} (U (M),U (M))\ , 
\end{equation}
which is illustrated in Figure~\ref{fig:bulkfieldfromleft}.

\begin{figure}
\begin{center}
\scalebox{.9}{\includegraphics{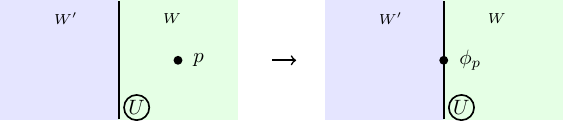}}
\end{center}
\caption{\label{fig:bulkfieldfromright}When a bulk field with polynomial representative $p\in R$ (in the functorial framework) is moved to the interface, it results in the degree-$0$ endomorphism $\phi_p$ of $U$ given by $\phi_p(M)=U(p\cdot \one_M)$, see~\eqref{bulkfieldfromright}.}
\end{figure}
\begin{figure}
\begin{center}
\scalebox{.9}{\includegraphics{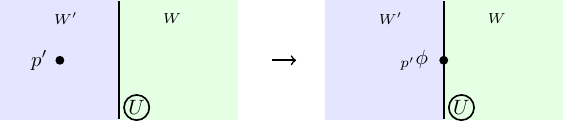}}
\end{center}
\caption{\label{fig:bulkfieldfromleft}When a bulk field with polynomial representative $p'\in R'$ (in the functorial framework) is moved to the interface, it induces the degree-$0$ endomorphism $_{p'}\phi$ of $U$ given by $_{p'}\phi(M)=p'\cdot \one_{U(M)}$, see~\eqref{bulkfieldfromleft}.}
\end{figure}

In this way, the functorial formalism provides a notion of bulk-to-interface transport compatible with fusion; upon applying $\Pi^{I_W}$ it reduces to the usual matrix factorization description where bulk fields are taken modulo the Jacobi relations.

If we have $U (p\cdot \one_M)=p'\cdot \one_{U (M)}$ for some polynomials $p\in R$, $p'\in R'$, it means that the bulk field can be moved through the interface where it changes from $p$ to $p'$. An example for this are fusion functors $U=\omega^{*}$ that arise from a ring homomorphism $\omega:R\to R'$ as the extension of scalars (see Example~\ref{ex:invertiblefunctors}). Here, $p$ and $\omega(p)$ are glued at the interface (i.e.\ they induce the same endomorphism after transport). Note that the superpotentials $W$ and $W'$ are always glued in this sense,
\begin{equation}
    \phi_W (M)= U(W\cdot \one_M) = W'\cdot U(\one_M) = W'\cdot \one_{U(M)} = \, _{W'}\phi (M)\ .
\end{equation}

\begin{ex}\label{ex:KSMM}
Let $R=\mathbb{C}[y_{1},y_{2}]$, $R'=\mathbb{C}[x_{1},x_{2}]$, and $W'=x_{1}^{m}+x_{2}^{m}$. Let $\Omega:R\to R'$ be the homomorphism introduced in Example~\ref{ex:KR}, and let $W\in R$ be the polynomial determined by $\Omega (W)=W'$.
We consider the fusion functor $\Omega^{*}$ (extension of scalars). The Landau--Ginzburg model with superpotential $W'$ flows in the infrared to a product of two minimal models~\cite{Vafa:1988uu}, while the model with superpotential $W$ corresponds~\cite{Gepner:1988wi,Lerche:1989uy} to an $SU (3)/U (2)$ Kazama--Suzuki model~\cite{Kazama:1988uz,Kazama:1989qp}. The fusion functor $\Omega^{*}$ therefore has the interpretation of an interface between a Kazama--Suzuki model and a product of two minimal models~\cite{Behr:2012xg,Behr:2014bta,Behr:2020gqw}.
Any bulk field labelled by a symmetric polynomial $p' (x_{1},x_{2})\in R'$ can be moved through the interface, where it becomes the bulk field labelled by $p\in R$ such that $\Omega (p)=p'$.
\end{ex}

\subsection{Graded matrix factorizations and fusion functors}\label{sec:gradedMF}

Often one is interested in superpotentials $W\in R=\mathbb{C}[x_{1},\dots ,x_{d}]$ that are quasi-homogeneous,\footnote{Such superpotentials correspond in the infrared to superconformal field theories.} i.e.\ there exists an assignment of rational charges $q_{j}\in\mathbb{Q}$ to the
variables $x_{j}$ such that
\begin{equation}\label{Wquasihom}
W\big\vert_{x_{j}\mapsto e^{i\lambda q_{j}}x_{j}} = e^{2i\lambda} \, W  \quad
\text{for all}\ \lambda \in \mathbb{R} \ .
\end{equation}
When we introduce the ring isomorphism $\omega_{\lambda}:R\to R$ by
\begin{equation}
\omega_{\lambda} (x_{j}) = e^{i\lambda q_{j}} x_{j}\ ,
\end{equation}
the condition~\eqref{Wquasihom} becomes
\begin{equation}
\omega_{\lambda} (W) = e^{2i\lambda}\, W \ .
\end{equation}
Associated to $\omega_{\lambda}$ we define the extension of scalars functor $G_{\lambda}=\omega_{\lambda}^{*}$. We can then consider graded matrix factorizations:
\begin{Def}
A graded matrix factorization $(M,\sigma ,Q,\rho^{Q})$ is a matrix factorization together with a family of isomorphisms $\rho^{Q}_{\lambda}$ between $G_{\lambda} (Q)$ and $e^{i\lambda}Q$,
\begin{equation}\label{eq:UoneRfunct}
G_{\lambda} (Q) = \big( \rho^{Q}_{\lambda} \big)^{-1} e^{i\lambda}
Q \, \rho^{Q}_{\lambda} \ ,
\end{equation}
that satisfy
\begin{equation}\label{propertyrhoQ}
\rho^{Q}_{\lambda_{1}} G_{\lambda_{1}} (\rho^{Q}_{\lambda_{2}}) = \rho^{Q}_{\lambda_{1}+\lambda_{2}}\ .
\end{equation}
\end{Def}
The isomorphisms $\rho^{Q}_{\lambda}$ can be used to define charges of
morphisms, which leads to a $\mathbb{Z}$ grading of the spaces of
morphisms \cite{Hori:2004ja,Walcher:2004tx}.

We can define an analogous concept for fusion functors.
\begin{Def}
A graded fusion functor $(U,\rho^{U})$ is a fusion functor $U$ together with a
family of natural isomorphisms $\rho^{U}_{\lambda}$ between the functors $G_{\lambda} \, U\, G_{-\lambda}$ and $U$,
\begin{equation}
\big(G_{\lambda} \, U\, G_{-\lambda}\big) (f) =
\rho^{U}_{\lambda} (M')^{-1}\, U (f)
\, \rho^{U}_{\lambda} (M) \quad \text{for a morphism}\ f:M\to M' \ ,
\end{equation}
that satisfy for all modules $M$
\begin{equation}
\rho^{U}_{\lambda_{1}} (M)\, G_{\lambda_{1}}\!\big(\rho^{U}_{\lambda_{2}} (M) \big) =
\rho^{U}_{\lambda_{1}+\lambda_{2}} (M)\ .
\end{equation}
\end{Def}
On the morphism space $\Hom^{(W',W)}(U,V)$ between graded fusion functors, one then has an action of $\mathbb{R}$ defined by
\begin{equation}
R_{\lambda} (\phi^{(n)}_{UV}) (f_{n} \listdots f_{1}):= \rho^{V}_{\lambda} (M_{n+1})\,G_{\lambda}\big(\phi^{(n)}_{UV} (G_{-\lambda} (f_{n}) \listdots G_{-\lambda} (f_{1})) \big)\,\rho^{U}_{\lambda} (M_{1})^{-1}\ .
\end{equation}
It is straightforward to check that this action is compatible with the differential,
\begin{equation}
dR_{\lambda} = R_{\lambda}d \ ,
\end{equation}
and that it satisfies
\begin{equation}
R_{\lambda_{1}}R_{\lambda_{2}} = R_{\lambda_{1}+\lambda_{2}}\ .
\end{equation}
In this way one can define a $\mathbb{Z}$-grading on the space of morphisms.\footnote{If $N_{0}$ is the smallest positive integer with $N_{0}q_{i}\in \mathbb{Z}$ for all $x_{i}$, then $\omega_{2\pi N_{0}}$ is the identity and all charges of morphisms multiplied by $N_{0}$ are integers.}

\begin{ex}
Let $R=\mathbb{C}[x]$ and $W=x^{n}$, so that we can assign the charge $q=\frac{2}{n}$ to $x$. Consider the morphism $\phi^{(1)}_{G_{\xi_{1}}G_{\xi_{2}}}$ between the functors $G_{\xi_{1}}$ and $G_{\xi_{2}}$ from Example~\ref{ex:degree1morphism} with $\xi_{1}\not = \xi_{2}$. $G_{\xi_{i}}$ are graded fusion functors with $\rho^{G_{\xi_{i}}}_{\lambda}=\one$. Then we have
\begin{equation}
R_{\lambda} (\phi^{(1)}_{G_{\xi_{1}}G_{\xi_{2}}}) = e^{-iq\lambda}\, \phi^{(1)}_{G_{\xi_{1}}G_{\xi_{2}}}\ .
\end{equation}
The morphism $\phi^{(1)}_{G_{\xi_{1}}G_{\xi_{2}}}$ therefore has the rational charge $q$.

The fusion functor $\mathcal{C}=\mathcal{C} (\phi^{(1)}_{G_{\xi_{1}}G_{\xi_{2}}})$ that we obtain as a cone (see~\eqref{coneexample}) from $\phi^{(1)}_{G_{\xi_{1}}G_{\xi_{2}}}$ is again a graded fusion functor: for $f:M\to N$ we have
\begin{align}
\big(G_{\lambda} \, \mathcal{C}\, G_{-\lambda}\big) (f) &= \begin{pmatrix}
G_{\xi_{1}} (f) & 0\\
e^{-i\lambda q} \frac{1}{x} (G_{\xi_{1}}-G_{\xi_{2}}) (f) & G_{\xi_{2}} (f) \end{pmatrix}\\
&=  \begin{pmatrix}
e^{i\lambda q/2} & 0\\
0& e^{-i\lambda q/2} 
\end{pmatrix} \begin{pmatrix}
G_{\xi_{1}} (f) & 0\\
\frac{1}{x}(G_{\xi_{1}}-G_{\xi_{2}}) (f)  & G_{\xi_{2}} (f)
\end{pmatrix} \begin{pmatrix}
e^{-i\lambda q/2} & 0\\
0& e^{i\lambda q/2} 
\end{pmatrix}\ ,
\end{align}
therefore we can associate to $\mathcal{C}$ the natural transformations
\begin{equation}
\rho^{\mathcal{C}}_{\lambda} (M) =\begin{pmatrix}
e^{-i\lambda q/2}\cdot \one  & 0\\
0 & e^{i\lambda q/2}\cdot \one 
\end{pmatrix} \ .
\end{equation}
\end{ex}

A graded fusion functor $(U,\rho^{U})$ maps a graded matrix factorization $(M,\sigma ,Q,\rho^{Q})$ to a graded matrix factorization:
\begin{equation}
(U,\rho_{\lambda}^{U}): (M,\sigma ,Q,\rho_{\lambda}^{Q}) \mapsto \big(U (M),U (\sigma),U (Q), U (\rho_{\lambda}^{Q})\rho^{U}_{\lambda} (M)\big)\ .
\end{equation}
To check that this is indeed a graded matrix factorization, we consider the transformation of $U (Q)$:
\begin{align}
G_{\lambda} \big(U (Q)\big) &= G_{\lambda}\, U\, G_{-\lambda}
\big(G_{\lambda} (Q) \big)\nonumber\\
&= \rho^{U}_{\lambda} (M)^{-1} \, U\big(G_{\lambda} (Q)
\big) \, \rho^{U}_{\lambda} (M) \\
&= \rho^{U}_{\lambda} (M)^{-1}\, U\big(e^{i\lambda}
(\rho^{Q}_{\lambda})^{-1} Q \rho^{Q}_{\lambda}
\big) \, \rho^{U}_{\lambda} (M) \\
&= e^{i\lambda} \rho^{U}_{\lambda} (M)^{-1}U\big(\rho^{Q}_{\lambda}\big)^{-1}\, U(Q) 
\, U\big(\rho^{Q}_{\lambda}\big)\rho^{U}_{\lambda} (M)
\ .
\end{align}
It is also straightforward to check the property~\eqref{propertyrhoQ}:
\begin{align}
& U\big(\rho^{Q}_{\lambda_{1}}\big)\rho^{U}_{\lambda_{1}} (M)\cdot G_{\lambda_{1}}\Big(   U\big(\rho^{Q}_{\lambda_{2}}\big)\rho^{U}_{\lambda_{2}} (M)\Big) \nonumber\\
&\qquad = U\big(\rho^{Q}_{\lambda_{1}}\big)\rho^{U}_{\lambda_{1}} (M)\,\big( G_{\lambda_{1}}UG_{\lambda_{1}^{-1}}\big)\big( G_{\lambda_{1}} (\rho^{Q}_{\lambda_{2}})\big) G_{\lambda_{1}}\big(\rho^{U}_{\lambda_{2}} (M)\big)\\
&\qquad = U\big(\rho^{Q}_{\lambda_{1}}\big)\, U\big(G_{\lambda_{1}} (\rho_{\lambda_{2}}^{Q}) \big)\,\rho_{\lambda_{1}}^{U} (M)\, G_{\lambda_{1}}\big(\rho^{U}_{\lambda_{2}} (M)\big)\\
&\qquad = U (\rho_{\lambda_{1}+\lambda_{2}}^{Q})\,\rho_{\lambda_{1}+\lambda_{2}}^{U} (M)\ .
\end{align}
The composition $U_{2}U_{1}$ of two fusion functors can also be extended to graded fusion functors. For any homomorphism $f:M\to M'$ we have
\begin{align}
&\big(G_{\lambda}U_{2}U_{1}G_{-\lambda} \big) (f) = G_{\lambda}U_{2}G_{-\lambda} \big(G_{\lambda}U_{1}G_{-\lambda} (f) \big)\\
&\qquad = \rho_{\lambda}^{U_{2}}\big(U_{1} (M') \big)^{-1}U_{2}\big(G_{\lambda}U_{1}G_{-\lambda} (f) \big)\rho^{U_{2}}_{\lambda}\big(U_{1} (M) \big) \\
&\qquad = \rho_{\lambda}^{U_{2}}\big(U_{1} (M') \big)^{-1} U_{2}\big(\rho_{\lambda}^{U_{1}} (M')^{-1}U_{1} (f)\rho_{\lambda}^{U_{1}} (M) \big) \rho_{\lambda}^{U_{2}}\big(U_{1} (M) \big)\\
&\qquad = \Big(U_{2}\big(\rho_{\lambda}^{U_{1}} (M') \big)\rho_{\lambda}^{U_{2}}\big(U_{1} (M') \big) \Big)^{-1} U_{2}U_{1} (f)\, U_{2}\big(\rho_{\lambda}^{U_{1}} (M) \big) \rho_{\lambda}^{U_{2}}\big(U_{1} (M) \big)\ .
\end{align}
From this result, it follows that we can define a composition of graded fusion functors (denoted by a tensor product symbol) by
\begin{align}
(U_{2},\rho^{U_{2}}_{\lambda}) \otimes (U_{1},\rho^{U_{1}}_{\lambda}) &= \big(U_{2}U_{1}, \rho^{U_{2}\otimes U_{1}}_{\lambda } \big)\\
\text{with}\ \rho^{U_{2}\otimes U_{1}}_{\lambda} (M)&= U_{2}\big(\rho^{U_{1}}_{\lambda} (M) \big) \rho_{\lambda}^{U_{2}}\big(U_{1} (M) \big) \ .
\end{align}
The composition of $\rho_{\lambda}^{U_{2}}$ and $\rho_{\lambda}^{U_{1}}$ can be viewed as a horizontal composition (as discussed in Section~\ref{sec:horizontal}) of morphisms between functors $G_{\lambda}U_{i}G_{-\lambda}$ and $U_{i}$. 

\subsection{The case of one variable}\label{sec:onevariable}
In this subsection, we consider polynomial rings $R=\mathbb{C}[x]$ in one variable, and show that in this case the functor $\Pi^{I_{W}}$ is surjective on morphisms. Let $W\in R$ be a polynomial. The matrix factorization for the identity defect can be
represented as (see, e.g., \cite{Khovanov:2004,Kapustin:2004df})
\begin{equation}
I_{W} = \begin{pmatrix}
0 & x-\tilde{x}\\
\frac{W (x)-W (\tilde{x})}{x-\tilde{x}} & 0
\end{pmatrix} \ ,
\end{equation}
acting on the graded, free rank-two $(R,R)$-bimodule
\begin{equation}
M_{I_{W}}
= R\otimes_{\mathbb{C}} R
\oplus R\otimes_{\mathbb{C}} R\ \ ,\quad \sigma_{I_{W}}= \begin{pmatrix}
1 & 0\\
0 & -1
\end{pmatrix} \ .
\end{equation}
We represent the variable of the second factor $R$ by $\tilde{x}$. Let 
\begin{equation}
Q=\begin{pmatrix}0 & Q^{(1)}\\ Q^{(0)} & 0 \end{pmatrix}
\end{equation}
be a matrix factorization on the module $M_{Q}=M_{Q,0}\oplus M_{Q,1}$ which we decompose explicitly according to the $\mathbb{Z}_{2}$-grading.
We decompose the module $M_{I_{W}}\otimes_{R}M_{Q}$ as
\begin{equation}
M_{I_{W}}\otimes_{R}M_{Q} = \big( R\otimes_{\bC} M_{Q,0} \oplus
R\otimes_{\bC}
M_{Q,1} \big) \oplus \big( R\otimes_{\bC} M_{Q,0} \oplus R\otimes_{\bC}
M_{Q,1}\big)
\end{equation}
in an even and an odd part. For later considerations, it is useful to realize $M_{Q}\cong R\otimes_{R}M_{Q}$ in a tensor product form.

For a given $Q$ we can give concrete homomorphisms $\lambda_{Q}$ and $\lambda_{Q}^{-1}$ \cite{Carqueville:2009ev,Carqueville:2012st} that realize the isomorphism between $Q$ and $I_{W}\boxtimes Q$:
\begin{align}
\lambda_{Q} &: M_{I_{W}}\otimes_{R} M_{Q} \to R\otimes_{R}M_{Q}\\
\lambda_{Q} &= \begin{pmatrix}
\mu\otimes \one_{M_{Q,0}}   & 0 & 0 & 0\\
0 & 0 & 0 & \mu\otimes \one_{M_{Q,1}} 
\end{pmatrix}\\
\lambda_{Q}^{-1} &: R\otimes_{R} M_{Q} \to 
M_{I_{W}}\otimes_{R} M_{Q}\\
\lambda_{Q}^{-1} &= \begin{pmatrix}
\one & 0 \\
\frac{Q^{(0)}(x)-Q^{(0)}(\tilde{x})}{x-\tilde{x}} & 0\\
0 & \frac{Q^{(1)}(x)-Q^{(1)}(\tilde{x})}{x-\tilde{x}}\\
0& \one
\end{pmatrix}\begin{pmatrix}
\iota_{M_{Q,0}} & 0\\
0 & \iota_{M_{Q,1}}
\end{pmatrix}   \ .\label{deflambdaminusone}
\end{align}
Here, $\mu:R\otimes_{\mathbb{C}}R\to R$ is the multiplication in $R$,
\begin{equation}
\mu (p\otimes q)=pq\ .
\end{equation} 
The definition of $\iota_{M}:R\otimes_{R}M\to R\otimes_{\mathbb{C}}R\otimes_{R}M$ depends on a choice of basis $\{e_{j} \}$ for the free finite-rank module $M$,
\begin{equation}
\iota_{M} \big(p (x)\otimes q (x) e_{j} \big) = p (x)q (x) \otimes 1 \otimes  e_{j}\ .
\end{equation}
Also the expressions $Q(x)$ and $Q(\tilde{x})$ that appear in the definition~\eqref{deflambdaminusone} of $\lambda_{Q}^{-1}$ depend on the chosen basis. After the choice of a basis, we can represent $Q:R\otimes_{R}M_{Q}\to R\otimes_{R}M_{Q}$ as a matrix with polynomial entries,
\begin{equation}
Q (1\otimes e_{i}) = Q_{ji} (x)\otimes e_{j} \ ,
\end{equation}
where a sum over $j$ is understood. Then $Q (x)$ and $Q (\tilde{x})$ denote the $R\otimes_{\mathbb{C}} R$-homomor\-phisms from $R\otimes_{\mathbb{C}} R \otimes_{R}M_{Q}$ to itself satisfying
\begin{align}
Q (x): 1\otimes 1\otimes e_{i}&\mapsto Q_{ji} (x)\otimes 1 \otimes e_{j}\\
Q (\tilde{x}) : 1\otimes 1\otimes e_{i}&\mapsto 1\otimes Q_{ji} (\tilde{x})\otimes e_{j} \ .
\end{align}
The dependence on the basis reflects the fact that $\lambda_{Q}^{-1}$ is not unique. We can use any realization of $\lambda_{Q}^{-1}$, and in the following we assume that we have fixed the basis and chosen $\iota$ and the matrix representations $Q (x)$ and $Q (\tilde{x})$ accordingly. 
\smallskip

We now investigate whether $\Pi^{I_{W}}$ is surjective on morphisms. Let 
$\varphi$ be a morphism between the factorizations $U(I_{W})$ and $V(I_{W})$. We want to construct a morphism $\phi_{UV}$ such that $\Pi^{I_{W}} (\phi_{UV})=\varphi$.

The strategy is as follows. We start with the morphism $\varphi$ and
determine the induced morphism $\varphi_{Q}$ between $U (Q)\cong  U(I_{W})\boxtimes Q$
and $V (Q)\cong  V(I_{W})\boxtimes Q$ for any matrix factorization
$Q$. From this expression we read off a morphism $\phi_{UV}$
between $U$ and $V$ such that $\phi_{UV}(Q,\dotsc ,Q)=\varphi_{Q}$.

\subsubsection{Even morphisms}\label{sec:evenmorphisms}
Let us start with a bosonic morphism $\varphi$ between $U(I_{W})$ and
$V(I_{W})$. It is represented by a $\delta_{U (I_{W}),V (I_{W})}$-closed homomorphism from $U (I_{W})$ to $V (I_{W})$,
\begin{equation}\label{varphiclosed}
V(I_{W}) \, \varphi - \varphi \, U(I_{W})   = 0 \ .
\end{equation}
The induced morphism between $U(I_{W})\boxtimes Q$ and $V(I_{W})\boxtimes
Q$ is $\varphi \otimes \one$. By the maps
$V(\lambda_{Q})$ and $U(\lambda_{Q}^{-1})$ we can determine the
induced morphism between $U(Q)$ and $V(Q)$,
\begin{equation}
\varphi_{Q} = V(\lambda_{Q}) \, \big(\varphi \otimes \one_{M_{Q}} \big)
\, U(\lambda_{Q}^{-1}) \ .
\end{equation}
We write the original morphism as 
\begin{equation}
\varphi =\begin{pmatrix}
\varphi^{(0)} & 0\\
0 & \varphi^{(1)}
\end{pmatrix} \ ,
\end{equation}
where we understand $\varphi^{(i)}$ as maps $\varphi^{(i)}:U(R)\otimes_{\mathbb{C}}R\to V(R)\otimes_{\mathbb{C}}R$.
We get
\begin{align}
\varphi_{Q} &= \begin{pmatrix}
V (\mu)\otimes \one & 0 & 0 & 0\\
0 & 0 & 0 & V (\mu)\otimes \one
\end{pmatrix} \begin{pmatrix}
\varphi^{(0)}\otimes \one & & & \\
& \varphi^{(1)}\otimes \one & & \\
& & \varphi^{(1)}\otimes \one & \\
& & & \varphi^{(0)}\otimes \one
\end{pmatrix} \nonumber\\
&\qquad  \begin{pmatrix}
\one  & 0\\
U \big(\frac{Q^{(0)}(x)-Q^{(0)} (\tilde{x})}{x-\tilde{x}}\big) & 0 \\
0 & U \big(\frac{Q^{(1)} (x)-Q^{(1)} (\tilde{x})}{x-\tilde{x}}\big) \\
0 & \one 
\end{pmatrix} \begin{pmatrix}
U (\iota_{M_{Q,0}}) & 0\\
0 & U (\iota_{M_{Q,1}})
\end{pmatrix}\\
& = \begin{pmatrix}
\big( V(\mu)\otimes \one \big) \, \big( \varphi^{(0)}\otimes \one\big)\, U(\iota_{M_{Q,0}}) & 0\\
0 & \big( V(\mu)\otimes \one \big) \, \big(\varphi^{(0)}\otimes \one\big)\, U(\iota_{M_{Q,1}})
\end{pmatrix} \ .
\end{align}
We observe that $\varphi_{Q}$ is independent of the precise matrix $Q$, only the module
$M_{Q}=M_{Q,0}\oplus M_{Q,1}$ enters. It is then natural to define a
degree-$0$ morphism $\phi_{UV}$ between the functors $U$ and $V$ by
\begin{equation}\label{defofPhi}
\phi_{UV} (R\otimes_{R}M) = 
\big( V(\mu)\otimes \one\big) \, \big( \varphi^{(0)}\otimes \one\big)\, U(\iota_{M}):\ U (R)\otimes_{R}M\to V (R)\otimes_{R}M \ ,
\end{equation}
which reproduces $\varphi_{Q}$ when we apply it to $(M_{Q},\sigma_{Q},Q)$. This definition seems to depend on a chosen basis $\{e_{i} \}$ of $M$ because $\iota_{M}$ does, but we will see in a moment that this is not the case.  Concretely, we have
\begin{align}
U (\iota_{M}) : U (R) \otimes_{R} M &\longrightarrow U (R)\otimes_{\mathbb{C}}R\otimes_{R}M \nonumber\\
u\otimes p (x)e_{j} & \longmapsto U (p (x))u \otimes 1\otimes  e_{j}\ .
\end{align}
On the other hand, the map $V (\mu)$ is independent of the basis,
\begin{align}
V (\mu): V (R)\otimes_{\mathbb{C}}R &\longrightarrow  V (R)\nonumber\\
v \otimes p (\tilde{x}) & \longmapsto V (p(x))v \ .
\end{align}
Evaluating $\phi_{UV} (R\otimes_{R}M)$ on an element $u\otimes e_{i}$, we find
\begin{equation}\label{phiUVonbasis}
\phi_{UV} (R\otimes_{R}M) (u\otimes e_{i}) = V (\mu)\varphi^{(0)} (u\otimes 1) \otimes e_{i}\ .
\end{equation}
Let $\{\tilde{e}_{j} \}$ be another basis with $\tilde{e}_{j}=s_{ij}e_{i}$ where $s_{ij}\in R$ (summation over $i$ is understood). Then $u\otimes \tilde{e}_{j}=U (s_{ij})u\otimes e_{i}$ is mapped to
\begin{equation}
\phi_{UV} (R\otimes_{R}M) (u\otimes \tilde{e}_{j}) =  V (\mu)\varphi^{(0)} (U (s_{ij}) u\otimes 1) \otimes e_{i} \ .
\end{equation}
As we will show shortly, $\varphi^{(0)}$ satisfies for any $p\in R$
\begin{equation}\label{propertyofvarphi}
V (\mu) \varphi^{(0)} \, \big( U (p)\otimes 1\big) = V (p) V (\mu) \varphi^{(0)}\ .
\end{equation}
Therefore we have 
\begin{align}
\phi_{UV} (R\otimes_{R}M) (u\otimes \tilde{e}_{j}) &=  V (s_{ij}) V (\mu)\varphi^{(0)} (u\otimes 1) \otimes e_{i}\\
&=V (\mu)\varphi^{(0)} (u\otimes 1) \otimes s_{ij}e_{i} \\
&=V (\mu)\varphi^{(0)} (u\otimes 1) \otimes \tilde{e}_{j}\ .
\end{align}
Comparing this to~\eqref{phiUVonbasis} we see that the definition~\eqref{defofPhi} is independent of the basis.

We need to prove~\eqref{propertyofvarphi}. For this we have to recall that $\varphi^{(0)}$ is not arbitrary, but part of a homomorphism $\varphi$ that is
closed with respect to $\delta_{U(I_{W}),V(I_{W})}$ (satisfying~\eqref{varphiclosed}).
This implies the equation
\begin{equation}\label{varphiclosedexplicit}
\varphi^{(0)} \, U (x\otimes 1-1\otimes \tilde{x}) = V (x\otimes 1-1\otimes \tilde{x}) \, \varphi^{(1)} \ .
\end{equation}
As before for clarity we denote by $x$ and $\tilde{x}$ the
variables of the first and second factor of $R$, respectively. The above is a map from $U (R)\otimes_{\mathbb{C}}R$ to $V (R)\otimes_{\mathbb{C}} R$. When we apply $V (\mu)$ from the left, we observe that the right-hand side vanishes, and we get
\begin{equation}
0= V (\mu) \varphi^{(0)} \, \big( U (x)\otimes 1 -\one \otimes \tilde{x}\big)\ .
\end{equation}
Rearranging this equation leads to
\begin{align}
V (\mu) \varphi^{(0)} \, \big( U (x)\otimes 1\big) &= V (\mu) \varphi^{(0)} \, (\one \otimes \tilde{x}) \\
&=  V (\mu) (\one \otimes \tilde{x})\,\varphi^{(0)}\\
&= V (x) V (\mu) \varphi^{(0)}\ .
\end{align}
Because of the functorial property of $U$ and $V$, we conclude that this holds when we replace $x$ by any polynomial $p\in R$, which proves~\eqref{propertyofvarphi}.

We still have to check that the definition~\eqref{defofPhi} for $\phi_{UV}$ leads to a morphism that is
$d$-closed. Let $g:M\to N$ be a homomorphism of $R$-modules. Choosing a basis $\{e_{i} \}$ for $M$ and $\{f_{j} \}$ for $N$ we can represent $g$ by a matrix with entries $g_{ij}$ in $R$,
\begin{equation}
g (e_{i}) =  g_{ji}f_{j}\ ,
\end{equation}
where a sum over $j$ is understood.
We now apply $d\phi_{UV}$ to $1\otimes g:R\otimes_{R}M\to R\otimes_{R}N$:
\begin{align}
&d\phi_{UV} (1\otimes g)\nonumber\\
 &\quad = V\big(1 \otimes g\big)\,
\phi_{UV}(M)- \phi_{UV}(N)\, U\big(1 \otimes g \big) \\
&\quad =\big(\one \otimes g \big)\,\big(V (\mu)\otimes \one  \big) \, (\varphi^{(0)}\otimes \one)U (\iota_{M}) - \big(V (\mu)\otimes \one \big)\, (\varphi^{(0)}\otimes \one) U (\iota_{N})\,\big(\one \otimes g \big)\ .\label{dphievaluated}
\end{align}
When we evaluate this on $u\otimes e_{i}$, we obtain
\begin{align}
d\phi_{UV} (1\otimes g) (u\otimes e_{i})
 &= \big(\one \otimes g \big)\,\big(V (\mu)\otimes \one  \big) \, (\varphi^{(0)}\otimes \one)U (\iota_{M}) (u\otimes e_{i}) \nonumber\\
&\qquad - \big(V (\mu)\otimes \one \big)\, (\varphi^{(0)}\otimes \one) U (\iota_{N})\,\big(\one \otimes g \big) (u\otimes e_{i})\\
&= V (\mu)\varphi^{(0)} (u\otimes 1)\otimes g_{ji}f_{j} \nonumber\\
&\qquad - \big(V (\mu)\otimes \one \big)\, (\varphi^{(0)}\otimes \one) U (\iota_{N})\,\big(U (g_{ji})u \otimes f_{j} \big)\\
&= V (g_{ji})V (\mu)\varphi^{(0)} (u\otimes 1)\otimes f_{j} \nonumber\\
&\qquad - V (\mu) \varphi^{(0)}\big(U (g_{ji})u\otimes 1 \big)\otimes f_{j} \ .
\end{align}
This vanishes because of~\eqref{propertyofvarphi}, hence $\phi_{UV}$ is closed. 
\smallskip

We have thus shown by explicit construction that -- for the case of one
variable -- for any bosonic morphism $\varphi$ between $U(I_{W})$ and $V(I_{W})$, there is a closed degree-$0$ morphism $\phi_{UV}$ between the functors $U$ and $V$ such that
\begin{equation}
\phi_{UV} (M_{Q}) = \varphi_{Q}\ .
\end{equation}
In particular, $\phi_{UV} (M_{I_{W}})$ coincides with $\varphi$ up to an exact term as one can check explicitly (see~Appendix~\ref{sec:explicitcheck}).

\subsubsection{Odd morphisms}\label{sec:onevariableoddmorphisms}
We now analyze fermionic morphisms for $R=\mathbb{C}[x]$. Let $\psi$ be a fermionic morphism
between $U(I_{W})$ and $V(I_{W})$, i.e.\ it is a homomorphism
\begin{equation}
\psi = \begin{pmatrix}
0 & \psi^{(1)}\\
\psi^{(0)} & 0 
\end{pmatrix} : U(R)\otimes_{\bC}R\,\oplus \,U(R)\otimes_{\bC}R \to
V(R)\otimes_{\bC}R \, \oplus \, V(R)\otimes_{\bC}R \ , 
\end{equation}
satisfying
\begin{equation}
V(I_{W}) \, \psi+ \psi \, U(I_{W})  = 0 \ .
\end{equation}
It induces a morphism $\psi \otimes \one$ between
$U(I_{W})\boxtimes Q$ and $V(I_{W})\boxtimes Q$ for any matrix
factorization $Q$. Using the maps $U(\lambda_{Q}^{-1})$ and
$V(\lambda_{Q})$ we can determine the induced morphism $\psi_{Q}$
between $U(Q)$ and $V(Q)$,
\begin{align}
\psi_{Q} &= \begin{pmatrix}
V (\mu)\otimes \one  & 0 & 0 & 0\\
0 & 0 & 0 & V (\mu)\otimes \one 
\end{pmatrix} \, \begin{pmatrix}
0 & 0 & \psi^{(1)}\otimes \one & 0 \\
0 & 0 & 0 & \psi^{(0)}\otimes \one \\
\psi^{(0)}\otimes \one & 0 & 0 & 0\\
0 & \psi^{(1)}\otimes \one & 0 & 0
\end{pmatrix} \nonumber\\
&\qquad \, \begin{pmatrix}
\one & 0\\
U\big(\frac{Q^{(0)}(x)-Q^{(0)} (\tilde{x})}{x-\tilde{x}}\big) & 0 \\
0 & U\big(\frac{Q^{(1)} (x)-Q^{(1)} (\tilde{x})}{x-\tilde{x}} \big) \\
0 & \one
\end{pmatrix}\begin{pmatrix}
U (\iota_{M_{Q,0}}) & 0 \\
0 & U (\iota_{M_{Q,1}})
\end{pmatrix}\\
& = \begin{pmatrix}
0&\psi_{Q}^{(1)}\\
\psi_{Q}^{(0)} & 0
\end{pmatrix}\ ,
\end{align}
where
\begin{equation}
\psi_{Q}^{(i)} = (V (\mu)\psi^{(1)}\otimes \one)U\big(\tfrac{Q^{(i)}(x)-Q^{(i)}(\tilde{x})}{x-\tilde{x}} \iota_{M_{Q,i}} \big) \quad (i=0,1)\ .
\end{equation}
This expression suggests defining a morphism $\phi_{UV}$ of degree $1$ between the
functors $U$ and $V$ in such a way that a homomorphism $f:R\otimes_{R}M\to R\otimes_{R}N$ is mapped to the homomorphism
\begin{equation}
\phi_{UV}(f) = \big(V(\mu) \psi^{(1)}\otimes \one \big)
\, U\Big(\tfrac{f(x)-f(\tilde{x})}{x-\tilde{x}}\iota_{M} \Big) : 
U(R)\otimes_{R}M\to V(R)\otimes_{R}N \ .
\end{equation}
Note that the definition depends on the choice of a basis for $M$ through $\iota_{M}$. Because the morphism $\phi_{UV}$ is closed (as we check below), the cohomology class of $\phi_{UV}$ does not depend on the choice of basis (which is a consequence of Lemma~\ref{lem:phideterminedbyRendomorphisms}).

We now check that $\phi_{UV}$ is closed with respect to
$d$. For two homomorphisms $f:M\to N$, $g:N\to O$ we have (to declutter notation we
abbreviate $\psi^{(1)}\otimes \one$ by $\psi^{(1)}$ in the
following and also suppress the map $\iota$)
\begin{align}
d\phi_{UV}(g,f) &= V(g) \, \phi_{UV}(f)- \phi_{UV}(g\, f) +\phi_{UV}(g)\, U(f)
\nonumber\\
&= V(g)\, V(\mu) \, \psi^{(1)}\, 
U\Big(\tfrac{f(x)-f(\tilde{x})}{x-\tilde{x}} \Big) - 
V(\mu)\, \psi^{(1)}\, 
U\Big(\tfrac{g (x)\, f(x)-g (\tilde{x})\, f(\tilde{x})}{x-\tilde{x}}\Big)\nonumber\\
& \quad + V(\mu)\, \psi^{(1)}\, 
U\Big(\tfrac{g(x)-g(\tilde{x})}{x-\tilde{x}} \Big)\, U(f)\\
& = V(g)\, V (\mu) \, \psi^{(1)} \, 
U\Big(\tfrac{f(x)-f(\tilde{x})}{x-\tilde{x}} \Big)\nonumber\\
&\quad -V(\mu) \, \psi^{(1)}\, U\Big(g(\tilde{x})
\tfrac{f(x)-f (\tilde{x})}{x-\tilde{x}} \Big)\\
& = 0 \ .
\end{align}
We also need to check that $\phi_{UV}$ has the right behaviour with respect
to multiplying the morphism $f$ by the superpotential $W$,
\begin{equation}\label{WlinearityofPsi}
\phi_{UV} (W\, f) = W'\, \phi_{UV} (f) \ .
\end{equation}
To show this we have to recall that $\psi^{(1)}$ is part of a homomorphism $\psi$ that is closed
with respect to $\delta_{U(I_{W}),V(I_{W})}$. In particular, this means
that
\begin{equation}
V(x-\tilde{x})\, \psi^{(0)}+\psi^{(1)}\, U\Big(\tfrac{W(x)-W(\tilde{x})}{x-\tilde{x}} \Big)
 = 0\ .
\end{equation}
When we apply $V(\mu)$ from the left, the first term vanishes so that
we obtain the relation
\begin{equation}
V (\mu) \, \psi^{(1)} \, 
U\Big(\tfrac{W(x)-W(\tilde{x})}{x-\tilde{x}} \Big) = 0\ .
\end{equation}
Using this result we can now show that~\eqref{WlinearityofPsi} holds:
\begin{align}
\phi_{UV} (W\, f) &= V(\mu) \, \psi^{(1)} \, 
U\Big(\tfrac{W(x)f(x)-W(\tilde{x})f(\tilde{x})}{x-\tilde{x}}\Big)\nonumber\\
&= V(\mu) \, \psi^{(1)} \, 
U\Big(\tfrac{W(x)-W(\tilde{x})}{x-\tilde{x}} f(x) +
W(\tilde{x})\tfrac{f(x)-f(\tilde{x})}{x-\tilde{x}} \Big)\nonumber\\
&= V(W\cdot \one )\,  V(\mu)\, \psi^{(1)} \, 
U\Big(\tfrac{f(x)-f(\tilde{x})}{x-\tilde{x}} \Big) \nonumber\\
&= W' \, V(\mu) \, \psi^{(1)}\, 
 U\Big(\tfrac{f(x)-f(\tilde{x})}{x-\tilde{x}} \Big) \ ,
\end{align}
where in the last step we used that $V(W\cdot \one )=W'\cdot\one$.
This shows that indeed $\phi_{UV}$ is a valid degree-$1$ morphism between
the fusion functors $U$ and $V$. 
\smallskip

The construction above shows that for any odd homomorphism $\psi$ between $U (I_{W})$ and $V (I_{W})$ there exists a degree-$1$ morphism $\phi_{UV}$ between $U$ and $V$ such that for any matrix factorization $Q$ we have
\begin{equation}
\phi_{UV} (Q) = \psi_{Q}\ .
\end{equation} 
We conclude that in the case of one variable, every
morphism (even or odd) between operator-like defects can be represented by a
morphism between the corresponding fusion functors. As we have seen, one can give explicit
constructions for such morphisms.

\subsection{More variables}\label{sec:morevariables}
The explicit analysis becomes more involved when we consider polynomial rings with $n>1$ variables. When one tries to represent any given even morphism between operator-like defects by a morphism between fusion functors, one has to face the problem that they could arise from a combination of morphisms of degree $0,2,\dotsc ,2\lfloor\frac{n}{2} \rfloor $. Analogously, odd morphisms descend from morphisms of degree $1,3,\dotsc ,2\lfloor\frac{n-1}{2} \rfloor +1$. The only situation -- apart from the case of one variable -- where only one degree contributes occurs when we consider odd morphisms for polynomial rings in two variables. In this case, the analysis is similar to the case of odd morphisms for one variable, and we sketch it in the following.

Let $R=\mathbb{C}[x_{1},x_{2}]$, and let $U,V$ be $(W',W)$-fusion functors for $W,W'\in R$. The identity defect can be represented by the matrix factorization \cite{Khovanov:2004,Carqueville:2009ev}
\begin{equation}
I_{W}= \begin{pmatrix}
0 & 0 & (x_{1}-\tilde{x}_{1}) & (x_{2}-\tilde{x}_{2})\\
0 & 0 & -\frac{W (\tilde{x}_{1},x_{2})-W (\tilde{x}_{1},\tilde{x}_{2})}{x_{2}-\tilde{x}_{2}} & \frac{W (x_{1},x_{2})-W (\tilde{x}_{1},x_{2})}{x_{1}-\tilde{x}_{1}}\\
\frac{W (x_{1},x_{2})-W (\tilde{x}_{1},x_{2})}{x_{1}-\tilde{x}_{1}} & - (x_{2}-\tilde{x}_{2}) & 0 & 0\\
\frac{W (\tilde{x}_{1},x_{2})-W (\tilde{x}_{1},\tilde{x}_{2})}{x_{2}-\tilde{x}_{2}} & ( x_{1}-\tilde{x}_{1}) & 0 & 0 
\end{pmatrix}\ .
\end{equation} 
For a given matrix factorization $Q=\begin{pmatrix}
0 & Q^{(1)}\\
Q^{(0)} & 0 
\end{pmatrix}$ of $W$ one can write down explicit homomorphisms $\lambda_{Q},\lambda_{Q}^{-1}$ \cite{Carqueville:2009ev,Carqueville:2012st} that realize the isomorphism between $Q$ and $I_{W}\boxtimes Q$ as
\begin{align}
\lambda_{Q} & = \begin{pmatrix}
\mu  & 0 & 0 & 0 & 0 & 0 & 0 & 0\\
0 & 0 &0 & 0 & 0 & 0 & \mu  & 0
\end{pmatrix}\\
\lambda_{Q}^{-1} &= \begin{pmatrix}
(\lambda_{Q}^{-1})^{(0)} & 0\\
0 & (\lambda_{Q}^{-1})^{(1)}
\end{pmatrix}
\end{align}
with 
\begin{align}
(\lambda_{Q}^{-1})^{(0)} &= \begin{pmatrix}
\one \\
\frac{Q^{(1)}(\tilde{x}_{1},x_{2})-Q^{(1)} (\tilde{x}_{1},\tilde{x}_{2})}{x_{2}-\tilde{x}_{2}}\frac{Q^{(0)} (x_{1},x_{2})-Q^{(0)} (\tilde{x}_{1},x_{2})}{x_{1}-\tilde{x}_{1}}\\
\frac{Q^{(0)} (x_{1},x_{2})-Q^{(0)} (\tilde{x}_{1},x_{2})}{x_{1}-\tilde{x}_{1}}\\
\frac{Q^{(0)} (\tilde{x}_{1},x_{2})-Q^{(0)} (\tilde{x}_{1},\tilde{x}_{2})}{x_{2}-\tilde{x}_{2}}
\end{pmatrix}\\
(\lambda_{Q}^{-1})^{(1)} &= \begin{pmatrix}
\frac{Q^{(1)} (x_{1},x_{2})-Q^{(1)} (\tilde{x}_{1},x_{2})}{x_{1}-\tilde{x}_{1}}\\
\frac{Q^{(1)} (\tilde{x}_{1},x_{2})-Q^{(1)} (\tilde{x}_{1},\tilde{x}_{2})}{x_{2}-\tilde{x}_{2}}\\
\one \\
\frac{Q^{(0)} (\tilde{x}_{1},x_{2})-Q^{(0)} (\tilde{x}_{1},\tilde{x}_{2})}{x_{2}-\tilde{x}_{2}} \frac{Q^{(1)} (x_{1},x_{2})-Q^{(1)} (\tilde{x}_{1},x_{2})}{x_{1}-\tilde{x}_{1}}
\end{pmatrix}\ .
\end{align}
Here we use a somewhat condensed notation analogously to the one used in Section~\ref{sec:onevariableoddmorphisms}.

Then for an odd morphism 
\begin{equation}
\psi =\begin{pmatrix}
0 & \psi^{(1)}\\
\psi^{(0)} & 0
\end{pmatrix}
\end{equation}
between $U (I_{W})$ and $V (I_{W})$ we find 
\begin{equation}
\psi_{Q} = V (\lambda_{Q}) (\psi \otimes \one) U (\lambda_{Q}^{-1}) = \begin{pmatrix}
0 & \psi_{Q}^{(1)}\\
\psi_{Q}^{(0)} & 0 
\end{pmatrix}
\end{equation}
with
\begin{equation}
\psi_{Q}^{(i)} = \begin{pmatrix}
V (\mu)& 0
\end{pmatrix} \psi^{(1)} \begin{pmatrix}
U\left(\frac{Q^{(i)} (x_{1},x_{2})-Q^{(i)} (\tilde{x}_{1},x_{2})}{x_{1}-\tilde{x}_{1}} \right)\\
U\left(\frac{Q^{(i)} (\tilde{x}_{1},x_{2})-Q^{(i)} (\tilde{x}_{1},\tilde{x}_{2})}{x_{2}-\tilde{x}_{2}} \right)
\end{pmatrix}=\phi_{UV} (Q^{(i)})\ ,
\end{equation}
where the result is used to define a morphism $\phi_{UV}$ between $U$ and $V$ of degree $1$. As in the one-variable case one can show that it is closed and that it satisfies the fusion functor property~\eqref{WlinearityofPsi}. 

We conclude that for two variables, any odd morphism between $U(I_{W})$ and $V(I_{W})$ is induced by a degree-$1$ morphism between $U$ and $V$. 

\section{Summary and Conclusion}

In this article, we have formalized the notion of fusion functors and proposed it as a description of a subclass of interfaces in B-type topologically twisted Landau--Ginzburg models. These functors act on the categories of free finite-rank modules over a polynomial ring. A $(W',W)$-fusion functor directly maps a finite-rank matrix factorization of $W$ seen as a module homomorphism to a finite-rank matrix factorization of $W'$. Fusion functors have been introduced as a tool in~\cite{Behr:2020gqw} to facilitate the explicit computation of fusion of interfaces to boundary conditions. Here, we have shown that there is a larger structure behind them: fusion functors form a category where morphisms are higher-degree generalizations of natural transformations; in the functorial framework they are the analogue of interface fields.

We have introduced a functor $\Pi^{I_{W}}$ from the category of $(W',W)$-fusion functors to the category of $(W',W)$-matrix factorizations (see Definition~\ref{def:PiIW}). For polynomial rings in one variable, we have shown that this functor is full, hence any interface field between $U (I_{W})$ and $V (I_{W})$ has a realization as a morphism between the fusion functors $U$ and $V$. It would be interesting to investigate whether this holds for general polynomial rings. 

Connected to that is the question how to characterize the class of interface factorizations that have a realization as a fusion functor. For matrix factorizations, the fusion of a $(W',W)$-interface factorization via the graded tensor product directly gives rise to a functor from the category $\Hmf_{W}$ of finite-rank matrix factorizations to the category $\HMF_{W'}$ of matrix factorizations of any rank \cite{Khovanov:2004}. It is known that the resulting infinite-rank factorization can always be reduced to a finite-rank factorization \cite{Khovanov:2004,Brunner:2007qu}, and a direct construction \cite{Dyckerhoff:2013} and an explicit algorithm \cite{Carqueville:2011zea} for this reduction are available. It would be interesting to investigate whether these results can shed light on the question for which interface factorizations one can construct a fusion functor, or how one could generalize the notion of a fusion functor to model all interfaces in such a way.

We have given a prescription of horizontal composition which turns the category of defect fusion functors $\HFF_{W,W}$ into a strict monoidal supercategory. More generally, the horizontal composition gives rise to a strict 2-supercategory (in the sense of~\cite{Brundan:2017}) where the objects are the potentials, the 1-morphisms the fusion functors, and the 2-morphisms the morphisms between fusion functors. One would expect that $\Pi^{I_{W}}$ gives rise to a (weak) 2-superfunctor to the corresponding 2-supercategory of matrix factorizations. In the case of matrix factorizations one has more structure because of the existence of adjoints~\cite{Carqueville:2012st}. It would be interesting to investigate whether such a structure is also realized for fusion functors.  

As we have discussed, the functor $\Pi^{I_{W}}$ is not faithful. It would be interesting to work out the additional equivalences on the morphism spaces to obtain a category of fusion functors that allows for a faithful functor to the category of matrix factorizations.

Fusion of matrix factorizations is a basic ingredient in the Khovanov--Rozansky construction of link homology~\cite{Khovanov:2004}. The elementary matrix factorizations and morphisms that enter this framework admit a natural reformulation in terms of fusion functors and morphisms between them, as illustrated in Example~\ref{ex:KR}. At the braid level, one may therefore hope to formulate an analogous construction directly in the fusion-functor language. Implementing braid closure, however, requires an additional operation (on the matrix-factorization side realized by identifying variables) which is not directly built into the present functorial framework. Nevertheless, it would be interesting to explore whether the fusion-functor formulation can simplify parts of the computation by replacing explicit tensor-product calculations of matrix factorizations \cite{Carqueville:2011zea} with functorial fusion.

In a rational CFT there are finitely many elementary rational topological defects, and they close under fusion, forming a fusion (tensor) category. In Landau--Ginzburg models whose infrared fixed points are rational, it is natural to ask how this fusion-category structure is reflected in the matrix-factorization description of B-type topological defects. In the present paper we take as a working assumption that the relevant B-type topological defects of the Landau--Ginzburg model are realized within the matrix-factorization framework. Under this assumption, one is led to expect that the category of defect factorizations contains a tensor subcategory whose objects correspond to the rational defects in the infrared, and whose tensor product agrees with their fusion.  In this sense, and under the assumptions spelled out above, the presence of a non-trivial fusion subcategory is a natural structural feature one expects to appear in rational Landau--Ginzburg models.

On the level of fusion functors, the corresponding substructure is more subtle: there are typically infinitely many non-isomorphic functors that map under $\Pi^{I_W}$ to the same defect. Thus the functorial side does not itself form a fusion category, but its image under $\Pi^{I_W}$ may. Nevertheless, fusion functors can be useful in practice, since fusion is implemented by composition
by construction.

A first step towards identifying such structures in a non-trivial example has been taken in~\cite{Behr:2020gqw}. There, a proposal was formulated that associates to every rational B-type defect in an $SU(3)/U(2)$ Kazama--Suzuki model a fusion functor (and hence a matrix factorization), and for a subclass the fusion was computed and found to agree with the fusion of rational defects. It would be interesting to extend this analysis and to test, in further examples, whether the resulting tensor subcategory of matrix factorizations reproduces the tensor category of rational B-type defects, as is known in the
minimal model case~\cite{Davydov:2014oya}.

\section*{Acknowledgements}

I thank Nils Carqueville for many interesting discussions and useful suggestions. I am grateful to Nicolas Behr for a long-lasting collaboration that has provided the seeds out of which this work has grown. I also thank Peter Allmer and Ingo Runkel for helpful discussions. 

\appendix

\section{Some results on morphisms between functors}

\subsection{A projection to the space of normalized morphisms}\label{sec:projection}
The morphism spaces $\Hom(U,V)$ have the structure of a cosimplicial object with coface maps $\delta^i:\Hom_n(U,V)\to \Hom_{n+1}(U,V)$ given by
\begin{align}
    \delta^i\phi_n (f_{n+1},\dotsc,f_1) = \left\{\begin{array}{ll}
     V(f_{n+1})\,\phi_n ( f_n,\dotsc,f_1)    & i=0  \\
     \phi_n (f_{n+1},\dotsc,f_{n+2-i}f_{n+1-i},\dotsc,f_1)    & 1\leq i\leq n\\
     \phi_n ( f_{n+1},\dotsc,f_2)\,U(f_1) & i=n+1\,,
    \end{array} \right.
\end{align}
and codegeneracy maps $\sigma^j:\Hom_n(U,V)\to \Hom_{n-1}(U,V)$ given by
\begin{equation}\label{codegeneracy}
    \sigma^j\phi_n (f_{n-1},\dotsc,f_1) = \phi_n ( f_{n-1},\dotsc,f_{n-j},\one,f_{n-j-1},\dotsc,f_1)\ .
\end{equation}
In such a cosimplicial space, the cohomology is isomorphic to the cohomology of the subspace spanned by normalized elements, i.e.\ those morphisms $\phi_n$ that vanish when any of the entries is the identity (see, e.g., \cite[chapter 8]{Weibel}\cite[chapter 1.6]{Loday}\cite{Eilenberg:1947}). For completeness we include a proof here that defines a projection to the space of normalized morphisms which becomes an identity on the cohomology.

For any $n\in \mathbb{N}$ we define on $\Hom_{n} (U,V)$ for $s=1,\dotsc ,n$ the following operators,
\begin{align}
R^{(s)}:\Hom_{n} (U,V) & \longrightarrow \Hom_{n-1} (U,V)\\
R^{(s)}\phi (f_{n-1},\dotsc ,f_{1}) & = (-1)^{n-s+1} \phi (f_{n-1},\dotsc ,f_{s},\one,f_{s-1},\dotsc ,f_{1})\ .
\end{align}
Up to a sign and relabelling, these are just the codegeneracy maps defined in~\eqref{codegeneracy}.
We also define the following subspaces of $\Hom_{n} (U,V)$,
\begin{equation}
\Hom_{n}^{(s)} (U,V) = \{\phi \in  \Hom_{n} (U,V)\,\vert\, \phi (f_{n},\dotsc ,f_{1}) = 0 \ \text{if}\ f_{i}=\one  \ \text{for any}\ i\leq s\}\ .
\end{equation}
Obviously, we have
\begin{equation}
\Hom_{n}^{(s)} (U,V) \subset \Hom_{n}^{(t)} (U,V) \ \text{for}\ s\leq t \ .
\end{equation}
For $s=n$,
\begin{equation}
\Hom_{n}^{\mathrm{norm}} (U,V) = \Hom_{n}^{(n)} (U,V)
\end{equation}
we recover the space of normalized morphisms (as in~\eqref{defreducedspace}).

\begin{lem}\label{lem:kerRdplusdR}
\begin{equation}
\Hom_{n}^{(s)} (U,V)\subset \mathrm{Ker}\, (R^{(s)}d+dR^{(s)}) \ .
\end{equation}
\end{lem}
\begin{proof}
We determine the action of $R^{(s)}d+dR^{(s)}$. First, we evaluate $R^{(s)}d$:
\begin{align}
R^{(s)}d\phi (f_{n},\dotsc ,f_{1}) &= (-1)^{n-s} d\phi (f_{n},\dotsc ,f_{s},\one ,f_{s-1},\dotsc ,f_{1})\\
&= (-1)^{n-s} \big( V (f_{n})\,\phi (f_{n-1},\dotsc ,\underset{\mathclap{\color{blue}(s)}}{\one },\dotsc ,f_{1})\nonumber\\
&\qquad  +\dotsb \nonumber\\
&\qquad + (-1)^{n-s} \phi (f_{n},\dotsc ,f_{s+1}\,f_{s},\underset{\mathclap{\color{blue}(s)}}{\one},f_{s-1},\dotsc ,f_{1})\nonumber\\
&\qquad + (-1)^{n-s+1}\phi (f_{n},\dotsc ,f_{s},\underset{\mathclap{\color{blue}(s-1)}}{\one},f_{s-1}\,f_{s-2},\dotsc ,f_{1})\nonumber\\
&\qquad + \dotsb \nonumber\\
&\qquad + (-1)^{n} \phi (f_{n},\dotsc ,f_{s},\underset{\mathclap{\color{blue}(s-1)}}{\one},f_{s-1},\dotsc ,f_{2})U (f_{1})\big)\ .\label{Rd}
\end{align}
We indicate the position of the entry $\one$ (counted from the right) by a blue subscript.
For $dR^{(s)}$, we obtain
\begin{align}
dR^{(s)}\phi (f_{n},\dotsc ,f_{1}) &= (-1)^{n-s+1}\big( V (f_{n})\,\phi (f_{n-1},\dotsc ,\underset{\mathclap{\color{blue}(s)}}{\one},\dotsc ,f_{1})\nonumber\\
&\qquad + \dotsb \nonumber\\
&\qquad + (-1)^{n-s} \phi (f_{n},\dotsc ,f_{s+1}\,f_{s},\underset{\mathclap{\color{blue}(s)}}{\one},f_{s-1},\dotsc ,f_{1})\nonumber\\
&\qquad + (-1)^{n-s+1}\phi (f_{n},\dotsc ,f_{s+1},\underset{\mathclap{\color{blue}(s)}}{\one},f_{s}\,f_{s-1},\dotsc ,f_{1})\nonumber\\
&\qquad + \dotsb \nonumber\\
&\qquad + (-1)^{n} \phi (f_{n},\dotsc ,f_{s+1},\underset{\mathclap{\color{blue}(s)}}{\one},f_{s},\dotsc ,f_{2})U (f_{1})\big)\ .
\label{dR}
\end{align}
The first $n-s+1$ terms in the expressions~\eqref{Rd} for $R^{(s)}d$ and~\eqref{dR} for $dR^{(s)}$ coincide up to an overall minus sign. Hence, we find
\begin{align}
(R^{(s)}d+dR^{(s)})\phi (f_{n},\dotsc ,f_{1}) &= \phi (f_{n},\dotsc ,f_{s+1},\underset{\mathclap{\color{blue}(s)}}{\one},f_{s}\,f_{s-1},\dotsc ,f_{1})\nonumber\\
&\qquad + \dotsb \nonumber\\
&\qquad + (-1)^{s-1} \phi (f_{n},\dotsc ,f_{s+1},\underset{\mathclap{\color{blue}(s)}}{\one},f_{s},\dotsc ,f_{2})U (f_{1})\nonumber\\
&\qquad - \phi (f_{n},\dotsc ,f_{s},\underset{\mathclap{\color{blue}(s-1)}}{\one},f_{s-1}\,f_{s-2},\dotsc ,f_{1})\nonumber\\
&\qquad + \dotsb \nonumber\\
&\qquad + (-1)^{s} \phi (f_{n},\dotsc ,f_{s},\underset{\mathclap{\color{blue}(s-1)}}{\one},f_{s-1},\dotsc ,f_{2})U (f_{1})\ .\label{RdplusdR}
\end{align}
If $\phi \in \Hom_{n}^{(s)} (U,V)$, this expression vanishes because in all terms there is the identity $\one$ in the position $s$ or $s-1$.
\end{proof}

\begin{lem}\label{lem:imagePs}
\begin{equation}
\mathrm{Im}\, (\one -R^{(s)}d-dR^{(s)})\big\vert_{\Hom^{(s-1)}_{n} (U,V)} \subset \Hom^{(s)}_{n} (U,V)\ .
\end{equation}
\end{lem}
\begin{proof}
Let $\phi \in \Hom^{(s-1)}_{n} (U,V)$. When we apply $\one -R^{(s)}d-dR^{(s)}$ to $\phi$, we can use the result~\eqref{RdplusdR} from the previous lemma, where the last $s+1$ terms vanish because $\phi \in \Hom^{(s-1)}_{n} (U,V)$. Hence, we find
\begin{align}
&(\one -R^{(s)}d-dR^{(s)})\phi (f_{n},\dotsc ,f_{1})\nonumber\\
& \quad = \phi (f_{n},\dotsc ,f_{1}) - \phi (f_{n},\dotsc ,f_{s+1},\underset{\mathclap{\color{blue}(s)}}{\one},f_{s}\,f_{s-1},\dotsc ,f_{1})\nonumber\\
&\qquad\qquad \qquad \qquad + \dotsb \nonumber\\
&\qquad\qquad \qquad \qquad  + (-1)^{s} \phi (f_{n},\dotsc ,f_{s+1},\underset{\mathclap{\color{blue}(s)}}{\one},f_{s},\dotsc ,f_{2})U (f_{1})\ .
\end{align}
If any of the $f_{i}=\one$ for $i<s$, then
\begin{align}
&(\one -R^{(s)}d-dR^{(s)})\phi (f_{n},\dotsc ,\underset{\mathclap{\color{blue}(i<s)}}{\one},\dotsc ,f_{1})\nonumber\\
& \quad = \phi (f_{n},\dotsc ,\underset{\mathclap{\color{blue}(i<s)}}{\one},\dotsc ,f_{1}) - \phi (f_{n},\dotsc ,\underset{\mathclap{\color{blue}(s)}}{\one},f_{s}\,f_{s-1},\dotsc,\underset{\mathclap{\color{blue}(i<s)}}{\one} ,\dotsc ,f_{1})\nonumber\\
&\qquad\qquad \qquad \qquad\qquad \ + \dotsb \nonumber\\
&\qquad\qquad \qquad \qquad \qquad \  + (-1)^{s-i} \phi (f_{n},\dotsc ,\underset{\mathclap{\color{blue}(s)}}{\one},f_{s},\dotsc ,f_{i+1},f_{i-1},\dotsc ,f_{1})\nonumber\\
&\qquad\qquad \qquad \qquad\qquad \   + (-1)^{s-i+1} \phi (f_{n},\dotsc ,\underset{\mathclap{\color{blue}(s)}}{\one},f_{s},\dotsc ,f_{i+1},f_{i-1},\dotsc ,f_{1})\nonumber\\
&\qquad\qquad \qquad \qquad\qquad \  + \dotsb \nonumber\\
&\qquad\qquad \qquad \qquad\qquad \   + (-1)^{s} \phi (f_{n},\dotsc ,\underset{\mathclap{\color{blue}(s)}}{\one},\dotsc ,\underset{\mathclap{\color{blue}(i<s)}}{\one} ,\dotsc ,f_{2})U (f_{1})=0\ ,
\end{align}
because the two middle terms cancel, and all other terms are zero with a $\one$ appearing in an entry at position $i<s$.
If $f_{s}=\one$, we get
\begin{align}
&(\one -R^{(s)}d-dR^{(s)})\phi (f_{n},\dotsc ,\underset{\mathclap{\color{blue}(s)}}{\one},\dotsc ,f_{1})\nonumber\\
& \quad = \phi (f_{n},\dotsc ,\underset{\mathclap{\color{blue}(s)}}{\one},\dotsc ,f_{1}) - \phi (f_{n},\dotsc ,\underset{\mathclap{\color{blue}(s)}}{\one},\dotsc ,f_{1})\nonumber\\
&\qquad \qquad \qquad \qquad \qquad \ + \phi (f_{n},\dotsc ,\underset{\mathclap{\color{blue}(s)}}{\one}\,,\,\underset{\mathclap{\color{blue}(s-1)}}{\one},f_{s-1}\,f_{s-2},\dotsc ,f_{1})\nonumber\\
&\qquad\qquad \qquad \qquad\qquad \ + \dotsb \nonumber\\
&\qquad\qquad \qquad \qquad\qquad \   + (-1)^{s} \phi (f_{n},\dotsc ,\underset{\mathclap{\color{blue}(s)}}{\one}\,,\,\underset{\mathclap{\color{blue}(s-1)}}{\one} ,\dotsc ,f_{2})U (f_{1})=0\ .
\end{align}
Here, the first two terms cancel, and the others are zero because $\phi \in \Hom^{(s-1)}_{n} (U,V)$.
\end{proof}
Now we are ready to define the projection operator to the space of normalized morphisms:
\begin{Def}
On $\Hom_{n} (U,V)$ we define
\begin{equation}\label{defP}
P:= (\one-R^{(n)}d-dR^{(n)})\dotsb (\one-R^{(1)}d-dR^{(1)}) \ .
\end{equation}
\end{Def}
\begin{lem}
The operator $P$ satisfies the following properties:
\begin{align}
\mathrm{Im}\, P &\subset \Hom^{\mathrm{norm}}(U,V) \ ,\label{propP1}\\
P\big\vert_{\Hom^{\mathrm{norm}}(U,V)} & = \one \big\vert_{\Hom^{\mathrm{norm}}(U,V)}\ ,\label{propP2}\\
P\,P &= P\ .
\end{align}
\end{lem}
\begin{proof}
Applying Lemma~\ref{lem:imagePs} multiple times, we see that
\begin{equation}
\mathrm{Im}\,(\one-R^{(s)}d-dR^{(s)})\dotsb (\one-R^{(1)}d-dR^{(1)}) \subset \Hom^{(s)}(U,V) \ .
\end{equation}
When we set $s=n$, this becomes~\eqref{propP1} (since $\Hom_{n}^{\mathrm{norm}} (U,V) = \Hom_{n}^{(n)} (U,V)$).

Because of Lemma~\ref{lem:kerRdplusdR} we have for any $s=1,\dotsc ,n$
\begin{equation}
\Hom_{n}^{\mathrm{norm}} (U,V) \subset \Hom_{n}^{(s)} (U,V) \subset \mathrm{Ker}\, (R^{(s)}d+dR^{(s)})\ .
\end{equation}
Therefore
\begin{equation}
(\one-R^{(s)}d-dR^{(s)})\big\vert_{\Hom^{\mathrm{norm}}(U,V)} = \one\big\vert_{\Hom^{\mathrm{norm}}(U,V)}
\end{equation}
and equation~\eqref{propP2} follows. Finally, we have
\begin{equation}
P\,P = P \big\vert_{\Hom^{\mathrm{norm}}(U,V)}\,P = P\ .
\end{equation}
\end{proof}
\begin{lem}
Starting from $\hat{R}^{(1)}:=R^{(1)}$ we define $R:=\hat{R}^{(n)}$ through the recursion
\begin{equation}
\hat{R}^{(s+1)}=\hat{R}^{(s)}+R^{(s+1)}-R^{(s+1)}\hat{R}^{(s)}d-R^{(s+1)}d\hat{R}^{(s)} \quad \text{for}\ s=1,\dotsc ,n-1.
\end{equation}
Then we have
\begin{equation}
P = \one -R d-dR \ .
\end{equation}
\end{lem}
\begin{proof}
We prove
\begin{equation}
(\one-R^{(s)}d-dR^{(s)})\dotsb (\one-R^{(1)}d-dR^{(1)}) = \one -\hat{R}^{(s)}d-d\hat{R}^{(s)}
\end{equation}
by induction over $s$. It is obviously true for $s=1$. Suppose the statement holds for a given $s\geq 1$. Then
\begin{align}
&(\one-R^{(s+1)}d-dR^{(s+1)})\dotsb (\one-R^{(1)}d-dR^{(1)})\nonumber\\
&\qquad = (\one-R^{(s+1)}d-dR^{(s+1)}) (\one -\hat{R}^{(s)}d-d\hat{R}^{(s)})\\
&\qquad = \one -R^{(s+1)}d -dR^{(s+1)}-\hat{R}^{(s)}d - d \hat{R}^{(s)}\nonumber\\
&\qquad \quad + R^{(s+1)}d\hat{R}^{(s)}d + dR^{(s+1)}\hat{R}^{(s)}d + d R^{(s+1)}d\hat{R}^{(s)}\\
&\qquad = \one - \big(R^{(s+1)}+\hat{R}^{(s)}-R^{(s+1)}\hat{R}^{(s)}d - R^{(s+1)}d\hat{R}^{(s)} \big)d\nonumber\\
&\qquad \quad -d \big(R^{(s+1)}+\hat{R}^{(s)}-R^{(s+1)}\hat{R}^{(s)}d - R^{(s+1)}d\hat{R}^{(s)} \big)\\
&\qquad = \one -\hat{R}^{(s+1)}d -d\hat{R}^{(s+1)}\ .
\end{align}
\end{proof}

\subsection{Morphisms are determined by their action on polynomials}\label{sec:app-proof-of-lemma}
Here, we provide a proof of Lemma~\ref{lem:phideterminedbyRendomorphisms} that states that the cohomology class of a morphism is determined already by the action of the morphism on polynomials.
\begin{proof}[Proof of Lemma~\ref{lem:phideterminedbyRendomorphisms}]
\refstepcounter{dummy}\label{proof:phideterminedbyRendomorphisms}
Let $\phi =\phi^{(n)}\in\Hom_{n}(U,V)$ be a closed morphism of degree
$n$ between the functors $U$ and $V$ that vanishes on chains of endomorphisms of $R$. To prove the proposition, we fix for each free module $M$ of rank $m_{M}$ an isomorphism to $R^{m_{M}}$. We denote the map to the $k^{\text{th}}$ factor by $\pi_{M}^{k}:M\to R$ and the inclusion by $\iota_{M}^{k}:R\to M$ such that
\begin{equation}
\sum_{k=1}^{m_{M}} \iota_{M}^{k}\,\pi_{M}^{k} = \one_{M} \quad \text{and for all}\ k=1,\dotsc ,m_{M}:\ \pi_{M}^{k}\,\iota_{M}^{k} = \one_{R} \ .
\end{equation}
We now construct a morphism $\phi^{(n-1)}$ between $U$ and $V$ such that $\phi=d\phi^{(n-1)}$. For a sequence of homomorphisms $M_{n}\xleftarrow{f_{n-1}} M_{n-1}
\xleftarrow{f_{n-2}}\dotsb \xleftarrow{f_{1}} M_{1}$ we introduce the notations 
\begin{align}
\iota_{m}^{k}& :=\iota_{M_{m}}^{k} \ ,\quad  \pi_{m}^{k}:=\pi_{M_{m}}^{k}\\
p_{m}^{(k)}&:= \pi_{m+1}^{k_{m+1}}\, f_{m}\,
\iota_{m}^{k_{m}}:R\to R \ .
\end{align}
We then define $\phi^{(n-1)}$ by
\begin{equation}
\phi^{(n-1)} = \sum_{j=1}^{n} (-1)^{n-j+1}\, \phi^{(n-1)}_{[j]}
\end{equation}
with
\begin{align}
\phi^{(n-1)}_{[1]}(f_{n-1}\listdots f_{1}) &= V (\iota_{n}^{k_{n}}) \,  \phi 
(p_{n-1}^{(k)}\listdots p_{2}^{(k)},\pi_{2}^{k_{2}}\,  f_{1},\iota_{1}^{k_{1}})
\,  U (\pi_{1}^{k_{1}}) \nonumber\\
\phi^{(n-1)}_{[2]}(f_{n-1}\listdots f_{1}) &= V (\iota_{n}^{k_{n}})\,  \phi 
(p_{n-1}^{(k)}\listdots p_{3}^{(k)},\pi_{3}^{k_{3}}\,  f_{2},\iota_{2}^{k_{2}},\pi_{2}^{k_{2}}\,  f_{1})\nonumber\\
\vdots &\nonumber \\
\phi^{(n-1)}_{[n-2]}  (f_{n-1}\listdots f_{1}) &= V(\iota_{n}^{k_{n}})\,  \phi 
(p_{n-1}^{(k)},\pi_{n-1}^{k_{n-1}}f_{n-2},\iota_{n-2}^{k_{n-2}},\pi_{n-2}^{k_{n-2}} 
f_{n-3}\listdots f_{1})\nonumber\\
\phi^{(n-1)}_{[n-1]} (f_{n-1}\listdots f_{1}) &= V(\iota_{n}^{k_{n}}) \,  \phi  
(\pi_{n}^{k_{n}}f_{n-1}\iota_{n-1}^{k_{n-1}},\pi_{n-1}^{k_{n-1}}f_{n-2},f_{n-3}\listdots f_{1})\nonumber\\
\phi^{(n-1)}_{[n]} (f_{n-1}\listdots f_{1}) &= \phi 
(\iota_{n}^{k_{n}},\pi_{n}^{k_{n}}\,  f_{n-1},f_{n-2}\listdots f_{1} )\ ,
\end{align}
where a summation over all indices $k_{i}$ is always understood.

Consider first $d\phi^{(n-1)}_{[n]}$,
\begin{align}
&d\phi^{(n-1)}_{[n]} (f_{n}\listdots f_{1})\nonumber \\
&\quad = 
V (f_{n})\,  \phi^{(n-1)}_{[n]}(f_{n-1}\listdots f_{1}) - \phi^{(n-1)}_{[n]} (f_{n}\,  f_{n-1},f_{n-2}\listdots f_{1})
\nonumber\\
&\qquad +\dotsb + (-1)^{n-1} \phi^{(n-1)}_{[n]} (f_{n}\listdots f_{3},f_{2}\, 
f_{1}) + (-1)^{n} \phi^{(n-1)}_{[n]}
(f_{n}\listdots f_{2})\,  U (f_{1})\\
&\quad = V (f_{n})\,  \phi  (\iota_{n}^{k_{n}},\pi_{n}^{k_{n}}\, 
f_{n-1},f_{n-2}\listdots f_{1}) -\phi  (\iota_{n+1}^{k_{n+1}},\pi_{n+1}^{k_{n+1}}  f_{n}
f_{n-1},f_{n-2}\listdots f_{1})\nonumber\\
&\qquad +\dotsb + (-1)^{n} \phi  (\iota_{n+1}^{k_{n+1}},\pi_{n+1}^{k_{n+1}}\, 
f_{n},f_{n-1}\listdots f_{2})\,  U (f_{1})\\
&\quad =V (f_{n})\,  \phi  (\iota_{n}^{k_{n}},\pi_{n}^{k_{n}}\, 
f_{n-1},f_{n-2}\listdots f_{1})\nonumber\\
&\qquad +\Big\{ V (\iota_{n+1}^{k_{n+1}})\,  \phi  (\pi_{n+1}^{k_{n+1}}\, 
f_{n},f_{n-1}\listdots f_{1}) -\phi  (f_{n}\listdots f_{1})
\Big\} \ ,
\label{firstphi_nminusone}
\end{align}
where we used in the last step that 
\begin{equation}
d\phi  (\iota_{n+1}^{k_{n+1}},\pi_{n+1}^{k_{n+1}}\,  f_{n},f_{n-1},\dotsc ,f_{1}) = 0
\end{equation}
due to the closure of $\phi $.

Now consider $\phi^{(n-1)}_{[j]}$ for $1<j<n$. We have
\begin{align}
&d\phi^{(n-1)}_{[j]} (f_{n}\listdots f_{1}) = V (\iota_{n+1}^{k_{n+1}})\, \nonumber\\
 &\bigg\{
\Big[ V (p_{n}^{(k)})\,  \phi  (p_{n-1}^{(k)}\listdots p_{j+1}^{(k)},\pi_{j+1}^{k_{j+1}}\, 
f_{j},\iota_{j}^{k_{j}},\pi_{j}^{k_{j}}\,  f_{j-1},f_{j-2}\listdots f_{1})\nonumber\\
&\quad - \phi  (p_{n}^{(k)}  p_{n-1}^{(k)},p_{n-2}^{(k)}\listdots p_{j+1}^{(k)},\pi_{j+1}^{k_{j+1}}\, 
f_{j},\iota_{j}^{k_{j}},\pi_{j}^{k_{j}}\,  f_{j-1},f_{j-2}\listdots f_{1}) \nonumber\\
&\quad +\dotsb \nonumber\\
&\quad + (-1)^{n-j-1}  \phi  (p_{n}^{(k)}\listdots p_{j+2}^{(k)}\, 
p_{j+1}^{(k)},\pi_{j+1}^{k_{j+1}}\,  f_{j},\iota_{j}^{k_{j}},\pi_{j}^{k_{j}}\, 
f_{j-1},f_{j-2}\listdots f_{1})\nonumber\\
&\quad + (-1)^{n-j} \phi  (p_{n}^{(k)}\listdots p_{j+2}^{(k)},\pi_{j+2}^{k_{j+2}}\, 
f_{j+1}\,  f_{j},\iota_{j}^{k_{j}},\pi_{j}^{k_{j}}\,  f_{j-1},f_{j-2}\listdots f_{1})
\Big]\nonumber\\[2mm]
&\quad + \Big[ (-1)^{n-j+1} \phi  (p_{n}^{(k)}\listdots p_{j+2}^{(k)},\pi_{j+2}^{k_{j+2}}\,  f_{j+1},\iota_{j+1}^{k_{j+1}},\pi_{j+1}^{k_{j+1}}\,  f_{j}\, 
f_{j-1},f_{j-2}\listdots f_{1}) \nonumber\\
&\quad + (-1)^{n-j+2} \phi  (p_{n}^{(k)}\listdots p_{j+2}^{(k)},\pi_{j+2}^{k_{j+2}}\,  f_{j+1},\iota_{j+1}^{k_{j+1}},\pi_{j+1}^{k_{j+1}}\, 
f_{j},f_{j-1}\,  f_{j-2}\listdots f_{1}) \nonumber\\
&\quad +\dotsb \nonumber\\
&\quad + (-1)^{n-1} \phi  (p_{n}^{(k)}\listdots p_{j+2}^{(k)},\pi_{j+2}^{k_{j+2}}\,  f_{j+1},\iota_{j+1}^{k_{j+1}},\pi_{j+1}^{k_{j+1}}\, 
f_{j},f_{j-1}\listdots f_{3},f_{2}\, 
f_{1}) \nonumber\\
&\quad + (-1)^{n} \phi  (p_{n}^{(k)}\listdots p_{j+2}^{(k)},\pi_{j+2}^{k_{j+2}}\,  f_{j+1},\iota_{j+1}^{k_{j+1}},\pi_{j+1}^{k_{j+1}}\, 
f_{j},f_{j-1}\listdots f_{2}) \,  U (f_{1}) 
\Big]\bigg\}\nonumber \\
&= V (\iota_{n+1}^{k_{n+1}})\,  \Big\{ A (n,j) + B (n,j)\Big\}\ ,
\end{align}
where $A(n,j)$ combines the first $n-j+1$ summands (first brackets),
and $B (n,j)$ the remaining $j$ summands (second brackets).

From the closure of $\phi $ we know that for $1<j<n-1$ we have
\begin{align}
0 &= d\phi  (p_{n}^{(k)}\listdots p_{j+2}^{(k)},\pi^{k_{j+2}}_{j+2}\,  f_{j+1},\iota_{j+1}^{k_{j+1}},\pi^{k_{j+1}}_{j+1}\, 
f_{j},f_{j-1}\listdots f_{1})\\
&=V(p_{n}^{(k)})  \phi  (p_{n-1}^{(k)}\listdots p_{j+2}^{(k)},\pi^{k_{j+2}}_{j+2}\,  f_{j+1},\iota_{j+1}^{k_{j+1}},\pi^{k_{j+1}}_{j+1}\, 
f_{j},f_{j-1}\listdots f_{1}) \nonumber\\
&\quad - \phi  (p_{n}^{(k)}  p_{n-1}^{(k)}\listdots p_{j+2}^{(k)},\pi^{k_{j+2}}_{j+2}\,  f_{j+1},\iota_{j+1}^{k_{j+1}},\pi^{k_{j+1}}_{j+1}\, 
f_{j},f_{j-1}\listdots f_{1})\nonumber\\
&\quad  +\dotsb \nonumber\\
&\quad + (-1)^{n-j-1} \phi  (p_{n}^{(k)}\listdots p_{j+3}^{(k)},\pi_{j+3}^{k_{j+3}}\, 
f_{j+2}\,  f_{j+1},\iota_{j+1}^{k_{j+1}},\pi^{k_{j+1}}_{j+1}\,  f_{j},f_{j-1}\listdots f_{1})\nonumber\\[2mm]
&\quad + (-1)^{n-j} \phi  (p_{n}^{(k)}\listdots p_{j+1}^{(k)},\pi^{k_{j+1}}_{j+1}\, 
f_{j},f_{j-1}\listdots f_{1})\nonumber\\
&\quad + (-1)^{n-j+1} \phi  (p_{n}^{(k)}\listdots p_{j+2}^{(k)},\pi^{k_{j+2}}_{j+2}\, 
f_{j+1},f_{j}\listdots f_{1})\nonumber\\[2mm]
&\quad + (-1)^{n-j+2} \phi  (p_{n}^{(k)}\listdots p_{j+2}^{(k)},\pi^{k_{j+2}}_{j+2}\, 
f_{j+1},\iota_{j+1}^{k_{j+1}},\pi^{k_{j+1}}_{j+1}\,  f_{j}\, 
f_{j-1},f_{j-2}\listdots f_{1})\nonumber\\
&\quad +\dotsb \nonumber\\
&\quad + (-1)^{n} \phi  (p_{n}^{(k)}\listdots p_{j+2}^{(k)},\pi^{k_{j+2}}_{j+2}\, 
f_{j+1},\iota_{j+1}^{k_{j+1}},\pi^{k_{j+1}}_{j+1}\, 
f_{j},f_{j-1}\listdots f_{3},f_{2}\, 
f_{1}) \nonumber\\
&\quad + (-1)^{n+1} \phi  (p_{n}^{(k)}\listdots p_{j+2}^{(k)},\pi^{k_{j+2}}_{j+2}\, 
f_{j+1},\iota_{j+1}^{k_{j+1}},\pi^{k_{j+1}}_{j+1}\, 
f_{j},f_{j-1}\listdots f_{2})\,  U(f_{1})\nonumber\\
&= A (n,j+1) - B (n,j)\nonumber\\
&\quad + (-1)^{n-j} \phi  (p_{n}^{(k)}\listdots p_{j+1}^{(k)},\pi^{k_{j+1}}_{j+1}\, 
f_{j},f_{j-1}\listdots f_{1})\nonumber\\
&\quad + (-1)^{n-j+1} \phi  (p_{n}^{(k)}\listdots p_{j+2}^{(k)},\pi^{k_{j+2}}_{j+2}\, 
f_{j+1},f_{j}\listdots f_{1}) \ .
\end{align}
This also holds for $j=1$ and $j=n-1$ if we set
\begin{align}
B (n,1) &= (-1)^{n} \phi  (p_{n}^{(k)},\dotsc ,p_{3}^{(k)},\pi_{3}^{k_{3}}\,  f_{2},\iota_{2}^{k_{2}})\,  U (\pi_{2}^{k_{2}}\,  f_{1})\\
A (n,n) &= V (\pi_{n+1}^{k_{n+1}}\,  f_{n}) \,  \phi  (\iota_{n}^{k_{n}},\pi_{n}^{k_{n}}\, 
f_{n-1},f_{n-2},\dotsc ,f_{1}) \ .
\end{align}
Therefore the contribution of the $\phi^{(n-1)}_{[j]}$ for $1<j<n$ to
$d\phi^{(n-1)}$ is
\begin{align}
&\sum_{j=2}^{n-1} (-1)^{n-j+1} d\phi^{(n-1)}_{[j]} (f_{n},\dotsc
,f_{1})\nonumber\\
&\quad = \sum_{j=2}^{n-1} (-1)^{n-j+1}V (\iota_{n+1}^{k_{n+1}})\,  \Big\{A (n,j)+B (n,j)\Big\}\nonumber\\
&\quad = V (\iota_{n+1}^{k_{n+1}})\,  \bigg\{A (n,n) - (-1)^{n}B (n,1)
-\sum_{j=1}^{n-1} (-1)^{n-j+1}\Big\{A (n,j+1)-B (n,j)\Big\}\bigg\}\nonumber\\
&\quad = V (\iota_{n+1}^{k_{n+1}})\,  \bigg\{A (n,n) - (-1)^{n}B (n,1)\nonumber\\
&\quad \qquad \qquad \qquad +\sum_{j=1}^{n-1} \Big\{
\phi  (p_{n}^{(k)},\dotsc ,p_{j+2}^{(k)},\pi_{j+2}^{k_{j+2}}\, 
f_{j+1},f_{j},\dotsc ,f_{1}) \nonumber\\
&\quad \qquad \qquad \quad \qquad \qquad -\phi  (p_{n}^{(k)},\dotsc ,p_{j+1}^{(k)},\pi_{j+1}^{k_{j+1}}\, 
f_{j},f_{j-1},\dotsc ,f_{1})
\Big\}\bigg\}\nonumber\\
&\quad = V (f_{n}) \,  \phi  (\iota_{n}^{k_{n}},\pi_{n}^{k_{n}}\, 
f_{n-1},f_{n-2},\dotsc ,f_{1})\nonumber\\
&\quad \quad - V (\iota_{n+1}^{k_{n+1}})\,  \phi  (p_{n}^{(k)},\dotsc ,p_{3}^{(k)},\pi_{3}^{k_{3}}\,  f_{2},\iota_{2}^{k_{2}})
\,  U (\pi_{2}^{k_{2}}\,  f_{1})\nonumber\\
&\quad \quad +V (\iota_{n+1}^{k_{n+1}})\,  \phi  (\pi_{n+1}^{k_{n+1}}\, 
f_{n},f_{n-1},\dotsc ,f_{1}) \nonumber\\
&\quad \quad -V (\iota_{n+1}^{k_{n+1}})\,  \phi  (p_{n}^{(k)},\dotsc ,p_{2}^{(k)},\pi_{2}^{k_{2}}\,  f_{1})
\ . 
\end{align}
Combining this result with~\eqref{firstphi_nminusone} we obtain
\begin{align}
&\sum_{j=2}^{n} (-1)^{n-j+1} d\phi^{(n-1)}_{[j]} (f_{n},\dotsc
,f_{1}) \nonumber\\
&\qquad
= \phi  (f_{n},\dotsc ,f_{1}) - V (\iota_{n+1}^{k_{n+1}})\,  \phi  
(p_{n}^{(k)},\dotsc ,p_{3}^{(k)},\pi_{3}^{k_{3}}\,  f_{2},\iota_{2}^{k_{2}})
\,  U (\pi_{2}^{k_{2}}\,  f_{1})\nonumber\\
&\qquad \quad -V (\iota_{n+1}^{k_{n+1}})\,  \phi  
(p_{n}^{(k)},\dotsc ,p_{2}^{(k)},\pi_{2}^{k_{2}}\,  f_{1})\ .
\end{align}
Let us now consider the contribution of $\phi^{(n-1)}_{[1]}$. We have
\begin{align}
&d\phi^{(n-1)}_{[1]} (f_{n},\dotsc ,f_{1}) \nonumber\\
&= V (f_{n}  \iota_{n}^{k_{n}})  \phi  (p_{n-1}^{(k)},\dotsc ,p_{2}^{(k)},\pi_{2}^{k_{2}} 
f_{1},\iota_{1}^{k_{1}})   U (\pi_{1}^{k_{1}})\nonumber\\
&\quad -V (\iota_{n+1}^{k_{n+1}})\,  \phi  (p_{n}^{(k)}  p_{n-1}^{(k)} ,p_{n-2}^{(k)},\dotsc ,p_{2}^{(k)},\pi_{2}^{k_{2}}\, 
f_{1},\iota_{1}^{k_{1}}) \,  U(\pi_{1}^{k_{1}})\nonumber\\
&\quad +\dotsb \nonumber\\
&\quad + (-1)^{n-2} V (\iota_{n+1}^{k_{n+1}})\,  \phi  (p_{n}^{(k)},\dotsc ,
p_{3}^{(k)}  p_{2}^{(k)},\pi_{2}^{k_{2}}\, 
f_{1},\iota_{1}^{k_{1}}) \,  U (\pi_{1}^{k_{1}})\nonumber\\
&\quad + (-1)^{n-1} V (\iota_{n+1}^{k_{n+1}})\,  \phi  (p_{n}^{(k)},\dotsc ,p_{3}^{(k)},\pi_{3}^{k_{3}}\, 
f_{2}\,  f_{1},\iota_{1}^{k_{1}})\,  U (\pi_{1}^{k_{1}})\nonumber\\
&\quad + (-1)^{n} V (\iota_{n+1}^{k_{n+1}})\,  \phi  (p_{n}^{(k)},\dotsc ,p_{3}^{(k)},\pi_{3}^{k_{3}}\, 
f_{2},\iota_{2}^{k_{2}}) \,  U (\pi_{2}^{k_{2}}\,  f_{1})\ .
\label{dphifirst}
\end{align}
From the closure of $\phi $ we conclude that
\begin{align}
0&=d\phi  (p_{n}^{(k)},\dotsc ,p_{2}^{(k)},\pi_{2}^{k_{2}}\,  f_{1},\iota_{1}^{k_{1}}) \nonumber\\
&= V (p_{n}^{(k)})  \phi  (p_{n-1}^{(k)},\dotsc ,p_{2}^{(k)},\pi_{2}^{k_{2}}\, 
f_{1},\iota_{1}^{k_{1}}) \nonumber\\
&\quad - \phi  (p_{n}^{(k)}  p_{n-1}^{(k)},p_{n-2}^{(k)},\dotsc ,p_{2}^{(k)},\pi_{2}^{k_{2}}\, 
f_{1},\iota_{1}^{k_{1}})\nonumber\\
&\quad +\dotsb \nonumber\\
&\quad + (-1)^{n-2} \phi  (p_{n}^{(k)},\dotsc ,
p_{3}^{(k)} 
p_{2}^{(k)},\pi_{2}^{k_{2}}\,  f_{1},\iota_{1}^{k_{1}})\nonumber\\
&\quad + (-1)^{n-1}  \phi  (p_{n}^{(k)},\dotsc ,p_{3}^{(k)},\pi_{3}^{k_{3}}\,  f_{2}\, 
f_{1},\iota_{1}^{k_{1}})\nonumber\\
&\quad + (-1)^{n} \phi  (p_{n}^{(k)},\dotsc ,p_{1}^{(k)})\nonumber\\
&\quad + (-1)^{n+1} \phi  (p_{n}^{(k)},\dotsc ,p_{2}^{(k)},\pi_{2}^{k_{2}}\,  f_{1}) \,  U(\iota_{1}^{k_{1}})
\ .
\end{align}
When acting with $V (\iota_{n+1}^{k_{n+1}})$ from the left and with $ U(\pi_{1}^{k_{1}})$ from the right, the first $n$ terms (up to and including the term starting with $(-1)^{n-1}$) coincide with the first $n$ terms appearing in the expression for $d\phi^{(n-1)}_{[1]}$ in~\eqref{dphifirst},
and hence
\begin{align}
&d\phi^{(n-1)}_{[1]} (f_{n},\dotsc ,f_{1}) \nonumber\\
&= - (-1)^{n}V (\iota_{n+1}^{k_{n+1}})\,  \phi  (p_{n}^{(k)},\dotsc ,p_{1}^{(k)})  U
(\pi_{1}^{k_{1}})\nonumber\\
&\quad - (-1)^{n+1} V (\iota_{n+1}^{k_{n+1}})\,  \phi  (p_{n}^{(k)},\dotsc
,p_{2}^{(k)},\pi_{2}^{k_{2}}\,  f_{1})\nonumber\\
&\quad + (-1)^{n} V (\iota_{n+1}^{k_{n+1}}\,  \phi  (p_{n}^{(k)},\dotsc ,p_{3}^{(k)},\pi_{3}^{k_{3}}\, 
f_{2},\iota_{2}^{k_{2}}) \,  U (\pi_{2}^{k_{2}}\,  f_{1})\ .
\end{align}
Therefore we get the desired result for $d\phi^{(n-1)}$ 
\begin{align}
d\phi^{(n-1)} (f_{n},\dotsc ,f_{1}) &=\sum_{j=1}^{n} (-1)^{n-j+1}d\phi^{(n-1)}_{[j]} (f_{n},\dotsc ,f_{1})\\
&= \phi  (f_{n},\dotsc ,f_{1}) \nonumber\\
&\quad -V (\iota_{n+1}^{k_{n+1}})\,  \phi  (p_{n}^{(k)},\dotsc ,p_{1}^{(k)})\,  U
(\pi_{1}^{k_{1}}) \\
&= \phi  (f_{n},\dotsc ,f_{1}) \ ,
\end{align}
where we have used that $\phi$ vanishes for chains of endomorphisms $p_{i}^{(k)}$ of $R$.
\end{proof}

\subsection{Proof of Proposition~\ref{prop:naturaltransformationPiQ}}\label{sec:app-proof-naturality}
Here we give a proof of Proposition~\ref{prop:naturaltransformationPiQ} which states that a morphism $\varphi$ between matrix factorizations $Q_1$ and $Q_2$ induces a natural transformation between the functors $\Pi^{Q_1}$ and $\Pi^{Q_2}$. In particular, we prove that (see~\eqref{naturalityxi})
\begin{equation}
V (\varphi)\, \phi^{(n)}_{UV} (Q_{1},\dotsc ,Q_{1}) = \phi^{(n)}_{UV} (Q_{2},\dotsc ,Q_{2})\,U (\varphi) + (\delta_{U (Q_{1}),V (Q_{2})}\text{-exact terms}) 
\end{equation}
for all closed morphisms $\phi^{(n)}_{UV}$ between $(W',W)$-fusion functors $U$ and $V$.

\begin{proof}
Let $1\leq s\leq n-1$. Then
\begin{align}
0 &= d\phi^{(n)}_{UV} (\underbrace{Q_{2} \listdots   Q_{2}}_{s},\varphi,\underbrace{Q_{1}\listdots  Q_{1}}_{n-s})\\
&= V (Q_{2})\phi^{(n)}_{UV} (\underbrace{Q_{2} \listdots Q_{2}}_{s-1},\varphi,\underbrace{Q_{1}\listdots Q_{1}}_{n-s}) + \dotsb\nonumber\\
&\quad
+ (-1)^{s}\phi^{(n)}_{UV} (\underbrace{Q_{2} \listdots Q_{2}}_{s-1},Q_{2}\,\varphi,\underbrace{Q_{1}\listdots Q_{1}}_{n-s})\nonumber\\
&\quad + (-1)^{s+1} \phi^{(n)}_{UV} (\underbrace{Q_{2} \listdots Q_{2}}_{s},\varphi\,Q_{1},\underbrace{Q_{1}\listdots Q_{1}}_{n-s-1}) + \dotsb \nonumber\\
&\quad + (-1)^{n+1} \phi^{(n)}_{UV} (\underbrace{Q_{2}\listdots Q_{2}}_{s},\varphi,\underbrace{Q_{1}\listdots Q_{1}}_{n-s-1})U (Q_{1})\ .\label{dphiexprwithdots}
\end{align}
The summands appearing as dots in~\eqref{dphiexprwithdots} are zero, for example
\begin{equation}
\phi^{(n)}_{UV} (Q_{2}^{2},\dotsc) = \phi^{(n)}_{UV} (W\cdot \one ,\dotsc)=W' \phi^{(n)}_{UV} (\one ,\dotsc ) = 0\ .
\end{equation}
In the next to last line we can express $\varphi Q_{1}=Q_{2}\varphi$ because $\varphi$ is closed. Next we add the $\delta_{U (Q_{1}),V (Q_{2})}$-exact term
\begin{multline}
\delta_{U (Q_{1}),V (Q_{2})}\phi^{(n)} (\underbrace{Q_{2}\listdots Q_{2}}_{s},\varphi,\underbrace{Q_{1}\listdots Q_{1}}_{n-s-1})\\
 = V (Q_{2})  \phi^{(n)} (\underbrace{Q_{2}\listdots Q_{2}}_{s},\varphi,\underbrace{Q_{1}\listdots Q_{1}}_{n-s-1}) + (-1)^{n} \phi^{(n)} (\underbrace{Q_{2}\listdots Q_{2}}_{s},\varphi,\underbrace{Q_{1}\listdots Q_{1}}_{n-s-1}) U (Q_{1})
\end{multline}
and find 
\begin{align}
0&=  V (Q_{2})\phi^{(n)}_{UV} (\underbrace{Q_{2}\listdots Q_{2}}_{s-1},\varphi,\underbrace{Q_{1}\listdots Q_{1}}_{n-s})
+ (-1)^{s}\phi^{(n)}_{UV} (\underbrace{Q_{2}\listdots Q_{2}}_{s-1},Q_{2}\,\varphi,\underbrace{Q_{1}\listdots Q_{1}}_{n-s})\nonumber\\
&\quad + V (Q_{2}) \phi^{(n)}_{UV} (\underbrace{Q_{2}\listdots Q_{2}}_{s},\varphi,\underbrace{Q_{1}\listdots Q_{1}}_{n-s-1})\nonumber\\
&\quad + (-1)^{s+1} \phi^{(n)}_{UV} (\underbrace{Q_{2}\listdots Q_{2}}_{s},Q_{2}\,\varphi,\underbrace{Q_{1}\listdots Q_{1}}_{n-s-1}) + (\delta_{U (Q_{1}),V (Q_{2})}\text{-exact terms}) \ .
\end{align}
For $s=0$ or $s=n$ we obtain by similar arguments
\begin{align}
0 &= V (\varphi) \phi^{(n)}_{UV} (\underbrace{Q_{1}\listdots Q_{1}}_{n}) \nonumber\\
&\quad + V (Q_{2}) \phi^{(n)}_{UV} (\varphi, \underbrace{Q_{1}\listdots Q_{1}}_{n-1}) - \phi^{(n)}_{UV} (Q_{2}\varphi, \underbrace{Q_{1}\listdots Q_{1}}_{n-1})
+ (\delta_{U (Q_{1}),V (Q_{2})}\text{-exact terms}) \ ,\\
0 &= V (Q_{2}) \phi^{(n)}_{UV} (\underbrace{Q_{2}\listdots Q_{2}}_{n-1},\varphi)+ (-1)^{n}\phi^{(n)}_{UV} (\underbrace{Q_{2}\listdots Q_{2}}_{n-1},Q_{2}\,\varphi) \nonumber\\
&\quad + (-1)^{n+1} \phi^{(n)}_{UV} (\underbrace{Q_{2}\listdots Q_{2}}_{n})U (\varphi) + (\delta_{U (Q_{1}),V (Q_{2})}\text{-exact terms}) \ .
\end{align}
Adding all these terms with alternating sign we obtain
\begin{align}
0 &= \sum_{s=0}^{n} (-1)^{s} d\phi^{(n)}_{UV} (\underbrace{Q_{2} \listdots   Q_{2}}_{s},\varphi,\underbrace{Q_{1}\listdots  Q_{1}}_{n-s})\\
&= V (\varphi) \phi^{(n)}_{UV} (\underbrace{Q_{1}\listdots Q_{1}}_{n}) - \phi^{(n)}_{UV} (\underbrace{Q_{2}\listdots Q_{2}}_{n})U (\varphi) + (\delta_{U (Q_{1}),V (Q_{2})}\text{-exact terms})\ ,
\end{align}
which concludes the proof.
\end{proof}

\subsection{Extension to $R\otimes \tilde{R}$-modules}\label{sec:app-bimodules}
Given a $(W',W)$-fusion functor $U$ that acts on free finite-rank $R$-modules, we can extend it to a functor $\hat{U}\in \Fun (R'\otimes \tilde{R},R\otimes \tilde{R})$ from the category of (free, finite-rank) $(R\otimes_{\mathbb{C}}\tilde{R})$-modules to the category of (free, finite-rank) $(R'\otimes_{\mathbb{C}}\tilde{R})$-modules. For every such $(R\otimes_{\mathbb{C}}\tilde{R})$-module $N$ we can find an isomorphism 
\begin{equation}
\xi_{N}:N\to M\otimes_{\mathbb{C}}\tilde{M}
\end{equation}
to a tensor product of an $R$-module $M$ and an $\tilde{R}$-module $\tilde{M}$. We then set
\begin{equation}
\hat{U} (N):= U (M)\otimes_{\mathbb{C}}\tilde{M}\ ,
\end{equation}
where within this appendix we distinguish the extension of $U$ by a hat; in the main text we will denote it for better readability by $U$.

Similarly, for a bimodule homomorphism $\varphi:N_{1}\to N_{2}$, we can write
\begin{equation}\label{fgmap}
\xi_{N_{2}}\, \varphi\, \xi_{N_{1}}^{-1}= \sum_{i}f_{i}\otimes g_{i}: M_{1}\otimes_{\mathbb{C}}\tilde{M}_{1}\to M_{2}\otimes_{\mathbb{C}}\tilde{M}_{2}\ ,
\end{equation}
where $f_{i}:M_{1}\to M_{2}$ are $R$-module homomorphisms, and $g_{i}:\tilde{M}_{1}\to \tilde{M}_{2}$ are $\tilde{R}$-module homomorphisms. We then set
\begin{equation}
\hat{U} (\varphi) := \sum_{i} U (f_{i})\otimes g_{i}\ .
\end{equation}
Note that this definition does not depend on the precise choice of the functions $f_{i},g_{i}$ in~\eqref{fgmap}. On the other hand, the extension depends on the choice of the isomorphisms $\xi_{N}$, but different choices lead to isomorphic functors:
\begin{prop}
Let $\hat{U}^{(a)}$ and $\hat{U}^{(b)}$ be two extensions of $U$ in the sense above corresponding to the two choices $\xi_{N}^{(i)}$ ($i=a,b$) of isomorphisms. Then there is a natural isomorphism between $\hat{U}^{(a)}$ and $\hat{U}^{(b)}$.
\end{prop}
\begin{proof}
For every bimodule $N$ we can write
\begin{equation}
\xi_{N}^{(b)} ( \xi_{N}^{(a)})^{-1} = \sum_{i}\eta_{N,i}\otimes \kappa_{N,i}:M^{(a)}\otimes_{\mathbb{C}}\tilde{M}^{(a)} \to M^{(b)}\otimes_{\mathbb{C}}\tilde{M}^{(b)}\ .
\end{equation}
We then define a natural transformation $\phi$ between $\hat{U}^{(a)}$ and $\hat{U}^{(b)}$ by setting
\begin{equation}
\phi (N) = \sum_{i}U (\eta_{N,i})\otimes \kappa_{N,i}\ .
\end{equation}
To show that this indeed defines a natural transformation, we have to check that for any homomorphism $\varphi:N_{1}\to N_{2}$ the following diagram commutes:
\begin{center}
\begin{tikzpicture}[thick, scale=1.2, rotate=0]

\coordinate (A) at (-2,1);
\coordinate (B) at ( 2,1);
\coordinate (C) at (-2,-1);
\coordinate (D) at ( 2,-1);

\draw[->] (-1.3,1)  -- (1.3,1) ;
\draw[->] (-1.3,-1) -- (1.3,-1) ;
\draw[->] (-2,0.7)  -- (-2,-0.7)  ;
\draw[->] (2,0.7)   -- (2,-0.7) ;

\draw (A) node[color=black] {$\hat{U}^{(a)} (N_{1})$};
\draw (B)+(0.1,0) node[color=black] {$\hat{U}^{(b)} (N_{1})$};
\draw (C) node[color=black] {$\hat{U}^{(a)} (N_{2})$};
\draw (D)+(0.1,0) node[color=black] {$\hat{U}^{(b)} (N_{2})$};

\draw (0,1.4) node[color=black] {$\phi (N_{1}) $};
\draw (0,-1.3) node[color=black] {$\phi (N_{2})  $};
\draw (-2.7,0) node[color=black] {$\hat{U}^{(a)} (\varphi)  $};
\draw ( 2.7,0) node[color=black] {$ \hat{U}^{(b)} (\varphi)$};
\end{tikzpicture}
\end{center}
Writing out the definitions of $\hat{U}^{(a)}$ and $\hat{U}^{(b)}$, the diagram reads
\begin{center}
\begin{tikzpicture}[thick, scale=1.2, rotate=0]

\coordinate (A) at (-2.8,1);
\coordinate (B) at ( 2.8,1);
\coordinate (C) at (-2.8,-1);
\coordinate (D) at ( 2.8,-1);

\draw[->] (-1.6,1)  -- (1.6,1) ;
\draw[->] (-1.6,-1) -- (1.6,-1) ;
\draw[->] (-2.8,0.7)  -- (-2.8,-0.7)  ;
\draw[->] (2.8,0.7)   -- (2.8,-0.7) ;

\draw (-3,1) node[color=black] {$U ( M_{1}^{(a)})\otimes_{\mathbb{C}}\tilde{M}_{1}^{(a)}$};
\draw (3,1) node[color=black] {$U (M_{1}^{(b)})\otimes_{\mathbb{C}}\tilde{M}_{1}^{(b)}$};
\draw (-3,-1) node[color=black] {$U (M_{2}^{(a)})\otimes_{\mathbb{C}}\tilde{M}_{2}^{(a)}$};
\draw (3,-1) node[color=black] {$U ( M_{2}^{(b)})\otimes_{\mathbb{C}}\tilde{M}_{2}^{(b)}$};

\draw (0,1.4) node[color=black] {$\sum U (\eta_{N_{1},j})\otimes \kappa_{N_{1}}$};
\draw (0,-1.3) node[color=black] {$\sum U ( \eta_{N_{2},j})\otimes \kappa_{N_{2}}$};
\draw (-4.1,0) node[color=black] {$\sum U ( f_{i}^{(a)})\otimes g_{i}^{(a)} $};
\draw ( 4.1,0) node[color=black] {$\sum U ( f_{i}^{(b)})\otimes g_{i}^{(b)}$};
\end{tikzpicture}
\end{center}
The commutativity of this diagram is a consequence of the commutativity of the following diagram (which commutes by definition of the maps involved):
\begin{center}
\begin{tikzpicture}[thick, scale=1.2, rotate=0]

\coordinate (A) at (-2.5,1);
\coordinate (B) at ( 2.5,1);
\coordinate (C) at (-2.5,-1);
\coordinate (D) at ( 2.5,-1);

\draw[->] (-1.3,1)  -- (1.3,1) ;
\draw[->] (-1.3,-1) -- (1.3,-1) ;
\draw[->] (-2.5,0.7)  -- (-2.5,-0.7)  ;
\draw[->] (2.5,0.7)   -- (2.5,-0.7) ;

\draw (A) node[color=black] {$M_{1}^{(a)}\otimes_{\mathbb{C}}\tilde{M}_{1}^{(a)}$};
\draw (B)+(0.1,0) node[color=black] {$M_{1}^{(b)}\otimes_{\mathbb{C}}\tilde{M}_{1}^{(b)}$};
\draw (C) node[color=black] {$M_{2}^{(a)}\otimes_{\mathbb{C}}\tilde{M}_{2}^{(a)}$};
\draw (D)+(0.1,0) node[color=black] {$M_{2}^{(b)}\otimes_{\mathbb{C}}\tilde{M}_{2}^{(b)}$};

\draw (0,1.4) node[color=black] {$\sum \eta_{N_{1},j}\otimes \kappa_{N_{1}}$};
\draw (0,-1.3) node[color=black] {$\sum \eta_{N_{2},j}\otimes \kappa_{N_{2}}$};
\draw (-3.6,0) node[color=black] {$\sum f_{i}^{(a)}\otimes g_{i}^{(a)} $};
\draw ( 3.6,0) node[color=black] {$\sum f_{i}^{(b)}\otimes g_{i}^{(b)}$};
\end{tikzpicture}
\end{center}
\end{proof}
Similarly we can extend morphisms $\phi_{UV}$ between fusion functors $U,V$ to morphisms between the extensions $\hat{U}$ and $\hat{V}$:
\begin{Def}
Let $U,V$ be $(W',W)$-fusion functors, and let $\hat{U},\hat{V}$ be extensions as above to functors from the category of $(R\otimes_{\mathbb{C}}\tilde{R})$-modules to the category of $(R'\otimes_{\mathbb{C}}\tilde{R})$-modules. Let $\phi$ be a morphism of degree $n$ from $U$ to $V$. Then we define a morphism $\hat{\phi}$ between $\hat{U}$ and $\hat{V}$ as
\begin{equation}
\hat{\phi} \big( N_{n+1}  \xleftarrow{\varphi_{n}} N_{n}
\xleftarrow{\varphi_{n-1}} \dotsb \xleftarrow{\varphi_{1}} N_{1} \big) = \sum_{i_{1},\dotsc ,i_{n}} \phi (f_{n,i_{n}},\dotsc ,f_{1,i_{1}}) \otimes (g_{n,i_{n}}\dotsb g_{1,i_{1}})\ ,
\end{equation}
where 
\begin{equation}
\xi_{N_{j+1}}\varphi_{j}\xi_{N_{j}}^{-1} = \sum_{i_{j}} f_{j,i_{j}}\otimes g_{j,i_{j}}\ ,
\end{equation}
analogously to~\eqref{fgmap}.
\end{Def}
\begin{prop}
The extension is compatible with the differential,
\begin{equation}
\delta_{\hat{U}\hat{V}}\hat{\phi} = \widehat{\delta_{UV}\phi}\ .
\end{equation}
\end{prop}
\begin{proof}
Let $\phi$ be a morphism of degree $n$ between fusion functors $U$ and $V$, and let $\hat{\phi}$ be the corresponding morphism to the extensions $\hat{U}$ and $\hat{V}$. Using the notation as above, we have
\begin{align}
&\big(\delta_{\hat{U}\hat{V}}\hat{\phi} \big) (\varphi_{n+1},\dotsc ,\varphi_{1})\nonumber\\
 &\qquad = \hat{V} (\varphi_{n+1})\hat{\phi} (\varphi_{n},\dotsc ,\varphi_{1}) + \dotsb + (-1)^{n+1}\hat{\phi} (\varphi_{n+1},\dotsc ,\varphi_{2})\hat{U} (\varphi_{1})\nonumber\\
&\qquad = \sum_{i_{1},\dotsc ,i_{n+1}}\Big( \big(V (f_{n+1,i_{n+1}})\otimes g_{n+1,i_{n+1}} \big)\big(\phi (f_{n,i_{n}},\dotsc ,f_{1,i_{1}})\otimes g_{n,i_{n}}\dotsb g_{1,i_{1}} \big) + \dotsb \nonumber\\
&\qquad \qquad \qquad   + (-1)^{n+1} \big(\phi (f_{n+1,i_{n+1}},\dotsc ,f_{2,i_{2}})\otimes g_{n+1,i_{n+1}}\dotsb g_{2,i_{2}} \big)\big(U (f_{1,i_{1}})\otimes g_{1,i_{1}} \big)\Big)\nonumber\\
&\qquad = \sum_{i_{1},\dotsc ,i_{n+1}}\big(\delta_{UV}\phi \big) (f_{n+1,i_{n+1}},\dotsc ,f_{1,i_{1}}) \otimes g_{n+1,i_{n+1}}\dotsb g_{1,i_{1}}\nonumber\\
&\qquad =\widehat{\delta_{UV}\phi} ( \varphi_{n+1},\dotsc ,\varphi_{1})\ .
\end{align}
\end{proof}

\subsection{Hochschild cohomology of polynomial rings}\label{sec:app-proof-of-prop}
For completeness, we provide here an elementary proof of a statement on the structure of Hochschild cohomology of polynomial rings.
\begin{proof}[Proof of proposition~\ref{prop:highercohomologytrivial}]
\refstepcounter{dummy}\label{proof:highercohomologytrivial}
The Hochschild cohomology for an $R$-bimodule $M$ is given by 
\begin{equation}
HH^{n} (R,M) = \mathrm{Ext}^{n}_{R^{e}}(R,M) \ ,
\end{equation}
with the enveloping algebra\footnote{$R^{\mathrm{op}}$ denotes the opposite ring which is isomorphic to $R$ in our case.} $R^{e}\cong R\otimes_{\mathbb{C}} R^{\mathrm{op}}\cong \mathbb{C}[y_{1},\dotsc ,y_{d},z_{1},\dotsc ,z_{d}]$. The $\mathrm{Ext}$-groups can be computed from a Koszul resolution of $R$. The kernel $I$ of the multiplication map $\mu :R_{e}\to R$ is generated by the regular sequence $x=(y_{1}-z_{1},\dotsc ,y_{d}-z_{d})$,
\begin{equation}
R = \frac{R_{e}}{(y_{1}-z_{1},\dotsc ,y_{d}-z_{d})R_{e}} = \frac{R_{e}}{I}\ .
\end{equation} 
The regular sequence gives rise to the following exact sequence defining a free resolution of $R$ as an $R_{e}$-module,
\begin{equation}
\underbrace{0 \to \Lambda^{d} (R_{e})^{d} \to \dotsb \to \Lambda^{2}  (R_{e})^{d}\to R_{e}^{d} \xrightarrow{(y_{1}-z_{1},\dotsc ,y_{d}-z_{d})} R_{e}}_{K (x)} \to R \to 0\ ,
\end{equation}
where $\Lambda^{i}V$ denotes the completely antisymmetric subspace of $V^{\otimes i}$.  
On the Koszul sequence $K (x)$ we apply the functor $\Hom_{R^e} (-,M)$ and obtain the chain complex
\begin{equation}\label{chaincomplex}
0 \leftarrow \Hom_{R^e} (\Lambda^{d} (R_{e})^{d},M) \leftarrow \dotsb \leftarrow \Hom_{R^e} ((R_{e})^{d},M) \leftarrow \Hom_{R^e} (R_{e},M)\ .
\end{equation}
Its cohomology determines the $\mathrm{Ext}$-groups, $\mathrm{Ext}^{\bullet} (R,M)= H^{\bullet} (K (x),M)$. As the complex becomes trivial on the left, we find that 
\begin{equation}
\mathrm{Ext}_{R_{e}}^{n} (R,M) = 0 \ \text{for}\ n>d \ .
\end{equation}
\end{proof}
In the simplest case when $M=R$ we have $\Hom_{R^e} ( R_{e},R)\cong R$ as $R_{e}$-modules, and all maps in~\eqref{chaincomplex} become trivial. The cohomology is thus
\begin{equation}\label{HHcohomologyrings}
HH^{n} (R,R)\cong \mathsf{\Lambda}^{n}R^{d}\ .
\end{equation}

\section{Some explicit results on fusion functors}

\subsection{Explicit check of the induced morphism in the one variable case}\label{sec:explicitcheck}
We have seen in Section~\ref{sec:evenmorphisms} that for any given closed homomorphism $\varphi$ from $U (I_{W})$ to $V (I_{W})$, we can construct a morphism $\phi_{UV}$ (defined in~\eqref{defofPhi}) between the fusion functors $U$ and $V$. In this appendix, we verify that $\phi_{UV} (M_{I_{W}})$ coincides with $\varphi$ up to an exact term.

We want to evaluate
\begin{equation}\label{phiUVonMIW}
\phi_{UV} (M_{I_{W}}) = \begin{pmatrix}
\big(V (\mu)\otimes \one\big) (\varphi^{(0)}\otimes \one)U (\iota_{R\otimes_\bC R}) & 0\\
0 & \big(V (\mu)\otimes \one  \big) (\varphi^{(0)}\otimes \one)U (\iota_{R\otimes_\bC R}) 
\end{pmatrix}
\end{equation}
explicitly. We first write $\varphi^{(0)}$ as a matrix whose entries are polynomials in $x,\tilde{x}$, and expand it in powers of $\tilde{x}$,
\begin{equation}
\varphi^{(0)} (x,\tilde{x}) = \tilde{x}^{k}\varphi^{(0)}_{k} (x)\ ,
\end{equation}
a sum over $k$ is understood. Then $\phi_{UV} (M_{I_{W}})$ can be represented as
\begin{equation}
\tilde{\varphi}  = \begin{pmatrix}
 V (x^{k})\varphi^{(0)}_{k} (x) & 0\\
0 & V (x^{k})\varphi^{(0)}_{k} (x) 
\end{pmatrix}\ .
\end{equation}
We now show that this coincides with the original homomorphism $\varphi$ up to an exact term. 
We define an odd homomorphism from $U (M_{I_{W}})$ to $V (M_{I_{W}})$
\begin{equation}
\psi (x,\tilde{x}) = \begin{pmatrix}
0 & 0\\
\psi^{(0)} (x,\tilde{x})  & 0
\end{pmatrix}
\end{equation}
with 
\begin{align}
\psi^{(0)} (x,\tilde{x})&=V\Big(\frac{x^{k}-\tilde{x}^{k}}{x-\tilde{x}} \Big) \varphi_{k}^{(0)} (x)\\
&=\sum_{\ell \leq k-1}V (x)^{k-1-\ell}\,\tilde{x}^{\ell}\, \varphi_{k}^{(0)} (x)\ .
\end{align}
Acting with the differential on $\psi$, we obtain
\begin{align}
\delta_{U (I_{W}),V (I_{W})} \psi &= V (I_{W})\psi +\psi U (I_{W})\\
&=\begin{pmatrix}
V (x-\tilde{x})\psi^{(0)} (x,\tilde{x}) & 0\\
0 & \psi^{(0)} (x,\tilde{x}) U (x-\tilde{x})
\end{pmatrix}\ .
\end{align}
Evaluating the upper left entry we find
\begin{equation}
V (x-\tilde{x})\psi^{(0)} (x,\tilde{x}) = V (x^{k})\varphi^{(0)}_{k} (x)-\tilde{x}^{k}\varphi^{(0)}_{k} (x) \ .
\end{equation}
To evaluate the lower right entry we first write down explicitly the closure condition~\eqref{varphiclosedexplicit} for $\varphi$,
\begin{equation}
\tilde{x}^{k}\varphi^{(0)}_{k} (x)\big(U (x)-\tilde{x} \big) = \big(V (x)-\tilde{x} \big)\tilde{x}^{k}\varphi^{(1)}_{k} (x)\ ,
\end{equation}
where we also expanded $\varphi^{(1)}$ in powers of $\tilde{x}$. By comparing powers of $\tilde{x}$ we conclude
\begin{equation}
\varphi^{(0)}_{k}U (x)- \varphi^{(0)}_{k-1} (x) = V (x) \varphi^{(1)}_{k} (x) -\varphi_{k-1}^{(1)} (x)\ .
\end{equation}
Now we evaluate
\begin{align}
&\psi^{(0)} (x,\tilde{x}) U (x-\tilde{x}) \nonumber\\
&\qquad =
 \sum_{0\leq \ell \leq k-1} \big( V (x)^{k-1-\ell}\tilde{x}^{\ell} \varphi_{k}^{(0)} (x) U (x) - V (x)^{k-1-\ell}\tilde{x}^{\ell+1} \varphi_{k}^{(0)} (x)\big)\\
&\qquad =\sum_{0\leq \ell \leq k-1} V (x)^{k-1-\ell}\tilde{x}^{\ell} \big(\varphi_{k}^{(0)} (x) U (x)-\varphi^{(0)}_{k-1} (x)\big) + V (x)^{k-1}\varphi^{(0)}_{k-1} (x)\\
&\qquad = \sum_{0\leq \ell \leq k-1} V (x)^{k-1-\ell}\tilde{x}^{\ell} \big(V (x)\varphi_{k}^{(1)} (x)-\varphi^{(1)}_{k-1} (x)\big) + V (x)^{k-1}\varphi^{(0)}_{k-1} (x)\\
&\qquad =- \tilde{x}^{k-1}\varphi^{(1)}_{k-1} (x) + V (x)^{k-1}\varphi^{(0)}_{k-1} (x) \ .
\end{align}
In total we find
\begin{equation}
\delta_{U (I_{W}),V (I_{W})} \psi  = \tilde{\varphi}- \varphi\ ,
\end{equation}
and we conclude that $\phi_{UV} (M_{I_{W}})$ differs from $\varphi$ by an exact term.

\subsection{Fusion functors describing isomorphic interfaces}\label{sec:nonisoFF}
We give an example of two non-isomorphic fusion functors $U,V$ that describe isomorphic interfaces, $U (I_{W})\cong V (I_{W})$. Consider $R=\mathbb{C}[x]$ and $W=x^{2}$. Take $U=\id$ and 
\begin{equation}
V=\begin{pmatrix}
\id  & 0 & 0\\
\frac{1}{x} (G-\id)& G & 0\\
0 & 0 & \id 
\end{pmatrix}\ ,
\end{equation}
where $G=\omega^{*}$ is the extension of scalars for the automorphism $\omega :x\mapsto -x$ of $R$. $U$ and $V$ are non-isomorphic because there cannot be an invertible homomorphism between $U (R)=R$ and $V (R)=R^{3}$.

The factorization for the identity defect can be represented as 
\begin{equation}
I_{W} = \begin{pmatrix}
0 & x-\tilde{x}\\
x+\tilde{x} & 0
\end{pmatrix}\ .
\end{equation}
Acting on the upper right corner, $I_{W}^{1}=x-\tilde{x}$, with $V$ yields
\begin{equation}
V (I_{W}^{1}) = \begin{pmatrix}
x-\tilde{x} & 0 & 0\\
-2 & -x-\tilde{x} & 0\\
0 & 0 & x-\tilde{x}
\end{pmatrix} \longrightarrow \begin{pmatrix}
1 & 0 & 0\\
0 & x^{2}-\tilde{x}^{2} & 0\\
0 & 0 & x-\tilde{x}
\end{pmatrix}
\end{equation}
where we have used elementary row and column transformations to simplify the resulting matrix. Similarly, we obtain
\begin{equation}
V (I_{W}^{0}) =V (x+\tilde{x}) \longrightarrow \begin{pmatrix}
x^{2}-\tilde{x}^{2} & 0 & 0\\
0 & 1 & 0\\
0 & 0 & x+\tilde{x}
\end{pmatrix}\ .
\end{equation}
The factorization $V (I_{W})$ differs from $U (I_{W})=I_{W}$ only by the addition of trivial factorizations of $x^{2}-\tilde{x}^{2}$, and hence $V (I_{W})$ and $U (I_{W})$ are isomorphic although $U$ and $V$ are non-isomorphic fusion functors.

\bibliographystyle{mystyle5}
\bibliography{references}

\end{document}